\documentclass[final,12pt]{colt2021} %
\usepackage{times}

\usepackage[utf8]{inputenc} %
\usepackage[T1]{fontenc}    %
\usepackage{hyperref}       %
\usepackage{url}            %
\usepackage{booktabs}       %
\usepackage{amsfonts}       %
\usepackage{nicefrac}       %
\usepackage{microtype}      %
\usepackage{color}
\usepackage{empheq}
\usepackage{mathtools}
\usepackage{mdframed}
\usepackage{hyperref}
\usepackage{cleveref}
\usepackage{stmaryrd}
\usepackage{enumerate}

\usepackage{color-edits}
\addauthor{vs}{red}
\addauthor{gl}{blue}

\usepackage{float}
\usepackage{ifthen}
\usepackage{amsbsy}
\usepackage{setspace}
\usepackage{comment}
\usepackage{mathtools}
\usepackage{xcolor}
\let\Pr\undefined

\DeclareMathOperator{\Pr}{Pr}

\DeclareMathOperator*{\argmin}{arg\,min} %
             
\DeclareMathOperator*{\arginf}{arg\,inf}

\def\ddefloop#1{\ifx\ddefloop#1\else\ddef{#1}\expandafter\ddefloop\fi}
\def\ddef#1{\expandafter\def\csname bb#1\endcsname{\ensuremath{\mathbb{#1}}}}
\ddefloop ABCDEFGHIJKLMNOPQRSTUVWXYZ\ddefloop
\def\ddefloop#1{\ifx\ddefloop#1\else\ddef{#1}\expandafter\ddefloop\fi}
\def\ddef#1{\expandafter\def\csname b#1\endcsname{\ensuremath{\mathbf{#1}}}}
\ddefloop ABCDEFGHIJKLMNOPQRSTUVWXYZ\ddefloop
\def\ddef#1{\expandafter\def\csname c#1\endcsname{\ensuremath{\mathcal{#1}}}}
\ddefloop ABCDEFGHIJKLMNOPQRSTUVWXYZ\ddefloop
\def\ddef#1{\expandafter\def\csname h#1\endcsname{\ensuremath{\widehat{#1}}}}
\ddefloop ABCDEFGHIJKLMNOPQRSTUVWXYZ\ddefloop
\def\ddef#1{\expandafter\def\csname hc#1\endcsname{\ensuremath{\widehat{\mathcal{#1}}}}}
\ddefloop ABCDEFGHIJKLMNOPQRSTUVWXYZ\ddefloop
\def\ddef#1{\expandafter\def\csname t#1\endcsname{\ensuremath{\widetilde{#1}}}}
\ddefloop ABCDEFGHIJKLMNOPQRSTUVWXYZ\ddefloop
\def\ddef#1{\expandafter\def\csname tc#1\endcsname{\ensuremath{\widetilde{\mathcal{#1}}}}}
\ddefloop ABCDEFGHIJKLMNOPQRSTUVWXYZ\ddefloop

\usepackage{prettyref}

\usepackage{nicefrac}
\usepackage{algorithm}
\usepackage{algorithmic}
\usepackage{autonum}

\usepackage{hyperref}

\newcommand{\kibitz}[2]{\ifnum\Comments=1{\color{#1}{#2}}\fi}

\newcommand{\E}{\mathbb{E}}

\newcommand{\R}{\mathbb{R}}

\renewcommand{\Pr}{\ensuremath{\mathrm{Pr}}}
\newcommand{\loss}{\ensuremath{{\cal L}}}

\newcommand{\indep}{\perp \!\!\! \perp}
\newcommand{\grad}{\nabla} %

\newcommand{\mcD}{{\mathcal D}}

\newcommand{\mcX}{{\mathcal X}}

\newcommand{\mcF}{{\mathcal F}}

\newcommand{\mcT}{{\mathcal T}}

\newcommand{\mcV}{{\mathcal V}}

\newcommand{\ba}{\begin{array}}
\newcommand{\ea}{\end{array}}
\newcommand{\bs}{\begin{align}\begin{split}\nonumber}
\newcommand{\bsnumber}{\begin{align}\begin{split}}
\newcommand{\es}{\end{split}\end{align}}

\newtheorem{assumption}{ASSUMPTION}

\newcommand{\CI}{\mathrm{CI}}

\newcommand{\mcC}{{\cal C}}
\newcommand{\pol}{\pi} %
\newcommand{\fdim}{r} %

\newcommand{\ldot}[2]{\langle #1, #2 \rangle}

\newcommand{\Var}{\ensuremath{\text{Var}}}

\def\half{\frac{1}{2}}

\def\balign#1\ealign{\begin{align}#1\end{align}}
\def\balignat#1\ealign{\begin{alignat}#1\end{alignat}}
\def\bitemize#1\eitemize{\begin{itemize}#1\end{itemize}}
\def\benumerate#1\eenumerate{\begin{enumerate}#1\end{enumerate}}
\newenvironment{talign}
 {\csname align\endcsname}
 {\endalign}
\def\balignt#1\ealignt{\begin{talign}#1\end{talign}}%

\newcommand{\Cov}{\ensuremath{\mathtt{Cov}}}
\newcommand{\mom}{\ensuremath{{\cal M}}}

\usepackage{etoc}

\title[Double/Debiased Machine Learning for Dynamic Treatment Effects]{Double/Debiased Machine Learning for Dynamic Treatment Effects via $g$-Estimation}

\coltauthor{%
 \Name{Greg Lewis} \Email{glewis@microsoft.com}\\
 \addr Microsoft Research
 \AND
 \Name{Vasilis Syrgkanis} \Email{vasy@microsoft.com}\\
 \addr Microsoft Research%
}

\begin{document}

\maketitle

\begin{abstract}%
We consider the estimation of treatment effects in settings when multiple treatments are assigned over time and treatments can have a causal effect on future outcomes or the state of the treated unit. We propose an extension of the double/debiased machine learning framework to estimate the dynamic effects of treatments, which can be viewed as a Neyman orthogonal (locally robust) cross-fitted version of $g$-estimation in the dynamic treatment regime. Our method applies to a general class of non-linear dynamic treatment models known as Structural Nested Mean Models and allows the use of machine learning methods to control for potentially high dimensional state variables, subject to a mean square error guarantee, while still allowing parametric estimation and construction of confidence intervals for the structural parameters of interest. These structural parameters can be used for off-policy evaluation of any target dynamic policy at parametric rates, subject to semi-parametric restrictions on the data generating process. Our work is based on a recursive peeling process, typical in $g$-estimation, and formulates a strongly convex objective at each stage, which allows us to extend the $g$-estimation framework in multiple directions: i) to provide finite sample guarantees, ii) to estimate non-linear effect heterogeneity with respect to fixed unit characteristics, within arbitrary function spaces, enabling a dynamic analogue of the RLearner algorithm for heterogeneous effects, iii) to allow for high-dimensional sparse parameterizations of the target structural functions, enabling automated model selection via a recursive lasso algorithm. We also provide guarantees for data stemming from a single treated unit over a long horizon and under stationarity conditions.
\end{abstract}


\etocdepthtag.toc{mtchapter}
\etocsettagdepth{mtchapter}{subsection}
\etocsettagdepth{mtappendix}{none}

\section{Introduction}
\label{sec:intro}

Improving outcomes often requires multiple treatments: patients may need a course of drugs to manage or cure a disease, soil may need multiple additives to improve fertility, companies may need multiple marketing efforts to close the sale. To make data-driven decisions, policy-makers need precise estimates of what will happen when a new policy is pursued.   
Because of its importance, this topic has been studied by many communities and under multiple regimes and formulations; examples include the field of reinforcement learning \citep{Sutton1998}, longitudinal analysis in biostatistics \citep{hernan2010causal}, the dynamic treatment regime in causal inference and adaptive clinical trials \citep{lei2012smart}.

This paper offers a new method for estimating and making inferences about the counterfactual effect of a new treatment policy.  Our method is designed to work with observational data, in an environment with multiple treatments (either discrete or continuous), and a high-dimensional state. 
Valid causal inference is necessary to correctly attribute changes in outcomes to the different treatments applied. 
But it is more challenging than in a static context, since there are multiple causal pathways from treatments to subsequent outcomes (e.g. directly, or by changing future states, or by affecting intermediate outcomes, or by influencing future treatments).

Our work bridges many distinct literatures. The first is the econometrics literature on semi-parametric inference \citep{Neyman:1979,robinson:88,newey1990semiparametric,Ai2003,Chernozhukov2016locally,chernozhukov2018double}.  We extend this literature, which has typically focused on static treatment regimes, to dynamic treatment regimes and to estimation of counterfactual dynamic policies. We propose an estimation algorithm that estimates the dynamic effects of interest, from observational data, at parametric root-$n$ rates.
We prove asymptotic normality of the estimates, even when the state is high-dimensional and machine learning algorithms are used to control for indirect effects of the state, as opposed to effects stemming from the treatments.
Our formulation can be viewed as a dynamic version of Robinson's classic partial linear model in semi-parametric inference \citep{robinson:88}, where the controls evolve over time and are affected by prior treatments. Our estimation algorithm builds upon and generalizes the recent work of \citep{chernozhukov2018double,semenova:17} to estimate not only contemporaneous effects, but also dynamic effects over time. In particular, we propose a sequential residualization approach, where the effects at every period are estimated in a Neyman orthogonal manner and then peeled-off from the outcome, so as to define a new ``calibrated'' outcome, which will be used to estimate the effect of the treatment in the previous period.

In doing this, we build on results in the semi-parametric inference literature in bio-statistics \citep{bickel1993efficient,van2003unified,tsiatis2007semiparametric} on estimating causal effects in structural nested models (see \citep{Vansteelandt2016,Robins2004,hernan2010causal,Chakraborty2013} for recent overviews). In particular, our identification and estimation strategy for dynamic treatment effects is a variant of the well-studied $g$-estimation framework for structural nested mean models (SNMMs) \citep{robins1986new,robins1992g,robins1994correcting,robins1997toward,robins2000marginal,Robins2004,lok2012impact,vansteelandt2014structural}. 
The works cited above have provided doubly robust estimators for this setting.
One practical challenge faced by these estimation approaches is that the nuisance functions required are hard to estimate when there are many treatments or treatments are continuous and they do not yield and obvious way to perform cross-fitting, which is crucial for enabling the use of machine learning approaches for nuisance estimation. Moreover, typically $g$-estimation has been used in practice with either logistic regresion or ordinary least squares as nuisance function estimators.

Our approach addresses these challenges by providing a Neyman orthogonal (aka locally robust) $g$-estimation algorithm for linear structural nested mean models \citep{robins1994correcting,Robins2004}, that allows for both continuous and discrete treatments in each time period.  
We propose a two stage algorithm, where in the first stage a sequence of regression and classification models are fitted and evaluated in a cross-fitting manner and in the second stage a simple linear system of equations or a simple square loss minimization problem is solved. This approach also allows for an easy sample-splitting/cross-fitting approach, which allows the use of arbitrary machine learning approaches in the first stage. 

While Neyman orthogonality is a weaker condition than double robustness, it is sufficient to achieve robustness to bias introduced by using machine learning to estimate the nuisance parameters.
This approach thus permits machine learning for the nuisance model estimation, which is practically important in  high-dimensional state space/control variable settings. Our work extends and enriches the $g$-estimation framework for dynamic treatment regimes with arbitrary target dynamic policies in multiple ways:
\begin{enumerate}
    \item Allows the use of aribtrary machine learning algorithms for the estimation of the nuisance components in $g$-estimation, while maintaining valid inference. Our variant of the set of moment conditions used in the $g$-estimation process, requires the estimation of a set of nuisances that corresponds to either a regression or a classification task and which can be estimated a priori in a manner that does not depend on the structural parameter of interest. Existing approaches to achieving double robustness (and hence also local robustness) in $g$-estimation require the estimation of nuisance parameters that implicitly depends on the structural parameter that is being estimated. Solutions to this obstacle has been established in particular sub-cases such as binary treatments and static target policies \citep{hernan2010causal}, or when ordinary least squares is used for one of the nuisance functions \citep{Wallace2019} and this obstacle has been one of the leading factors of not using the doubly robust correction in $g$-estimation (see e.g. \citep{vansteelandt2014structural}). 
    \item Provides finite sample and high probability estimation error rates for the structural parameters of interest that provides explicit dependence on parameters of the data generating process that capture the degree of long-range dependencies of the treatment policy. If for instance, the random exploration to the treatment assignments at period $t$ strongly correlates with all subsequent treatments in a non-vanishing manner, then the finite sample error maybe depend exponentially in the horizon. However, if these long-range dependencies are vanishing, then the error depends only polynomially in the horizon. Such a phase transition is typically hidden in asymptotic normality statements in the longitudinal data analysis literature, that typically keep the horizon fixed and hence this dependence is irrelevant to the asymptotic statement.
    \item Provides a new method for the estimation of heterogeneous structural parameters, with respect to fixed unit characteristics, in structural nested mean models, allowing for arbitrary function spaces to be used to capture parameter heterogeneity. This extends the RLearner algorithm \citep{nie2017quasi} to the dynamic treatment regime, which we dub the \emph{Dynamic RLearner}. We provide finite sample high probability mean squared error rates for the recovered heterogeneous dynamic effect functions, via tight measures of statistical complexity of function spaces (critical radius from the localized Rademacher complexity literature). In practice this enables the use of machine learning approaches, such as random forests, reproducing kernel Hilbert spaces and neural networks, even for the target structural parameter of interest, which is now infinite dimensional. 
    \item Provides a new method for the estimation of high-dimensional structural parameters in structural nested mean models. A typical problem with $g$-estimation is that it requires the practitioner to hard code a parametric form for the dynamic effect functions in structural nested mean models, known as the \emph{blip functions} (see e.g. \citep{Chakraborty2013}). Our work allows for practitioners to use high-dimensional parametric forms and estimate these functions under a sparsity condition on which features are relevant. This enables automated model selection for the \emph{blip functions} via a recursive lasso algorithm.
    \item Provides results for the case where data are stemming from a single treated unit, as opposed to short chains of data from multiple units (panel). Subject to stationarity conditions we provide asymptotic normality properties of a progressive nuisance estimation analogue of our main algorithm.
\end{enumerate}

Many of our aforementioned advancements stem from the fact that the Neyman orthogonal moment that we propose at each stage of $g$-estimation is the gradient of a strongly convex square loss. This allows for stronger finite sample learning rates, as well as the use of techniques from the recent framework of orthogonal statistical learning \citep{foster2019orthogonal,chernozhukov2018plugin}. We extend these techniques to also account for the recursive nature of the final stage estimation and the fact that each of these steps in the final stage cannot be viewed in isolation, since Neyman orthgoonality does not hold within structural parameters fitted at different steps of the final stage and we also have sample re-use. 

Our work is also closely related to the work on doubly robust estimation in offline policy evaluation in reinforcement learning \citep{nie2019learning,Thomas2016,petersen2014targeted,kallus2019efficiently,kallus2019double} from the causal machine learning community.
However, one crucial point of departure from all of these works is that we formulate minimal parametric assumptions that allow us to avoid non-parametric rates or high-variance estimation. Typical fully non-parametric approaches in off-policy evaluation in dynamic settings, requires the repeated estimation of inverse propensity weights at each step, which leads to a dependence on quantities that can be very ill-posed in practice, such as the product of the ratios of states under the observational and the target policy. Moreover, these approaches are typically restricted to discrete states and actions. Our goal is to capture settings where the treatment can also be continuous (investment level, price discount level, drug dosage), and the state contains many continuous variables and is potentially high dimensional. Inverse propensity weighting approaches, though more robust in terms of the assumptions that they make on the model, can be quite prohibitive in these settings even in moderately large data sets. 

Our work is also somewhat related to the online debiasing literature \citep{deshpande2018accurate,deshpande2019online,Zhang2020,zhan2021policy,Hadad2021,bibaut2021post,zhan2021off}, since we perform inference from adaptively collected data. However, our inferential target is different (dynamic effects vs same-period effects). Moreover, in our main setting, we assume multiple independent small chains of samples (each from a separate treated unit), rather than one long time series, while in our single time series setting, we make stationarity assumptions that avoid the main problems from adaptive data collection that are being handled by that literature. It is an interesting avenue for future work to combine techniques from this literature and lift some of our stationarity conditions in the estimation of dynamic treatment effects from a single time series.

\section{An Expository High-Dimensional Partially Linear Markovian Model}\label{sec:expository}

We begin by presenting the main algorithm of this work in the context of a high-dimensional partially linear Markovian data generating process. In Section~\ref{sec:ext} we show that our results generalize to more complex dynamic treatment models, known in the bio-statistics and causal inference literature as Structural Nested Mean Models (SNMMs), but for simplicity of exposition we focus on this simplified version, as it captures all the complexities needed to highlight our main contributions.

Consider a partially linear state space Markov decision process $\{X_t, T_t, Y_t\}_{t=1}^{m}$, where $X_t\in \R^p$ is the state at time $t$, $T_t\in \R^d$ is the action or treatment at time $t$ and $Y_t \in \R$ is an observed outcome of interest at time $t$. We assume that these variables are related via a linear Markovian process:
\begin{align}
    \forall t\in \{1,\ldots, m\}: & & X_t =~& A\cdot T_{t-1} + B\cdot X_{t-1} + \eta_t \nonumber\\
    & & T_t =~& p(T_{t-1}, X_{t}, \zeta_t) \label{eqn:maineq}\\
    & & Y_t =~& \theta_0' T_t + \mu' X_{t} + \epsilon_t \nonumber
\end{align}
where $\eta_t, \zeta_t$ and $\epsilon_t$ are exogenous mean-zero random shocks, independent of all contemporaneous and lagged treatments and states, that for simplicity we assume are each drawn i.i.d. across time. Moreover, for simplicity we assume $T_0=X_0=0$. In Section~\ref{sec:ext}, we will substantially drop the heavy assumptions on the exogenous shocks and merely assume that the observational process satisfies the quite permissive notion of \emph{conditional sequential exogeneity}, which essentially only requires that conditional on the past history, the treatment is randomized in an exogenous manner and there is no unobserved confounder at each stage of the observational decision process. 

The structural parameter $A$ is a $p\times d$ matrix that governs how past treatments affect next period's states. The structural parameter $B$ is a $p\times p$ matrix that governs how past states affect next period's states, i.e. how the system evolves in the absence of treatments. The function $p(T_{t-1}, X_t, \zeta_t)$ is the observational dynamic policy that determines the distribution of next period's treatments as a function of past periods treatments and states. Finally, the parameter vector $\theta_0\in \R^d$ is the contemporaneous effect of the treatments and $\mu\in R^p$ is the contemporaneous effect of the states on the outcome.

\begin{figure}[htpb]
\centering
    \includegraphics[scale=.6]{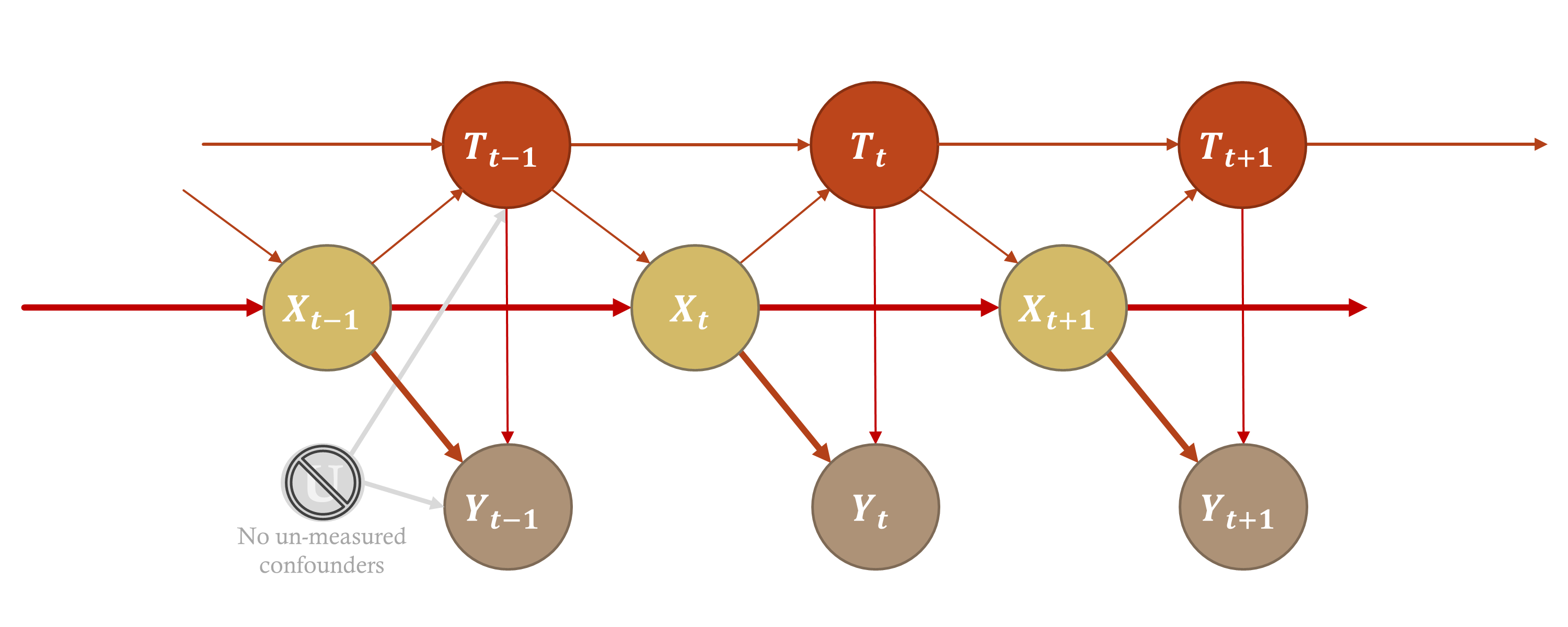}
    \caption{Causal graph associated with expository partially linear Markovian model.}
    \label{fig:cg}
\end{figure}

Our goal is to estimate the effect of a change in the treatment policy on the final outcome $Y_m$. Since throughout the analysis we will primarily care about the final outcome, we denote it for simplicity as:
\begin{align}
Y := Y_m.
\end{align}
A similar analysis can be derived for any other linear combination of the outcomes at all periods.
More concretely, suppose that were to make an intervention and set each of the treatments to some sequence of values: $\{\tau_1, \ldots, \tau_m\}$, then what would be the expected difference in the final outcome $Y_m$ as compared to some baseline policy? For simplicity and without loss of generality, we will consider the baseline policy to be setting all treatments to zero. We will denote this expected difference as: $V(\tau_1, \ldots, \tau_m)$. Equivalently, 
we can express the quantity we are interested in do-calculus: if we denote with:
\begin{align}
    R(\tau) :=~& \E\left[Y \mid do(T_1=\tau_1,\ldots, T_m=\tau_m)\right], &
    V(\tau) :=~& R(\tau) - R(0).
\end{align}
In Section~\ref{sec:snmm}, we will also analyze the estimation of the effect of adaptive counterfactual treatment policies, where the treatment at each step can be a function of the state.

\subsection{Identification via Dynamic Effects}\label{sec:identification}
Our first observation is that we can decompose the quantity $V(\tau)$ into the estimation of the \emph{dynamic treatment effects}: if we were to make an intervention and increase the treatment at period $t$ by $1$ unit, then what is the change $\psi_{t}$ in the outcome $Y_{m}$, for $t \in \{1,\ldots, m\}$; assuming that we set all subsequent treatments to zero (or equivalently to some constant value). This quantity is the effect in the final outcome, that does not go through the changes in the subsequent treatments, due to our observational Markovian treatment policy, but only the part of the effect that goes through changes in the state space $X_t$, that is not part of our decision process. This effect can also be expressed in terms of the constants in our Markov process as:
\begin{align}
   \psi_m =~& \theta_0 & \forall t \in \{1,\ldots, m - 1\}: \psi_{t} =~& \mu'\, B^{m-t-1}\, A
\end{align}
\begin{lemma}\label{lem:char} 
The counterfactual value function $V:\R^{d\cdot m}\rightarrow \R$, can be expressed in terms of the dynamic treatment effects as: $V(\tau_1, \ldots, \tau_m) = \sum_{t=1}^{m} \psi_{t}' \tau_{t}$
\end{lemma}

Thus to estimate the function $V$, it suffices to estimate the dynamic treatment effects: $\psi_1, \ldots, \psi_{m}$.
We first start by showing that the parameters $\psi_1, \ldots, \psi_m$ are identifiable from the observational data. 
Identification is not immediately obvious. If we write the final outcome as a linear function of all the treatments and the initial state by repeatedly expanding the structural equations, we will arrive at an equation of the form:
\begin{align}
Y = \sum_{t=0}^{m} \psi_t' T_{t} + \mu'B^{m} X_1 +  \sum_{j=2}^{m} \mu' B^{m-j} \eta_{j} + \epsilon_m    
\end{align}
However, the shocks $\{\eta_j\}_{j=t+1}^m$ are heavily correlated with the treatments $T_2,\ldots, T_m$, since the shocks at period $t$ affect the state at period $t+1$, which in turn affects the observed treatment at period $t+1$. 
As a result the natural moment condition that the sum of shocks is conditionally mean zero $\E[Y - \sum_{t=0}^{m} \psi_t' T_{t} - \mu'B^{m} X_1 \mid X_1, T_1, \ldots, T_m]=0$  is not valid.
In other words, if we view the problem as a simultaneous treatment problem, where the final outcome is the outcome, then we essentially have a problem of unmeasured confounding (implicitly because we ignored the confounding through intermediate states, sometimes referred as the ``\emph{treatment-confounder feedback}'').

However, note that if we apply this recursive expansion process up until any period $t=\{1,\ldots, m\}$, then we can write:
\begin{align}
Y = \psi_t' T_{t} + \sum_{j=t+1}^{m} \psi_j' T_{j} + \mu'B^{m-t} X_{t} + \sum_{j=t+1}^{m} \mu' B^{m-j} \eta_{j} + \epsilon_m
\end{align}
Since the random shocks $\{\eta_j\}_{j=t+1}^m$ and $\epsilon_m$ are independent of $T_t, X_t$ and mean zero, we thus have that the following conditional moment restriction is satisfied: 
\begin{align}
\E\left[Y - \psi_t' T_{t} - \sum_{j=t+1}^{m} \psi_j' T_{j} - \mu'B^{m-t} X_{t}\mid T_t, X_t\right]=0.
\end{align}
This leads to an identification strategy of the dynamic effects via a recursive peeling process, which as we show in the next section, leads to an estimation strategy that achieves parametric rates. 

\begin{theorem}\label{thm:id}
The dynamic treatment effects satisfy
the following set of conditional moment restrictions: $\forall t \in \{1, \ldots, m\}$
\begin{align}
    \E[\bar{Y}_{t} - \psi_t' T_{t} - \mu'B^{m-t}\,X_{t} \mid T_{t}, X_{t}] = 0
\end{align}
where: $\bar{Y}_{t} = Y - \sum_{j = t+1}^{m} \psi_j' T_{j}$. Moreover, if the covariance matrix $J:=\E[\Cov(T_t, T_t\mid X_t)]$ is invertible, then these conditional moment restrictions uniquely identify $\psi_t$.
\end{theorem}

\subsection{Dynamic DML Estimation}

We now address the estimation problem. We assume that we are given access to $n$ i.i.d. samples from the Markovian process, i.e. we are given $n$ independent time-series, and we denote sample $i$, with $\{X_{t}^{i}, T_t^i, Y_t^i\}$. 
Our goal is to develop an estimator of the function $V$ or equivalently of the parameter vector $\psi=(\psi_{1},\ldots, \psi_m)$. We will consider the case of a high-dimensional state space, i.e. $p\gg n$, but low dimensional treatment space and a low dimensional number of periods $m$, i.e. $d, m\ll n$ is a constant independent of $n$. We want to estimate the parameters $\psi$ at $\sqrt{n}$-rates and in a way that our estimator is asymptotically normal, so that we can construct asymptotically valid confidence intervals around our dynamic treatment effects and our estimate of the function $V$. The latter is a non-trivial task due to the high-dimensionality of the state space. For instance, the latter would be statistically impossible if we were to take the direct route of estimating the whole Markov process (i.e. the high-dimensional quantities $A, B, \mu$): if these quantities have a number of non-zero coefficients that grows with $n$ at any polynomial rate, then known results on sparse linear regression, preclude their estimation at root-n rates (see e.g. \citep{Wainwright15}). However, we are not really interested in these low-level parameters of the dynamic process, but solely on the low dimensional parameter vector $\theta$.
We will treat this problem as a semi-parametric inference problem and develop a Neyman orthogonal estimator for the parameter vector \citep{Neyman:1979,robinson:88,Ai2003,Chernozhukov2016locally,chernozhukov2018double}.

In particular, we consider a sequential version of the double machine learning algorithm proposed in \citep{chernozhukov2018double}. In the case of a single time-period, i.e. $m=0$, then \citep{chernozhukov2018double}, recommends the following estimator for $\psi_m := \theta_0$: using half of your data, fit a model $\hat{q}_0(X_0)$ of $\E[Y_0 \mid X_0]$, i.e. that predicts the outcome $Y_0$ from the controls $X_0$ and a model $\hat{p}_0(X_0)$ for $\E[T_0 \mid X_0]$. Then estimate $\theta_0$ on the other half of the data, based on the estimating equation:
\begin{equation}\label{eqn:moment-static}
 \E[m(\theta; \hat{p}_0, \hat{q}_0)] := \E\left[(\tilde{Y}_0 - \theta_0'\, \tilde{T}_0)\, \tilde{T}_0\right] = 0
\end{equation}
where $\tilde{Y}_0=Y_0-\hat{q}_0(X_0)$ and $\tilde{T}_0 = T_0 - \hat{p}_0(X_0)$ are the residual outcome and treatment. 

We propose a sequential version of the double machine learning process that we call \emph{Dynamic DML} for \emph{dynamic double/debiased machine learning}. Intuitively our algorithm proceeds as follows: 
\begin{enumerate}
    \item We can construct an estimate $\hat{\psi}_m$ of $\psi_m$ in a robust (Neyman orthogonal) manner, by applying the approach of \citep{chernozhukov2018double} on the final step of the process, i.e. on time step $T_m$; this will estimate all the contemporaneous effects of the treatments,
    \item Subsequently we can remove the effect of the observed final step treatment from the observed final step outcome, i.e. by re-defining the random variable $\bar{Y}_{m-1}^{i} = Y^{i} - \psi_m'\, T_{m}^i$; doing this we have removed any effects on $Y_{m}^i$, caused by the final treatment $T_{m}^i$.
    \item We can then estimate the one-step dynamic effect $\psi_{m-1}$, by performing the residual-on-residual estimation approach with target outcome the ``calibrated'' outcome $\bar{Y}_{m-1}^{i}$, treatment $T_{m-1}^i$ and controls $X_{m-1}^i$. Theorem~\ref{thm:id} tells us that the required \emph{conditional exogeneity} moment required to apply the residualization process is valid for these random variables. We can continue in a similar manner, by removing the estimated effect of $T_{m-1}$ from $\bar{Y}_{m-1}^i$ and repeating the above process.
\end{enumerate}
We provide a formal statement of the Dynamic DML process in Algorithm~\ref{alg:sr}, which also describes more formally the sample splitting and cross-fitting approach that we follow in order to estimate the nuisance models $p$ and $q$ required for calculating the estimated residuals.

\begin{algorithm}[htpb]
   \caption{Dynamic DML: Neyman orthogonal (locally robust) $g$-estimation}
   \label{alg:sr}
\begin{algorithmic}
   \STATE {\bf Input:} A data set consisting of $n$ samples of $m$-length paths: $\left\{\left(X_1^i, T_1^i, \ldots, X_m^i, T_m^i, Y^i\right)\right\}_{i=1}^n$
   \STATE Randomly split the $n$ samples in $S, S'$
   \FOR{each $t \in \{1, \ldots, m\}$} 
       \STATE Regress $Y$ on $X_{t}$ using $S$ to learn estimate $\hat{q}_t$ of model: $q_{t}^*(x_{t}) := \E[Y \mid {X}_{t}={x}_t]$
       \STATE Calculate residuals for each sample $i\in S'$: $\tilde{Y}_{t}^i := Y^i - \hat{q}_{t}({X}_{t}^i)$
       \STATE Vice versa use $S'$ to learn model $\hat{q}_t$ and calculate residuals on $S$.
       \FOR{each $j \in \{t, \ldots, m\}$}
            \STATE Regress $T_{j}$ on ${X}_{t}$ using $S$ to learn estimate $\hat{p}_{j, t}$ of model:
            $p_{j, t}^*({x}_{t}) := \E[T_j \mid {X}_{t}={x}_t]$
            \STATE Calculate residuals for each sample $i\in S'$: $\tilde{T}_{j, t}^i := T_{j}^i - p_{j, t}({X}_{t}^i)$
            \STATE Vice versa use $S'$ to learn model $\hat{p}_{j, t}$ and calculate residuals on $S$.
       \ENDFOR
   \ENDFOR
   \STATE Using all the data $S\cup S'$
   \FOR{$t=m$ {\bfseries down to} 1}
        \STATE Regress $\bar{Y}_t^i := \tilde{Y}_{t}^i - \sum_{j=t+1}^m \hat{\psi}_{j}' \tilde{T}_{j, t}^i$ on $\tilde{T}_{t, t}^i$ with oridinary least squares, within an $\ell_2$-ball:
        \begin{align}\label{eqn:constraint-erm-sr}
            \hat{\psi}_t = \argmin_{\psi_t\in \R^d: \|\psi_t\|_2\leq M} \frac{1}{n}\sum_{i=1}^n \left(\bar{Y}_t^i - \hat{\psi}_t'\tilde{T}_{t,t}^i\right)^2
        \end{align}
        Alternatively, find solution $\hat{\psi}_t \in \R^d$ to the system of equations:
        \begin{equation}\label{eqn:z-estimator-alg}
            \frac{1}{n} \sum_{i=1}^n \left(\bar{Y}_t^i - \hat{\psi}_{t}'\tilde{T}_{t,t}^{i}\right)\, \tilde{T}_{t,t}^{i} = 0
        \end{equation}
   \ENDFOR
   \STATE {\bf Return:} Dynamic treatment effect estimates $\hat{\psi} = (\hat{\psi}_1, \ldots, \hat{\psi}_m)$
\end{algorithmic}
\end{algorithm}

\subsection{Estimation Rates and Normality}\label{sec:main}

We show that subject to the first stage models of the conditional expectations achieving a small (but relatively slow) estimation error, then the recovered parameters are root-n-consistent and asymptotically normal.

In the static case, the estimate that is based on the estimating Equation~\eqref{eqn:moment-static} is a special case of a broader class of moment based $Z$-estimators, where the true parameter $\theta$ is known to satisfy a vector of moment restrictions: $\E[m(W;\theta, \nu)]=0$, where $W$ is the vector of all random variables and $\nu\in \mcV$ is an unknown (potentially infinite dimensional) nuisance parameter, which we do not care about, but on which the moment conditions depends.A plug-in $Z$-estimate $\hat{\theta}$ is a solution to an empirical analogue of a vector of moment equations $\E_n[m(W;\theta, \hat{\nu})] = \frac{1}{n}\sum_{i=1}^n m(W^i;\theta, \hat{\nu})$, where $\hat{\nu}$ is some estimate of the nuisance parameter obtained via some separate statistical learning process and potentially on a separate sample.

A vector of moments satisfies Neyman orthogonality (aka local robustness) if:
\begin{align}\label{eqn:neyman-orthogonality}
\forall \nu \in \mcV: \frac{\partial}{\partial t}\E[m(W;\theta_0, \nu_0 + t\,(\nu-\nu_0))]\mid_{t=0} = 0
\end{align}
Neyman orthogonality implies that small perturbations to the nuisance functions around their true values only has a second order effect on the moment function and hence cannot impact a lot the target parameter estimate. Neyman orthogonality (accompanied with sample splitting techniques) allows one to estimate the target parameter at $\sqrt{n}$-rates and with an asymptotic normal distribution, subject only to much slower mean squared error rates for the nuisance components. Thus allowing machine learning approaches to be used for nuisance estimation. 

Our asymptotic normality proof relies on showing that one can re-interpret our \emph{Dynamic DML} algorithm as a $Z$-estimator based on a set of moments that satisfy the property of Neyman orthogonality. Our finite sample $\ell_2$-error result uses the fact that each step of the recursive process in the final stage is a minimization of a strongly convex loss and crucially invokes Neyman orthogonality in arguing that the true parameter is an approximate first order optimal of the population analogue of the strongly convex loss. The further difficulty in the finite sample result is that the loss at each stage is biased due to the errors propagating from previous stage estimates, with respect to which Neyman orthogonality is not satisfied and which are not constructed in a cross-fitting manner, which introduces sample re-use considerations. We provide a recursive upper bound on the $\ell_2$-errors and complete the theorem by induction.

To present the theorems, we introduce some notation. Let ${h}=\{{p}_{j,t}, {q}_t\}_{1\leq t\leq j\leq m}$ denote the vector of all nuisance functions and $h^*=\{{p}_{j,t}^*, {q}_t^*\}_{1\leq t\leq j\leq m}$ their corresponding true values. Moreover, let $\hat{\psi}=(\hat{\psi}_1,\ldots,\hat{\psi}_m)$ denote the vector of dynamic effect parameter estimates, and $\psi^*=(\psi_1^*, \ldots, \psi_m^*)$ their corresponding true values. We provide both finite sample $\ell_2$-error rates and asymptotic normality of our estimates (proofs in Appendix \ref{app:finite_sample} and \ref{app:normality}), subject to mean squared error guarantees for the nuisance functions. At the end of the section we discuss how the required guarantees for the nuisance functions can be easily satisfied using estimation algorithms such as the Lasso and under sparsity conditions.

\begin{theorem}[Finite Sample $\ell_2$-error]\label{thm:mse}
Suppose that all random variables and the ranges of all nuisance functions are absolutely bounded by a constant $H$ a.s.. Moreover, suppose that:
\begin{align}
\E[\Cov(T_t, T_t\mid X_t)]\succeq \lambda I.
\end{align}
For each split $O\in \{S,S'\}$, let $\hat{h}_O$, denote the nuisance estimate applied to samples in $O$. Let:
\begin{align}
    \rho(\hat{h}) :=~& \|\hat{p}_{t,t}-p_{t,t}^*\|_{2,2} \left(\|\hat{q}_t - q_t^*\|_{2} + 4\, M\,\sum_{j=t}^m \|\hat{p}_{j, t} - p_{j,t}^*\|_2 \right)\\
    \kappa :=~& \max_{t=1}^m \frac{2}{\lambda} \sum_{j=t+1}^{m} \|\E[\Cov(T_j, T_t\mid X_t)]\|_{op}\\
    \mu(\delta) :=~& \frac{2}{\lambda} \left(2 d H m \left(32 H M \sqrt{\frac{d\, m}{n}} + \sqrt{\frac{9\,d\,\log(2d/\delta)}{n}}\right) + \max_{O\in \{S,S'\}} \rho(\hat{h}_O)\right)
\end{align}
for some sufficiently large universal constant $c_0$. Then with probability $1-m\,\delta$, the estimate produced by Equation~\eqref{eqn:constraint-erm-sr} of Algorithm~\ref{alg:sr} satisfies:
\begin{equation}
    \forall t\in [m]: \|\hat{\psi}_t - \psi_t^*\|_2 \leq \mu(\delta) \frac{\kappa^{m-t+1} - 1}{\kappa - 1}
\end{equation}
\end{theorem}

\begin{theorem}[Asymptotic Normality and Inference]\label{thm:normality}
Let $\mcD_n$ be a sequence of families of data generating processes obeying Equation~\eqref{eqn:maineq} and such that all random variables and the ranges of all nuisance functions are bounded by a constant a.s.. Moreover, for any $D\in \mcD_n$:
\begin{align}
\E[\Cov(T_t, T_t\mid X_t)]\succeq \lambda I.
\end{align}
and $\max_{O\in \{S,S'\}} \|\hat{h}_O-h^*\|_{2,2}=o_p(1)$ and
$\max_{O\in \{S, S'\}} \rho(\hat{h}_O) = o_p(n^{-1/2})$,
where $\hat{h}_O$ and $\rho(\cdot)$ as defined in Theorem~\ref{thm:mse} with $M=\max_{t=1}^m \|\psi_t^*\|_2$.

Let $J$ denote a $(m\,d)\times (m\,d)$ upper triangular matrix consisting of $d\times d$ blocks, such that the $(t,j)$ block is defined as $J_{t,j}:=1\{t\leq j\}\E[Cov(T_t, T_j\mid X_t)]$. Moreover, let $\Sigma$ be a $(m\,d)\times (m\,d)$ block matrix with $d\times d$ blocks, such that the $(t,j)$ block is defined as:
\begin{align}
\Sigma_{t,j} := \E[(\bar{Y}_{t-1} - \E[\bar{Y}_{t-1}\mid X_t])\, (T_t - \E[T_t\mid X_t])\, (T_j - \E[T_j\mid X_j])'\, (\bar{Y}_{j-1} - \E[\bar{Y}_{j-1}\mid X_j])]
\end{align}
with $\bar{Y}_t$ as defined in Theorem~\ref{thm:id} and let $V=J^{-1} \Sigma (J^{-1})'$. Then:
\begin{align}
    \sqrt{n} V^{-1/2} (\hat{\psi} - \psi^*) = -\frac{1}{\sqrt{n}}\sum_{i=1}^n V^{-1/2} J^{-1}\, m(Z^i; \psi^*, h^*) + o_p(1) \to_d N(0, I_{d\cdot m})
\end{align}
where $m(\cdot) = (m_1(\cdot);\ldots; m_m(\cdot))$ and $m_t(Z; \psi^*, h^*) = (\bar{Y}_{t-1} - \E[\bar{Y}_{t-1}\mid X_t])\, (T_t - \E[T_t\mid X_t])$.

Moreover, let $\hat{J}$ be an estimate of $J$ with the $(t,j)$ block being: $\hat{J}_{t,j} = 1\{j\geq t\} \frac{1}{n}\sum_i \tilde{T}_{t,t}^i (\tilde{T}_{j,t}^i)'$ and $\hat{\Sigma}$ be the estimate of the $\Sigma$, whose $(t,j)$ block is defined as:
\begin{align}
    \hat{\Sigma}_{t,j} = \frac{1}{n} \sum_{i=1}^n \left(\bar{Y}_t^{i} - \hat{\psi}_t'\tilde{T}_{t,t}^i\right) \tilde{T}_{t,t}^i\, (\tilde{T}_{j,j}^i)' \left(\bar{Y}_j^i - \hat{\psi}_j'\tilde{T}_{j,j}^i\right)
\end{align}
and let $\hat{V}=\hat{J}^{-1} \hat{\Sigma} (\hat{J}^{-1})'$.
Let $\Phi$ be the CDF of the standard normal distribution. Then for any vector $\nu \in \R^{d\,m}$, the confidence interval:
\begin{align}
    \CI := \left[\nu'\hat{\psi} \pm \Phi^{-1}(1-\alpha/2)\, \sqrt{\frac{\nu'\hat{V} \nu}{n}}\right]
\end{align}
is asymptotically uniformly valid:
    $\sup_{D\in \mcD_n} \left|\Pr_D(\nu'\psi^* \in \CI) - (1-\alpha)\right| \to 0$.
\end{theorem}

\paragraph{Concrete Rates for Lasso Nuisance Estimates.} Suppose that the observational policy $p$ is also linear, i.e.
\begin{align}\label{eqn:linear-policy}
p(X, \zeta)=\Gamma X + \zeta
\end{align}
for some $d\times p$ matrix $\Gamma$. Then all the models $q_{t}$ and $p_{j, t}$ are high-dimensional linear functions of their input arguments, i.e. $q_{t}(x)= \phi_t'x$ and $p_{j, t}(x) = \Pi_{j, t} x$. If these linear functions satisfy a sparsity constraint then under standard regularity assumptions we can guarantee if we use the Lasso regression to estimate each of these functions that w.p. $1-\delta$, the estimation error of all nuisance models is $O\left(s\sqrt{\frac{\log(p/\delta)}{n}}\right)$, where $s$ is an upper bound on the number of non-zero coefficients. One sufficient regularity condition is that the expected co-variance matrix of every period's state has full rank, i.e. $\E[X_{t} X_{t}']\succeq \lambda I$ (we note that for an MSE rate we do not require the minimum eigenvalue condition, albeit then a computationally inefficient, support enumeration based estimation algorithm needs to be used and the Lasso results would not apply). Thus the requirements of the main theorems of this section would be satisfied as long as the sparsity grows as $s=o(n^{1/4})$, so that the error from the nuisance estimates is of second order importance. These sparsity conditions are for instance satisfied if only $s$ coordinates of the high-dimensional state, which coevolve separately from the remainder states, have any effect on the final outcome (i.e. are outcome-relevant), and similarly if only $s$ coordinates of the high-dimensional state, which coevolve separately from the remainder states, enter the observational policy.

\paragraph{The Constants $\lambda, \kappa$.} Given the potentially cryptic nature of some of the constants in our main theorems, we connect them here to some more low level quantities in the data generating process for the case of linear policies, i.e. under Equation~\eqref{eqn:linear-policy}. There are two main constants $\kappa$ and $\lambda$ that govern our finite sample error rates. Especially parameter $\kappa$ greatly impacts the estimation rate, as a function of the number of iterations $m$, since for $\kappa < 1$ the estimation error grows polynomially with the number of rounds, while for $\kappa > 1$ it grows exponentially. Thus understanding which regime occurs in a setting of interest is of great importance in understanding whether long time sequences can be tolerated. For simplicity of the calculations and exposition, we will further assume a scalar treatment (or binary treatment), i.e. $d=1$, in which case we will write: $p(x, \zeta)=\gamma'x + \zeta$. Moreover, we note that in this case the matrix $A$ in Equation~\eqref{eqn:maineq} is a column vector and we will denote it with $\alpha$.

First we note that the parameter $\lambda$ can be easily characterized under the linear policy as:
\begin{align}
    \E[\Cov(T_t, T_t\mid X_t)] = \E[\zeta_t^2]
\end{align}
Thus $\lambda=\E[\zeta^2]$ is the variance of the exploration/randomization of the observational policy deployed at each round $t$. Now let us analyze the constant $\kappa$. First we need to understand the conditional covariances $\Cov(T_j, T_t\mid X_t)$. Note that, if we denote with $\Delta = \alpha\, \gamma' + B$, then we can write by recursively expanding the linear Markovian expressions and using the fact that noise shocks are exogenous and jointly independent (see Appendix~\ref{app:kappa} for details):
\begin{align}
    \Cov(T_t, T_j\mid X_t) = \E\left[\zeta_t^2\right] \gamma'\Delta^{j-t}\alpha
\end{align}
and we conclude that:
\begin{align}
    \kappa = \max_{t=1}^m \frac{2}{\lambda} \sum_{j=t+1}^{m} \|\Cov(T_t, T_j\mid X_t)\|_{op} \leq 2 \|\gamma\|_2 \|\alpha\|_2 \|\Delta\|_{op} \frac{\|\Delta\|_{op}^{m-1} - 1}{\|\Delta\|_{op}-1}
\end{align}
Note that the matrix $\Delta:=\alpha\gamma' + B$ captures the relationship between treatment $X_t$ and $X_{t-1}$ under the observational policy. Observe that under standard linear system theory $\|\Delta\|_{op}\leq 1$ is required for the system of states to be stable and not rapidly growing, which we would expect in many settings. Thus in practicy, under some stationarity of the states we would expect $\|\Delta\|_{op}$ to be small and therefore $\kappa$ to also be small. On the other hand if the data generating process is far from stationary and has long-range dependencies (i.e. a small change in the initial state at period $1$ can have a tremendous impact on the state at period $M$), then the estimation algorithm will incur an exponential dependence on the time horizon $m$.

\subsection{Dependent Single Time-Series Samples}\label{sec:dependent}

Thus far we have assumed that we are working with $n$ independent time series, each of duration $m$. Though this is applicable to many settings where we have panel data with many units over time, in some other settings it is unreasonable to assume that we have many units over time, but rather that we have the same unit over a long period. In this case, we would want to do asymptotics as the number of periods grows. Our goal is still to estimate the dynamic treatment effects, i.e. the effect $\theta_{\kappa}$ of a treatment at period $t$ on an outcome in period $t+\kappa$, for $\kappa\in \{0, \ldots,m\}$) for some fixed look-ahead horizon $m$.

These quantities can allow us to evaluate the effect of counterfactual treatment policies on the discounted sum of the outcomes, i.e. $\sum_{t=0}^{\infty} \gamma^t Y_t$ for $\gamma<1$. We can write the counterfactual value function for any non-adaptive policy as:
$V(\tau) = \sum_{t=0}^{\infty} \gamma^t \sum_{q\leq t} \theta_{t-q} \tau_{q}$. Assuming outcomes are bounded, the effect $\sum_{q\leq t} \theta_{t-q} \tau_{q}$ on any period $t$ can be at most some constant. Thus taking $m$ to be roughly $\log_{\gamma}(n)$, suffices to achieve a good approximation of the effect function $V(\tau)$, since the rewards vanish after that many periods, i.e. if we let:
$V_{m}(\tau) = \sum_{t=0}^{m} \gamma^t \sum_{q\leq t} \theta_{t-q} \tau_{q}$, then observe that: $\|V_{m}(\tau) - V(\tau)\| \leq O(\gamma^{m})$. Thus after $m=\log_{1/\gamma}(n)$, we have that the approximation error is smaller than $1/\sqrt{n}$. Thus it suffices to learn the dynamic treatment effect parameters for a small number of steps. To account for this logarithmic growth, we will make the dependence on $m$ explicit in our theorems below.

For any $m$, we will estimate these parameters by splitting the time-series into sequential $B=n/m$ blocks of size $m$. Then we will treat each of these blocks roughly as independent observations and apply our dynamic DML algorithm to estimate parameters $\hat{\psi}_t$ and observe that under the markovian stationary nature of the DGP, we have that $\hat{\psi}_t = \hat{\theta}_{m-t}$. We denote the resulting estimate as $\hat{\theta}$. The main challenge in our proofs is dealing with the fact that these blocks are not independent but serially correlated. However, we can still apply techniques, such as martingale Bernstein concentration inequalities and martingale Central Limit Theorems to achieve the desired estimation rates.

The other important change that we need to make is in the way that we fit our nuisance estimates. To avoid using future samples to train models that will be used in prior samples (which would ruin the martingale structure), we instead propose a progressive nuisance estimation fitting approach, where at every period, all prior blocks are used to train the nuisance models and then they are evaluated on the next block. We present a formal description of this progressive splitting process in Algorithm~\ref{alg:block-sr} and prove asymptotic normality of the resulting estimate.

\begin{theorem}[Asymptotic Normality with Single Time Series]\label{thm:normality-dependent}
Let $\mcD_n$ be a sequence of families of data generating processes obeying Equation~\eqref{eqn:maineq}, with the further restriction that $p(x, t, \zeta) = f(x, t) + \zeta$, and such that all random variables and the ranges of all nuisance functions are bounded by a constant a.s. and that $M:=\max_{t=1}^m \|\psi_t^*\|_2$ is bounded by a constant. Moreover, for any $D\in \mcD_n$:
\begin{align}
\E[\Cov(T_t, T_t\mid X_t)] = \E[\zeta_t\, \zeta_t'] \succeq \lambda I.
\end{align}
Moreover, let $\mcF_b$ denote the filtration up until (not including) block $b$ and let:
\begin{align}
    \epsilon_B(\hat{h}) :=~&  \frac{2}{B}\sum_{b=B/2}^B \E[\|\hat{h}_b(Z_b) - h^*(Z_b)\|_{2}^2\mid \mcF_b] \\
    \kappa :=~& \max_{t=1}^m \frac{1}{\lambda} \sum_{j=t+1}^m \|\E[\Cov(T_{b,t}, T_{b,j}\mid X_{b,t})]\|_{op}
\end{align}
where $Z$ denotes all random variables within a block $b$ and $\hat{h}_b$ are the estimates of the nuisance models used within that block. Assume that $m$ and $\hat{h}$ satisfy for some constant $\epsilon>0$:
\begin{align}
m^3 \sqrt{\log(d\,m)}\max_{t\in [m]}\frac{(\kappa + \epsilon)^{m-t+1}-1}{\kappa + \epsilon - 1} \E[\epsilon_B(\hat{h})]=~& o(1), &
\sqrt{B} m^2 \E[\epsilon_B(\hat{h})] =~& o(1)\\
\frac{m^3 \log(d\,m)}{\sqrt{B}} \max_{t\in [m]}\frac{(\kappa + \epsilon)^{m-t+1}-1}{\kappa + \epsilon - 1} =~& o(1) &
\end{align}
Let $J$ denote a $(m\,d)\times (m\,d)$ upper triangular matrix consisting of $d\times d$ blocks, such that the $(t,j)$ block is defined as $J_{t,j}:=1\{t\leq j\}\E[Cov(T_t, T_j\mid X_t)]$. Moreover, let $\Sigma$ be a $(m\,d)\times (m\,d)$ block diagonal matrix with $d\times d$ blocks, such that the $(t,t)$ diagonal block is defined as:
\begin{align}
\Sigma_{t,t} := \E[(\bar{Y}_{t-1} - \E[\bar{Y}_{t-1}\mid X_t])^2\, \Cov(T_t, T_t\mid X_t)]
\end{align}
Then the estimate $\hat{\psi}$ defined in Algorithm~\ref{alg:block-sr} satisfies:
\begin{align}
    \sqrt{B/2} \Sigma^{-1/2} J (\hat{\psi} - \psi^*) \to_d N(0, I_{d\cdot m})
\end{align}
\end{theorem}

\begin{algorithm}[htpb]
   \caption{Block Dynamic DML with Progressive Nuisance Estimation}
   \label{alg:block-sr}
\begin{algorithmic}
   \STATE {\bf Input:} A single time-series consisting of $n=B\, m$ samples: $\left\{\left(X_1^i, T_1^i, \ldots, X_m^i, T_m^i, Y^i\right)\right\}_{i=1}^n$
   \STATE Partition the data into $B=n/m$ blocks of $m$ periods
   \STATE Denote with $X_{t}^b, T_{t}^b, Y_{t}^b$ the state, action, outcome pairs in the $t$-th period of block $b$.
   \FOR{each block $b \in \{B/2, \ldots, B\}$}
       \FOR{each $t \in \{1, \ldots, m\}$} 
           \STATE Regress $Y_m^{b'}$ on $X_{t}^{b'}$ using blocks $b'< b$ to learn estimate $\hat{q}_{b,t}$ of: $q_{t}^*(x_{t}) := \E[Y_m^b \mid {X}_{t}^b={x}_t]$
           Calculate residuals: $\tilde{Y}_{t}^b := Y_m^b - \hat{q}_{b,t}({X}_{t}^b)$
           \FOR{each $j \in \{t, \ldots, m\}$}
                \STATE Regress $T_{j}^{b'}$ on ${X}_{t}^{b'}$ using all blocks $b'< b$ to learn estimate $\hat{p}_{b, j, t}$ of conditional expectation:
                $p_{b, j, t}^*({x}_{t}) := \E[T_j^b \mid {X}_{t}^b={x}_t]$
                \STATE Calculate residuals: $\tilde{T}_{j, t}^b := T_{j}^b - p_{b, j, t}({X}_{t}^b)$
           \ENDFOR
       \ENDFOR
   \ENDFOR
   \STATE Using all the blocks $b\in \{B/2,\ldots, B\}$
   \FOR{$t=m$ {\bfseries down to} 1}
        \STATE Let $\bar{Y}_t^b := \tilde{Y}_{t}^b - \sum_{j=t+1}^m \hat{\psi}_{j}' \tilde{T}_{j, t}^b$
        \STATE Find solution $\hat{\psi}_t \in \R^d$ to the system of equations:
        \begin{equation}\label{eqn:z-estimator-alg-block}
            \frac{2}{B} \sum_{b=B/2}^B \left(\bar{Y}_t^b - \hat{\psi}_{t}'\tilde{T}_{t,t}^{b}\right)\, \tilde{T}_{t,t}^{b} = 0
        \end{equation}
   \ENDFOR
   \STATE {\bf Return:} Dynamic treatment effect estimates $\hat{\psi} = (\hat{\psi}_1, \ldots, \hat{\psi}_m)$
\end{algorithmic}
\end{algorithm}

\paragraph{Finite sample $\ell_2$-error.} We note that the proof of the latter theorem also provides a finite sample $\ell_2$-error bound. However we omit a separate such theorem for succinctness.

\paragraph{Conditions on $m$ and dependence on $\kappa$.} We provide some exposition on the conditions on $m$ as a function of the eigenvalues of the linear dynamical system in the case of a single treatment and a linear policy, e.g. $f(x,t)=\gamma'x$, as described in the corresponding remark at the end of Section~\ref{sec:main}. Note that as long as the observational state transition matrix $\Delta=\alpha\gamma'+B$ has a small maximum eigenvalue $\ll 1$ and that $\|\alpha\|_2\|\gamma\|_2\ll 1$, then $\kappa \ll 1$, irrespective of the value of $m$. In other words, in this setting the correlations among the randomizations in the treatments are vanishing in an exponential manner and hence even if we estimate parameters in a long-chain, the estimation errors do not propagate in a manner that explodes exponentially with the length of the path. In that case the conditions on $m$ in Theorem~\ref{thm:normality-dependent} simplify to:
\begin{align}
m^3 \sqrt{\log(d\,m)}\max_{t\in [m]}\E[\epsilon_B(\hat{h})]=~& o(1), &
\sqrt{B} m^2 \E[\epsilon_B(\hat{h})] =~& o(1) &
\frac{m^3 \log(d\,m)}{\sqrt{B}} =~& o(1) &
\end{align}
Note that the number of nuisance functions also grows quadratically with $m$. Thus we should expect that $\epsilon_B(\hat{h})$ to also grow as $m^2$ times the convergence rate of each of the nuisance components. If each nuisance component estimate (denoted here  $\hat{f}$) satisfies that $\frac{2}{B} \sum_{b=B/2}^{B} \E[(\hat{f}(Z_b) - f^*(Z_b))^2\mid \mcF_b]$ convergences at a rate of $\frac{1}{B^{1/2+\epsilon}}$, then a sufficient condition for all the latter properties is:
\begin{align}
    \frac{m^4 \log(d\,m)}{B^{\epsilon}} = \frac{m^{4+\epsilon} \log(d\,m)}{n^{\epsilon}} = o(1)
\end{align}
This allows for $m$ to grow polynomially with $n$, i.e. it suffices that $m = O(n^{\epsilon/(4+\epsilon) - \delta})$, for any $\delta>0$. Thus we can achieve very small approximation error if we are interested in a discounted reward, with discount $\gamma$, as in the beginning of a section, where with simply $m=\log_{1/\gamma}(n)$ we could achieve an error of $\gamma$. Thus we have that as long as the observational policy is such that the linear system is stable, estimation of long-term discounted rewards is feasible via Algorithm~\ref{alg:block-sr}.

\section{Generalization to Structural Nested Mean Models (SNMMs)}\label{sec:ext}\label{sec:snmm}

We present a more formal treatment of the extension of our main algorithm to $g$-estimation of structural nested models in biostatistics \citep{robins1986new}. Consider an arbitrary time-series process $\{X_t, T_t\}_{t=1}^{m}$, with $X_t\in \mcX_t$ and $T_t\in \mcT_t$. Let $Y$ denote some final outcome of interest. For any time $t$, let $\bar{X}_t=\{X_1,\ldots, X_t\}$ and $\bar{T}_t=\{T_1,\ldots, T_t\}$, denote the sequence of the variables up until time $t$ and similarly, let $\underline{X}_t = \{X_t, \ldots, X_m\}$ and $\underline{T}_t=\{T_t,\ldots, T_m\}$. We will also denote with $\bar{x}_t, \bar{\tau}_t, \underline{x}_t, \underline{\tau}_t$, corresponding realizations of the latter random sequences. Let $\pol=(\pol_1, \ldots, \pol_m)$ denote any dynamic policy, such that for each $t$, $\pol_t$ maps a history $\bar{x}_t, \bar{\tau}_{t-1}$ into a next period action $\tau_{t}$. For any such dynamic policy, let $Y^{(\pol)}$ denote the counterfactual outcome under policy $\pol$. For any static policy $\tau\in \times_{t=1}^m \mcT_t$, we will overload notation and let $Y^{(\tau)}$ denote the counterfactual outcome under this static treatment policy. Moreover, for any two policies (static or dynamic) we will be denoting with $(\bar{\pol}'_t, \underline{\pol}_{t+1})$, the policy that follows $\pol'$ up until time $t$ and then continues with policy $\pol$. We let $0\in \mcT_t$ denote a baseline policy value, which could be appropriately instantiated based on the context. 

We assume that the data generating process satisfies the following sequential conditional randomization condition:
\begin{assumption}[\emph{Sequential Conditional Exogeneity}]\label{ass:cond-ex} The data generating process satisfies the following conditional independence conditions:
\begin{equation}
    \forall t\in [m]: \{Y^{(\tau)}, \tau\in \times_{t=1}^m \mcT_t\} \indep T_t \mid \bar{T}_{t-1}, \bar{X}_t
\end{equation}
\end{assumption}
Identification of mean counterfactual outcomes $\E[Y^{(\pol)}]$ for a target policy of interest $\pol$ can be expressed in terms of the following conditional expectation functions:
\begin{align}
    \gamma_{t}(\bar{x}_t, \bar{\tau}_t) = \E\left[Y^{(\bar{\tau}_{t}, \underline{\pol}_{t+1})} - Y^{(\bar{\tau}_{t-1}, 0, \underline{\pol}_{t+1})} \mid \bar{T}_{t}=\bar{\tau}_{t}, \bar{X}_{t}=\bar{x}_t\right] 
\end{align}
which corresponds to the mean change in outcome if we go to all units which received treatment $\bar{\tau}_t$ up until time $t$ and had observed state history $\bar{x}_t$ and we remove their last treatment, while we subsequently always continue with the target policy $\pol$. These functions are known as the \emph{blip} functions \citep{Chakraborty2013,Robins2004} and can be shown to be non-parametrically identifiable, assuming sequential conditional exogeneity and a sequential analogue of the positivity (aka overlap) assumption \citep{Robins2004}. 

Theorem~3.1 of \citep{Robins2004} combines a telescoping sum argument and the sequential randomization condition to express counterfactual outcomes in terms of blip functions. We restate this result here, adapting it to our notation and providing a proof for completeness:
\begin{lemma}[Identification via Blip Functions]\label{lem:cntf-char-rho} For any dynamic policy $\pol$ and under the sequential conditional exogeneity assumption, the following identity holds about the counterfactual outcomes:
\begin{align}
    \E\left[Y^{(\bar{\tau}_{t-1}, \underline{\pol}_t)}\mid \bar{X}_t, \bar{T}_t=\bar{\tau}_t\right] = \textstyle{\E\left[Y + \sum_{j=t}^{m} \rho_j(\bar{X}_j, \bar{T}_j) \mid \bar{X}_t, \bar{T}_t=\bar{\tau}_t\right]}\label{eqn:cntf-rho}
\end{align}
where:
\begin{align}
\rho_j(\bar{X}_j, \bar{T}_j):=\gamma_j(\bar{X}_j, (\bar{T}_{j-1}, \pol(\bar{X}_j, \bar{T}_{j-1}))) - \gamma_j(\bar{X}_j, \bar{T}_j).
\end{align}
Hence also:
\begin{align}\label{eqn:off-policy}
\E\left[Y^{(\pol)}\right]=\E\left[Y + \sum_{t=1}^{m} \rho_t(\bar{X}_t, \bar{T}_t)\right]
\end{align}
\end{lemma}

Importantly, the conditioning set in Equation~\eqref{eqn:cntf-rho} contains the observed $t$ periods treatment. Intuitively, each term $\rho_j$, removes from the outcome the \emph{blip effect} of the observed action $T_j$ and adds the \emph{blip effect} of the target action $\pol(\bar{X}_j, \bar{T}_{j-1})$. Lemma~\ref{lem:cntf-char-rho}, together with conditional sequential exogeneity also implies that the following set of moment restrictions must be satisfied (the following is an adaptation of Theorem~3.2 of \citep{Robins2004} to our notation and we include its proof for completeness).

\begin{lemma}[Moment Restrictions for Blip Functions]\label{lem:moment-restrictions} For any parameterization of the blip functions $\gamma_t(\bar{x}_t, \bar{\tau}_t; \psi_t)$, if we let the random variable $H_t(\psi):=Y + \sum_{j=t}^{m} \rho_j(\bar{X}_j, \bar{T}_j; \psi_j)$, then the true parameter vector $\psi^*$ must satisfy the moment restrictions:
\begin{align}
    \forall t\in [m], \forall f \in \mcF: \E\left[ H_t(\psi^*)\, \left(f(\bar{X}_t, \bar{T}_t) - \E[f(\bar{X}_t, \bar{T}_t)\mid \bar{X}_t, \bar{T}_{t-1}]\right)\right] = 0
\end{align}
where $\mcF$ contains all functions mapping histories $\bar{x}_t, \bar{\tau}_t$ to $\R$.
\end{lemma}

To achieve parametric estimation rates for the quantities of interest, we will need to further make a semi-parametric assumption, i.e. that the \emph{blip} functions take a low-dimensional parametric form:
\begin{assumption}[Linear Blip Functions]\label{ass:linear-snmm}
The blip functions admit a linear parametric form:
\begin{align}
    \gamma_t(\bar{x}_t, \bar{\tau}_t;\psi_t) := \psi_t'\phi_t(\bar{x}_t, \bar{\tau}_t)  
\end{align}
for some known $\fdim$-dimensional feature vector maps $\phi_t$, satisfying $\phi_t(\bar{x}_t, (\bar{\tau}_{t-1}, 0))=0$ and such that for some true $\psi_t^*$, $\gamma_t(\cdot, \cdot;\psi_t^*)=\gamma_t(\cdot,\cdot)$.
\end{assumption}
Then we can identify $\psi^*$ by finding a parameter vector $\psi$ that satisfies the subset of the moment restrictions of the form:
\begin{align}
    \forall t\in [m]: \E[ H_t(\psi)\, \left(\phi_t(\bar{X}_t, \bar{T}_t) - \E[\phi_t(\bar{X}_t, \bar{T}_t)\mid \bar{X}_t, \bar{T}_{t-1}]\right)] = 0
\end{align}
Moreover, we can also subtract from $H_t(\psi)$, the conditional expectation $\E[H_t(\psi^*)\mid \bar{X}_t, \bar{T}_{t-1}]$ while maintaining the moment condition:
\begin{align}
    \forall t\in [m]: \E[ \left(H_t(\psi) - \E[H_t(\psi^*)\mid \bar{X}_t, \bar{T}_{t-1}]\right)\, \left(\phi_t(\bar{X}_t, \bar{T}_t) - \E[\phi_t(\bar{X}_t, \bar{T}_t)\mid \bar{X}_t, \bar{T}_{t-1}]\right)] = 0
\end{align}
This is the doubly robust moment condition proposed by \citep{robins1994correcting,Robins2004}, where it is shown that an estimator of $\psi$ based on this moment is correct if either the estimate of $\E[H_t(\psi^*)\mid \bar{X}_t, \bar{T}_{t-1}]$ or the estimate of $\E[\phi_t(\bar{X}_t, \bar{T}_t)\mid \bar{X}_t, \bar{T}_{t-1}]$ is correct. However, estimating $\E[H_t(\psi^*)\mid \bar{X}_t, \bar{T}_{t-1}]$, requires knowing the true parameter vector $\psi^*$, and since this is unknown, any feasible estimation approach must first construct preliminary estimates of $\psi^*$, which is computationally cumbersome and introduces another source of error. This issue has been discussed as one of the main points not to use the doubly robust correction in practice in $g$-estimation of structural nested models (see e.g. the discussion at the end of Section~6.1 of \citep{vansteelandt2014structural}). For a binary treatment and when the target policy is the all-zero policy \citep{hernan2010causal} (see Technical~Point~21.5) note that $\E[H_t(\psi^*)\mid \bar{X}_t, \bar{T}_t]$ can be estimated by regressing the outcome of the population that received zero subsequent treatment on the history. However, such a population can be quite small in practice and can have severe co-variate imbalances compared to the overall population. Moreover, this approach only applies to the case of a binary treatment and a static target policy.

\section{Our Approach: Neyman Orthogonal (Locally Robust) $G$-Estimation}

In this work, we show that we can achieve a Neyman orthogonal moment for identifying $\psi^*$, which is sufficient for robustness to biases stemming from machine learning models used to train the nuisance components, while avoiding the cumbersome part of estimating the nuisance $\E[H(\psi^*)\mid \bar{X}_t, \bar{T}_{t-1}]$. Moreover, our approach leads to a strongly convex loss for the parameters at each step of the recursive process, which is beneficial for finite sample guarantees and subsequently for generalizing it to linear models with parameter heterogeneity with respect to exogenous co-variates.

In particular, instead of subtracting $\E[H_t(\psi^*)\mid \bar{X}_t, \bar{T}_{t-1}]$, we subtract $\E[H_t(\psi)\mid \bar{X}_t, \bar{T}_{t-1}]$, i.e. the function that we subtract is dynamically dependent on the estimate of $\psi$. Even though this moment is not doubly robust, we will show that it remains \emph{locally robust} \citep{LRSP}, aka \emph{Neyman orthogonal} \citep{chernozhukov2018double}. To define our moment restrictions and connect them to the main results of the paper we first define several convenient random variables: for any $j\in [m]$, let
\begin{align}\label{eqn:q-j}
Q_j :=\phi_j(\bar{X}_j, \bar{T}_j) - \phi_j(\bar{X}_j, (\bar{T}_{j-1}, \pol(\bar{X}_j, \bar{T}_{j-1}))
\end{align}
and for any $1\leq t\leq j$, let
\begin{align}
    \tilde{T}_{j, t} :=~& Q_j - \E[Q_j\mid \bar{X}_t, \bar{T}_{t-1}]\\
    \tilde{Y}_{t} :=~& Y - \E[Y\mid \bar{X}_t, \bar{T}_{t-1}]
\end{align}
Then, note that for any $\psi$:
\begin{align}
    H_t(\psi) - \E[H_t(\psi)\mid \bar{X}_t, \bar{T}_{t-1}] =~& \tilde{Y}_{t} - \sum_{j=t}^m \psi_j'\tilde{T}_{j,t}
\end{align}
Moreover, note that since the second term in $Q_t$ only depends on the conditioning set $\bar{X}_t, \bar{T}_{t-1}$, then:
\begin{align}
\tilde{T}_{t,t}=\phi_t(\bar{X}_t, \bar{T}_t) - \E[\phi_t(\bar{X}_t, \bar{T}_t)\mid \bar{X}_t, \bar{T}_{t-1}].
\end{align}
Thus we conclude that the true parameter $\psi^*$ must be satisfying the moment restrictions:
\begin{align}
    \forall t\in [m]: \E\left[\left(\tilde{Y}_t - \sum_{j=t+1}^m \psi_j'\tilde{T}_{j,t} - \psi_j'\tilde{T}_{t,t}\right)\, \tilde{T}_{t,t}\right] = 0
\end{align}
These are exactly of the same form as the moment restrictions considered in the definition of the Dynamic DML algorithm and its analysis. Thus the results we presented so far directly extend to the estimation of structural nested mean models, with a linear parameterization of the blip functions, simply by using the different definition of the residual variables $\tilde{Y}_t$ and $\tilde{T}_{t,t}$ and letting $\psi_{t}=\theta_{m-t}$ in the definition of Algorithm~\ref{alg:sr} and in Theorems~\ref{thm:mse} and \ref{thm:normality}. 

Moreover, note that the nuisance models that are trained in the first stage have a larger conditioning set which includes all past history of states and treatments $\bar{X}_{t}, \bar{T}_{t-1}$. If we made further restrictions such that there was a "funnel state" $S_t$ at each period that summarizes the history and such that any dependence of the future to the past is going through that funnel state, then conditioning only on that state would have been sufficient. This is what we essentially did in the linear Markovian model. Moreover, the linear Markovian model with a static policy $\tau$, is a special case where the blip functions take the simple form: $\gamma_t(\bar{x}_t, \bar{\tau}_t)=\theta_{m-t}'\tau_t$ and are target policy independent. When the blip functions are target policy independent, then any target policy can be used to estimate the structural parameters $\psi^*$. In our main development, we essentially used the baseline zero policy as a target policy to estimate the structural parameters.

Thus our Dynamic DML algorithm extends to the estimation of the structural parameters in a structural nested mean model for any target dynamic policy $\pol$ and any user defined baseline policy $\bar{0}$ (referred to as a $(\pi,\bar{0})$-double regime structural nested mean model). It allows for the estimation of the nuisance functions with arbitrary machine learning algorithms, subject to a relatively slow mean-squared-error condition and reduces estimation to simple regression and classification oracles in the first stage, with only a simple linear system of equations in the second phase, which can also be solved in linear time in a recursive manner. 

For completeness, we present the generalization of Algorithm~\ref{alg:sr} to SNMMs in Algorithm~\ref{alg:snmms} and we re-state the main Theorems in this broader context.

\begin{algorithm}[htpb]
   \caption{Dynamic DML for Structural Nested Mean Models (SNMMs)}
   \label{alg:snmms}
\begin{algorithmic}
   \STATE {\bf Input:} A data set consisting of $n$ samples of $m$-length paths: $\left\{\left(X_1^i, T_1^i, \ldots, X_m^i, T_m^i, Y^i\right)\right\}_{i=1}^n$
   \STATE {\bf Input:} A target dynamic policy $\pol$
   \STATE {\bf Input:} $\fdim$-dimensional feature vector maps $\{\phi_t\}_{t=1}^m$, which parameterize the blip functions
   \STATE Randomly split the $n$ samples in $S, S'$
   \FOR{each $t \in \{1, \ldots, m\}$} 
       \STATE Regress $Y$ on $\bar{X}_{t}, \bar{T}_{t-1}$ using $S$ to learn estimate $\hat{q}_t$ of model
       $$q_{t}^*(\bar{x}_{t}, \bar{\tau}_{t-1}) := \E[Y \mid \bar{X}_{t}=\bar{x}_t, \bar{T}_{t-1}=\bar{\tau}_{t-1}]$$
        Calculate residuals for each sample $i\in S'$: $\tilde{Y}_{t}^i := Y^i - \hat{q}_{t}(\bar{X}_{t}^i, \bar{T}_{t-1}^i)$
       \STATE Vice versa use $S'$ to learn model $\hat{q}_t$ and calculate residuals on $S$.
       \FOR{each $j \in \{t, \ldots, m\}$}
            \STATE Let $Q_{j,t} := \phi_j(\bar{X}_j, \bar{T}_j) - \phi_j(\bar{X}_j, (\bar{T}_{j-1}, \pol(\bar{X}_j, \bar{T}_{j-1}))\,1\{j>t\}$
            \STATE Regress $Q_{j,t}$ on $\bar{X}_{t}, \bar{T}_{t-1}$ using $S$ to learn estimate $\hat{p}_{j, t}$ of model 
            $$p_{j, t}^*(\bar{x}_{t}, \bar{\tau}_{t-1}) := \E[Q_{j,t} \mid \bar{X}_{t}=\bar{x}_t, \bar{T}_{t-1}=\bar{\tau}_{t-1}]$$
            Calculate residuals for each sample $i\in S'$: $\tilde{T}_{j, t}^i := Q_{j,t}^i - p_{j, t}(\bar{X}_{t}^i, \bar{T}_{t-1}^i)$
            \STATE Vice versa use $S'$ to learn model $\hat{p}_{j, t}$ and calculate residuals on $S$.
       \ENDFOR
   \ENDFOR
   \STATE Using all the data $S\cup S'$
   \FOR{$t=m$ {\bfseries down to} 1}
        \STATE Regress $\bar{Y}_t^i:=\tilde{Y}_{t}^i - \sum_{j=t+1}^m \hat{\psi}_{j}' \tilde{T}_{j, t}^i$ on $\tilde{T}_{t, t}^i$ with oridinary least squares:
        \begin{equation}\label{eqn:square-loss-alg}
            \hat{\psi}_{t} := \argmin_{\psi_{t}\in \Psi_{t}} \frac{1}{n} \sum_{i=1}^n \left(\bar{Y}_t^i- \psi_{t}'\tilde{T}_{t,t}^{i}\right)^2
        \end{equation}
   \ENDFOR
   \STATE {\bf Return:} Structural parameter estimate $\hat{\psi} = (\hat{\psi}_1, \ldots, \hat{\psi}_m)$ of $\psi^*=(\psi_1^*, \ldots, \psi_m^*)$
\end{algorithmic}
\end{algorithm}

We present here the generalized version of the asymptotic normality theorem for SNMMs and we defer the finite sample guarantee to the next section where we will analyze a more general setting of heterogeneous dynamic effects and present a finite sample bound. The proof is identical to that of Theorem~\ref{thm:normality} and so we omit it.

\begin{theorem}[Asymptotic Normality]\label{thm:normality-snmm}
Let $\mcD_n$ be a sequence of families of data generating processes for a Structural Nested Mean Model and a target dynamic policy $\pol$, such that Assumptions~\ref{ass:cond-ex} and Assumption~\ref{ass:linear-snmm} are satisfied, for a constant feature map dimension $\fdim$ and such that all random variables and the ranges of all nuisance functions are bounded by a constant a.s.. Moreover, for any $D\in \mcD_n$:
\begin{align}
\E[\Cov(Q_{t,t}, Q_{t,t}\mid \bar{X}_t, \bar{T}_{t-1})]\succeq \lambda I.
\end{align}
and $\max_{O\in \{S,S'\}} \|\hat{h}_O-h^*\|_{2,2}=o_p(1)$ and
$\max_{O\in \{S, S'\}} \rho(\hat{h}_O) = o_p(n^{-1/2})$,
where $\hat{h}_O$ is the nuisance estimate in Algorithm~\ref{alg:snmms} used on the samples in $O$ and $\rho(\cdot)$ as defined in Theorem~\ref{thm:mse} with $M=\max_{t=1}^m \|\psi_t^*\|_2$..

Let $J$ denote a $(m\,\fdim)\times (m\,\fdim)$ upper triangular matrix consisting of $\fdim\times \fdim$ blocks, such that the $(t,j)$ block is defined as:
\begin{align}
J_{t,j}:=1\{t\leq j\}\E[Cov(Q_{t,t}, Q_{j,t}\mid \bar{X}_t, \bar{T}_{t-1})].
\end{align}
Moreover, let $\Sigma$ be a $(m\,\fdim)\times (m\,\fdim)$ block matrix with $\fdim\times \fdim$ blocks of the form:
\begin{align}
\Sigma_{t,j} := \E[\epsilon_t\, \epsilon_j\, (Q_{t,t} - \E[Q_{t,t}\mid \bar{X}_t, \bar{T}_{t-1}])\, (Q_{j,j} - \E[Q_{j,j}\mid \bar{X}_j, \bar{T}_{j-1}])']
\end{align}
where $\epsilon_t = H_t(\psi^*) - \E[H_t(\psi^*)\mid \bar{X}_t, \bar{T}_{t-1}]$ and let $V=J^{-1} \Sigma (J^{-1})'$. Then:
\begin{align}
    \sqrt{n} V^{-1/2} (\hat{\psi} - \psi^*) = -\frac{1}{\sqrt{n}}\sum_{i=1}^n V^{-1/2} J^{-1}\, m(Z^i; \psi^*, h^*) + o_p(1) \to_d N(0, I_{\fdim\cdot m})
\end{align}
where $m(\cdot) = (m_1(\cdot);\ldots; m_m(\cdot))$ and $m_t(Z; \psi^*, h^*) = \epsilon_t\, (Q_{t,t} - \E[Q_{t,t}\mid \bar{X}_t, \bar{T}_{t-1}])$.

Moreover, let $\hat{J}$ be an estimate of $J$ with the $(t,j)$ block being: $\hat{J}_{t,j} = 1\{j\geq t\} \frac{1}{n}\sum_i \tilde{T}_{t,t}^i (\tilde{T}_{j,t}^i)'$ and $\hat{\Sigma}$ be the estimate of the $\Sigma$, whose $(t,j)$ block is defined as:
\begin{align}
    \hat{\Sigma}_{t,j} = \frac{1}{n} \sum_{i=1}^n \left(\bar{Y}_t^{i} - \hat{\psi}_t'\tilde{T}_{t,t}^i\right)\, \left(\bar{Y}_j^{i} - \hat{\psi}_j'\tilde{T}_{j,j}^i\right)\, \tilde{T}_{t,t}^i\, (\tilde{T}_{j,j}^i)'
\end{align}
and let $\hat{V}=\hat{J}^{-1} \hat{\Sigma} (\hat{J}^{-1})'$.
Let $\Phi$ be the CDF of the standard normal distribution. Then for any vector $\nu \in \R^{d\,m}$, the confidence interval:
\begin{align}
    \CI := \left[\nu'\hat{\psi} \pm \Phi^{-1}(1-\alpha/2)\, \sqrt{\frac{\nu'\hat{V} \nu}{n}}\right]
\end{align}
is asymptotically uniformly valid:
    $\sup_{D\in \mcD_n} \left|\Pr_D(\nu'\psi^* \in \CI) - (1-\alpha)\right| \to 0$.
\end{theorem}

Using Equation~\ref{eqn:off-policy} of Lemma~\ref{lem:cntf-char-rho}, we have that asymptotic normality and asymptotic linearity of the structural parameters, derived in Theorem~\ref{thm:normality-snmm}, also implies asymptotic normality of the plug-in off-policy value estimate for the target dynamic policy $\pol$. 

\begin{corollary}[Dynamic Off-Policy Evaluation and Inference]
Under the assumptions and definitions of Theorem~\ref{thm:normality-snmm}, the following is an estimate of the off-policy value of the target policy $\pi$:
\begin{align}
    \hat{R}(\pi) := \frac{1}{n} \sum_{i=1}^n \left(Y^i + \hat{\psi}' Q^i\right)
\end{align}
where $Q=(Q_1,\ldots, Q_m)$ and $Q_t$ as defined in Equation~\eqref{eqn:q-j}. If we let $R^*(\pi)=\E[Y^{(\pi)}]$, then:
\begin{align}
    \frac{\sqrt{n}}{\sqrt{\gamma + \mu}} \left(\hat{R}(\pi) - R^*(\pi)\right) = \frac{1}{\sqrt{n}} \sum_{i=1}^n \frac{1}{\sqrt{\gamma + \mu}}f(Z^i;\psi^*,h^*) + o_p(1) \to_d N(0, 1)
\end{align}
where $f(Z; \psi^*, h^*) = Y + Q'\psi^* - \E[Y + Q'\psi^*] - \E[Q]'J^{-1}m(Z;\psi^*, h^*)$ and $\gamma=\Var(Y + Q'\psi^*)$ and $\mu=\E[Q]'V\E[Q]$. Moreover, if we let $\hat{\gamma}=\Var_n(Y + Q'\hat{\psi})$ and $\hat{\mu}=\E_n[Q]'\hat{V}\E_n[Q]$, then the confidence interval:
\begin{align}
    \CI := \left[\hat{R}(\pi) \pm \Phi^{-1}(1-\alpha/2)\, \sqrt{\frac{\hat{\gamma}+\hat{\mu}}{n}}\right]
\end{align}
is asymptotically uniformly valid:
    $\sup_{D\in \mcD_n} \left|\Pr_D(R^*(\pi) \in \CI) - (1-\alpha)\right| \to 0$.
\end{corollary}

\section{Heterogeneous Dynamic Effects}

We note that our moment condition that identifies each parameter $\psi_t$ is the derivative of a square loss and can be written as the solution to the square loss minimization problem defined in Equation~\eqref{eqn:square-loss-alg}.
This allows us to generalize the Dynamic DML to the case where we allow non-parametric heterogeneity in the parameters $\psi_t$, with respect to an exogenous fixed covariate vector of each sample, denoted as $X_0$, i.e. $\gamma_t(\bar{x}_t, \bar{\tau}_t)=\psi_t(x_0)'\phi(\bar{x}_t, \bar{\tau}_t)$, for a known feature map $\phi$ and unknown heterogeneous parameters $\psi$. Thus we can essentially generalize the $g$-estimation approach to SNMMs to allow for infinite or high dimensional parameters of the blip functions, as long as the input to these infinite dimensional parameters is fixed and not changing endogenously by the treatments (e.g. fixed characteristics of a unit). This can be achieved by simply minimizing recursively the square loss:
\begin{align}
    \hat{\psi_t} = \argmin_{\psi_t(\cdot) \in \Psi_t} \loss_{D,t}\left(\psi_t; \hat{\underline{\psi}}_{t+1}, \hat{h}\right) := \frac{1}{2}\E\left[\left(\tilde{Y}_t - \sum_{j=t+1}^m \hat{\psi}_j(X_0)'\tilde{T}_{j,t} - \psi_t(X_0)'\tilde{T}_{t,t}\right)^2\right]
\end{align}
over arbitrary function spaces $\Psi_t$ or by using any other machine learning techniques that achieve small excess risk with respect to the latter square loss problem (e.g. regularized least squares, early stopping, etc). In the latter, we denoted with $\hat{h}$ an estimate of the vector of all nuisance functions (e.g. estimated in the first stage of the Dynamic DML Algorithm) and $\hat{\underline{\psi}}_{t+1}=(\hat{\psi}_{t+1}, \ldots, \hat{\psi}_m)$ the estimates of the target structural parameters constructed in previous iterations of the recursion.

\begin{algorithm}[htpb]
   \caption{Dynamic RLearner for Structural Nested Mean Models (SNMMs)}
   \label{alg:hetero-sr}
\begin{algorithmic}
   \STATE {\bf Input:} A data set of $n$ samples of $m$-length paths: $\left\{\left(X_0^i, X_1^i, T_1^i, \ldots, X_m^i, T_m^i, Y^i\right)\right\}_{i=1}^n$
   \STATE {\bf Input:} A target dynamic policy $\pol$
   \STATE {\bf Input:} $\fdim$-dimensional feature vector maps $\{\phi_t\}_{t=1}^m$, which parameterize the blip functions
   \STATE {\bf Input:} Function spaces for the heterogeneous structural parameters $\{\Psi_t\}_{t=1}^m$
   \STATE Execute first stage identically to Algorithm~\ref{alg:snmms}
   \STATE Using all the data $S\cup S'$
   \FOR{$t=m$ {\bfseries down to} 1}
        \STATE Estimate heterogeneous structural function $\hat{\psi}_t$ by minimizing the following square loss:
        \begin{equation}
            \hat{\psi}_{t} := \argmin_{\psi_{t}\in \Psi_{t}^n} \frac{1}{n} \sum_{i=1}^n \left(\tilde{Y}_{t}- \sum_{j=t+1}^m \hat{\psi}_{j}(X_0^i)' \tilde{T}_{j, t} - \psi_{t}(X_0^i)'\tilde{T}_{t,t}^{i}\right)^2
        \end{equation}
   \ENDFOR
   \STATE {\bf Return:} Heterogeneous structural function estimates $\hat{\psi} = (\hat{\psi}_1, \ldots, \hat{\psi}_m)$ of $\psi^*=(\psi_1^*, \ldots, \psi_m^*)$
\end{algorithmic}
\end{algorithm}

Using techniques from the recently introduced orthogonal statistical learning framework \citep{foster2019orthogonal,chernozhukov2018plugin}, we show in the appendix that this estimation method provides mean-squared-error guarantees on the recovered heterogeneous parameters $\hat{\psi}_t$, that are robust to errors in the nuisance functions. This heterogeneous extension can also be viewed as an analogue of the RLearner meta-learner algorithm \citep{nie2017quasi}, generalized to the dynamic treatment regime setting. The formal description of the algorithm appears in Algorithm~\ref{alg:hetero-sr}, where we also allow for the function space over which we are optimizing the loss function $\loss_{D,t}$, denoted as $\Psi_{t}^n$ to not necessarily be equal to $\Psi_t$ (which we know containts $\psi_t^*$) and to be changing with the sample size.

\paragraph{Norm notation.} To state our main results we will introduce some norm notation. For any vector valued function $\psi$, taking as input a random variable $X$ and having output in $\R^r$, we will denote with:
\begin{align}
\|\psi\|_{u, v} := \E[\|\psi(X)\|_u^v]^{1/v} = \E\left[\left(\sum_{j=1}^r \psi_j(X)^u\right)^{v/u}\right]^{1/v}
\end{align}
for any $u,v>1$. If $\psi$ is a parameter vector, then we will overload notation and let $\|\psi\|_{u, v}=\|\psi\|_{u}=\left(\sum_{j=1}^{r}\psi_j^u\right)^{1/u}$. If $u$ or $v$ equals $\infty$, then this would designate the sup norm, e.g. $\|\psi\|_{\infty, v} := \E[\max_{j\in r} \psi_j(X)^v]^{1/v}$ and $\|\psi\|_{u,\infty}=\sup_{x\in \mcX} \|\psi(x)\|_u$. For any $u,v$, we will denote with $\bar{u},\bar{v}$ the parameters that correspond to the dual norm, i.e. $1/u + 1/\bar{u}=1$ and similarly for $\bar{v}$. For any two functions $f, g$, taking as input random variables $X,Y$ we will use the shorthand notation: $\|f\circ g\|_{u,v} = \E[\|f(X)\|_u^v \cdot \|g(Y)\|_{u}^v]^{1/v}$. For an $n\times m$ matrix $A$, we will use the matrix norms:
\begin{align}
\|A\|_{u,v} := \sup_{x\in \R^m} \frac{\|A x\|_{u}}{\|x\|_{v}}
\end{align}
and for $u=v=2$, we denote with $\|A\|_{op}=\|A\|_{2,2}$, the spectral or operator norm of $A$.

\paragraph{Algorithm agnostic robustness to nuisance.} We first prove a general bound on the estimation that is independent of the estimation process that is run in the second stage of Algorithm~\ref{alg:hetero-sr}. This lemma will be useful in multiple subsequent theorems in this and subsequent sections.

\begin{lemma}[Algorithm-agnostic analysis]\label{lem:general}
Let $C_{t,j} := \Cov(Q_{t,t}, Q_{j,t} \mid \bar{X}_t, \bar{T}_{t-1})$ and suppose that for all $t\in [m]$ and for $\lambda>0$, for all $x_0\in \mcX_0$:
\begin{align}
\E[C_{t,t} \mid X_0=x_0] \succeq \lambda I \tag{positivity}    
\end{align}
Consider any estimation algorithm that produces a vector of estimates $\hat{\psi}=(\psi_1,\ldots, \psi_m)$, with small plug-in excess risk at each stage $t$ of the final stage process of Algorithm~\ref{alg:hetero-sr}, with respect to any solution $\psi^{*,n}$, at some nuisance estimate $\hat{h}$, i.e.,
\begin{talign}
\loss_{D,t}(\hat{\psi}_t; \hat{\underline{\psi}}_{t+1}, \hat{h}) - \loss_{D,t}(\psi_t^{*,n}; \hat{\underline{\psi}}_{t+1}, \hat{h}) \leq \epsilon(\psi_t^{*,n}, \hat{\underline{\psi}}_{t}, \hat{h}).
\end{talign}
Moreover, for any function estimates $\hat{f}, \hat{g}$, with true values $f^*, g^*$ and inputs $Z_f, Z_g$, consider the random matrix:
\begin{align}
    \Delta(\hat{f}, \hat{g}) := \E[(\hat{f}(Z_f) - f^*(Z_f))\, (\hat{g}(Z_g) - g^*(Z_g))'\mid X_0]
\end{align}
and let $\nu_t := \hat{\psi}_t - \psi_t^{*,n}$ and $\nu_t^* := \hat{\psi}_t - \psi_t^{*}$ and $b_t := \psi_t^{*,n} - \psi_t^*$ and:\footnote{We note that by $\|\|\Delta(f,g)\|_{u,v}\|\|_w$, we denote the quantity where we first take the matrix norm of the random matrix $\Delta(f,g)$, which is then a random variable that depends on $X_0$ and then take the $\ell_{w}$ norm of this random variable.}
\begin{align}
    \rho_{t,u,v}(\hat{h}) := \left\|\|\Delta(\hat{p}_{t}, \hat{q}_t)\|_{\bar{u},u}\right\|_{\bar{v}}
    + 3 M \sum_{j=t}^m  \left\|\|\Delta(\hat{p}_{t}, \hat{p}_{j,t})\|_{\bar{u},u}\right\|_{\bar{v}}
\end{align}
Then
\begin{align}
    \frac{\lambda}{2}\|\nu_t\|_{2,2}^2 - \frac{\sigma}{4} \|\nu_t\|_{u,v}^2 \leq~&  \epsilon(\psi_t^{*,n}, \hat{\underline{\psi}}_{t}, \hat{h}) + \frac{4 c_{t,t} }{\sigma} \|b_t\|_{u,\bar{v}} + \frac{4}{\sigma} \rho_{t,u,v}(\hat{h})^2 + \frac{2}{\sigma} \left(\sum_{j=t+1}^m c_{t,j}\, \|\nu_j^*\|_{u,\bar{v}}\right)^2
\end{align}
with $c_{t,j} := \sup_{x_0\in \mcX_0} \|\E[C_{t,j} \mid X_0=x_0]\|_{\bar{u}, u}$ and $M:=\max_{t\in [m], \psi_t\in \Psi_t^n} \max\left\{\|\psi_t^*\|_{u,\infty}, \|\psi_t\|_{u, \infty}\right\}$.
\end{lemma}

\paragraph{MSE for plug-in empirical risk minimization.} The main result of this section applies the latter lemma with $u=v=2$, $\sigma=\lambda$ and with $\psi_t^{*,n}=\arginf_{\psi_t\in \Psi_t^n} \|\psi_t - \psi_t^*\|_{2,2}$, so as to get a guarantee on the mean-squared-error of the heterogeneous structural parameters, as a function of the statistical complexity of the function spaces $\{\Psi_t^n\}_{t=1}^m$ and their bias with respect to the true heterogeneous effect parameters. 
Note that in this case, $c_{t,j} := \sup_{x_0\in \mcX_0} \|\E[C_{t,j} \mid X_0=x_0]\|_{op}$ and $M:=\sup_{t\in [m], x_0\in \mcX_0, \psi_t\in \Psi_t} \|\psi_t(x_0)\|_2$ and the convergence rates on pairs of nuisances will be with respect to the quantity $\|\|\Delta(\hat{f},\hat{g})\|_{op}\|_2$. 

We will combine the conclusion of Lemma~\ref{lem:general} with a plug-in excess risk bound for the case when the estimate $\hat{\psi}$ is produced via running empirical risk minimization in the second stage of the Heterogeneous Dynamic DML algorithm, as described in Algorithm~\ref{alg:hetero-sr}. To state the theorem we will use the notion of the critical radius, which is a measure of statistical complexity of a function space. For any function space $\mcF$, with functions having range in $[-1,1]$, we consider the localized Rademacher complexity as:
\begin{align}
    R_n(\mcF; \delta) = \E_{\epsilon_{1:n}, X_{1:n}}\left[\sup_{f\in \mcF: \|f\|_{2}\leq \delta} \frac{1}{n} \sum_{i=1}^n \epsilon_i f(X_i)\right]
\end{align}
We denote as the critical radius $\delta_n>0$ of $\mcF$, any solution to the inequality:
\begin{align}
R_n(\mcF; \delta) \leq \delta^2
\end{align}
For each function space $\Psi_t$, we denote with $\Psi_{t,i}$ the marginal function space corresponding to the $i$-th coordinate output of the functions in $\Psi_t$. Moreover, we denote with $\Psi_{t,i}-\psi_{t,i}^*=\{\psi_{t,i} - \psi_{t,i}^*: \psi_{t,i}\in \Psi_{t,i}\}$. Finally, we define the star hull of a function space as: $\text{star}(\mcF) := \{\tau f: f\in \mcF, \tau\in [0,1]\}$. The critical radius is a well-established concept in modern statistical learning theory and has been characterized for many function spaces. Moreover, for many function spaces it yields minimax optimal statistical learning rates. See \citep{wainwright_2019} for an overview.

\begin{theorem}[MSE Rate for Dynamic RLearner]\label{thm:hetero-thm}
Suppose that all random variables and functions are bounded. Let:
\begin{align}
    \textsc{bias}_n = \max_{t=1}^m \inf_{\psi_t \in \Psi_t^n} \|\psi_t - \psi_t^*\|_{2,2}
\end{align}
and let $\psi_t^{*,n} = \arginf_{\psi_t \in \Psi_t^n} \|\psi_t - \psi_t^*\|_{2,2}$. Let $\delta_n$ be an upper bound on the critical radius of the star hull of all the function spaces $\{\Psi_{t,i}^n - \psi_{t,i}^{*,n}\}_{t\in [m], i\in [r]}$ and $\delta_n=\Omega\left(\sqrt{\frac{r\log\log(n)}{n}}\right)$. Suppose that the quantities $m, \lambda, \{c_{t,j}\}_{1\leq j\leq t\leq m}, M$ (as defined in Lemma~\ref{lem:general}, for $u=v=2$) are constants independent of $n$ and that the nuisance estimates of both splits satisfy:
\begin{align}
    \max_{1\leq t \leq j \leq m} \left\{\E\left[\left\|\|\Delta(\hat{p}_{t,t}, \hat{p}_{j,t})\|_{op}\right\|_{2}^2\right], \E\left[\left\|\|\Delta(\hat{p}_{t,t}, \hat{q}_t)\|_{op}\right\|_{2}^2\right]\right\} =~& O(r^2 \delta_{n/2}^2 + \textsc{bias}_n^2)
\end{align}
Then the output of Algorithm~\ref{alg:hetero-sr} satisfies:
\begin{align}
    \max_{t\in [m]} \E[\|\hat{\psi}_t-\psi_t^*\|_{2,2}^2] = O\left( r^2\, \delta_{n/2}^2 + \textsc{bias}_n^2 \right)
\end{align}
\end{theorem}
Analogous results hold with high probability and exponential tail, if we make such exponential tail assumptions also on the guarantees provided by the nuisance functions. We omit them for succinctness.

\paragraph{Partial double robustness.} Note that most conditions on the nuisance functions have a doubly robust flavor, i.e. we need the product of two different nuisance function errors to be small. In particular, observe that by Jensen's inequality, for any $f, g$, we have that:\footnote{Since: $\E\left[\E[\|\hat{f}(Z_f) - f^*(Z_f)\|_2 \|\hat{g}(Z_g) - g^*(Z_g)\|_2\mid X_0]^2\right] \leq \E\left[\|\hat{f}(Z_f) - f^*(Z_f)\|_2^2 \|\hat{g}(Z_g) - g^*(Z_g)\|_2^2\right]$}
\begin{align}\label{eqn:holder-delta}
    \left\|\|\Delta(\hat{f}, \hat{g})\|_{op}\right\|_{2}^2 \leq \|\hat{f}-f^*\|_{2,4}^2\|\hat{g}-g^*\|_{2,4}^2
\end{align}
while when $X_0$ is the empty set (i.e. no heterogeneity), then we have:
\begin{align}
    \left\|\|\Delta(\hat{f}, \hat{g})\|_{op}\right\|_{2}^2\leq \|\hat{f}-f^*\|_{2,2}^2 \|\hat{g}-g^*\|_{2,2}^2
\end{align}
in which case the conditions on the nuisance estimates boil down to the same as those in Theorem~\ref{thm:mse} in the expository section.
Thus we need that either one or the other nuisance to be modeled and estimated accurately. The only exception is the functions $p_{t,t}$, which also need to satisfy that: $\|\hat{p}_{t,t} - p_{t,t}^*\|_{2,4}^4 = o_p(\delta_{n/2}^2)$. Thus one step ahead treatment propensities, need to be more accurately estimated than the remainder of the nuisance functions.

\paragraph{Alternative norm bounds.} We note that when the $\ell_{2,2}$ norm of $\hat{\psi}_t-\psi_t^*\in \Psi_t$ is lower bounded by some fraction of its $\ell_{2,\infty}$ norm, i.e. the sup norm over $X_0$ (e.g. if $\psi_t$ is a linear class and $\E[X_0\, X_0']\succeq \mu I$, a special case of which is when there is no heterogeneity), then invoking Lemma~\ref{lem:general} in the proof of Theorem~\ref{thm:hetero-thm} with $v=\infty$ instead of $v=2$, we can get a result of the form:
\begin{align}
    \max_{t\in [m]} \E[\|\hat{\psi}_t-\psi_t^*\|_{2,\infty}^2] = O\left( r^2\, \delta_{n}^2 \right)
\end{align}
subject to a weaker $\ell_{1}$ norm convergence for the products of the nuisance components:
\begin{align}
    \max_{1\leq t \leq j \leq m} \max\left\{\E\left[\left\|\|\Delta(\hat{p}_{t,t}, \hat{p}_{j,t})\|_{op}\right\|_{1}^2\right], \E\left[\left\|\|\Delta(\hat{p}_{t,t}, \hat{q}_t)\|_{op}\right\|_{1}^2\right]\right\} =~& O(r^2 \delta_{n/2}^2)
\end{align}
By applying a Holder inequality,\footnote{Since: $\left\|\|\Delta(\hat{p}_{t,t}, \hat{p}_{j,t})\|_{op}\right\|_{1}\leq \E\left[\|\hat{f}(Z_f) - f^*(Z_f)\|_2 \|\hat{g}(Z_g) - g^*(Z_g)\|_2\right]\leq \|\hat{f}-f^*\|_{2,2}\, \|\hat{g}-g^*\|_{2,2}$} the latter is satisfied whenever:
\begin{align}\label{eqn:suff-ell22}
    \|\hat{p}_{t,t} - p_{t,t}^*\|_{2,2} \left(\|\hat{q}_t - q_t^*\|_{2,2} + \sum_{j=t+1}^m\|\hat{p}_{j,t} - p_{j,t}^*\|_{2, 2}\right) = o_p(r\, \delta_{n/2})
\end{align}
Recovering again qualitatively the same norm convergence conditions as in Theorem~\ref{thm:mse}. When $\Psi_t$ is a parametric class with a bounded domain of parameters, then $\delta_n=O\left(n^{-1/2}\right)$, in which case, the latter requirement is satisfied if each nuisance function $f\in \{\hat{p}_{j,t}, \hat{q}_t\}_{1\leq t\leq j\leq m}$, satisfies that $\|\hat{f}-f^*\|_{2,2}=o_p(n^{-1/4})$, recovering the typical conditions in the Neyman orthogonality literature with parametric target estimands \citep{chernozhukov2018double}. 

For more general function classes $\Psi_t$, where the $\ell_{2,2}$ norm of $\hat{\psi}_t$ is not related to its $\ell_{2,\infty}$ norm and without any further restrictions on the correlations of errors among the nuisance components in each of the product term conditions in Theorem~\ref{thm:hetero-thm}, then with a simple Holder inequality applied to the nuisance constraints, as in Equation~\eqref{eqn:holder-delta} we can still derive a sufficient condition of the same form as in Equation~\eqref{eqn:suff-ell22}, but with the $\ell_{2,2}$ norms replaced by the slightly stronger $\ell_{2,4}$ norms. Hence, it suffices that each nuisance function satisfies $\|\hat{f}-f^*\|_{2,4}=o_p(\sqrt{\delta_n})$, which for parametric classes would be $o_p(n^{-1/4})$. If errors in the nuisance components are un-correlated then an $\ell_{2,2}$ norm convergence would suffice for most nuisances, since each product of nuisance error terms can be upper bounded as:
\begin{align}
\|\|\Delta(\hat{f},\hat{g})\|_{op}\|_2^2\leq \E\left[\|\hat{f}(Z_f)-f^*(Z_f)\|_2^2\, \|\hat{g}(Z_g)-g^*(Z_g)\|_2^2\right] \lesssim \|\hat{f}-f^*\|_{2,2}^2\, \|\hat{g}-g^*\|_{2,2}^2    
\end{align}
with the only $\ell_{2,4}$ norm required for the nuisances $\{p_{t,t}\}_{t=1}^m$, i.e. the one step ahead observational propensity models.

\paragraph{Uniform consistency.} Looking at Lemma~\ref{lem:general}, we note that if the $\ell_{2,2}$ norm of $\hat{\psi}_t-\psi_t^{*,n}$ is lower bounded by some fraction of its $\ell_{2,\infty}$, then a uniform consistency result can be derived. Since both of these functions lie in $\Psi_t^n$, if we choose the classes $\Psi_t^n$, such that their elements satisfy this property and such that as $n$ grows, the function spaces $\Psi_t^n$ uniformly approximate $\Psi_t$, then we can achieve a uniform consistency theorem. Moreover, in this case we can invoke the lemma with $v=\infty$, which would only require the weaker norm guarantees on the nuisances. 

\begin{theorem}[Sieve-Based Uniform Consistency]\label{thm:uniform}
Suppose that all random variables and functions are bounded. For $i\in [\fdim]$ and $t\in [m]$, let $\Psi_{t,i}^n=\{x_0 \to \theta'\alpha_n(x_0): \theta\in \R^{d_n}, \|\theta\|_2\leq 1\}$, for some sequence of $d_n$-dimensional feature maps $\alpha_n$ and
let:
\begin{align}
    \textsc{bias}_{n,\infty} = \max_{t=1}^m \inf_{\psi_t \in \Psi_t^n} \|\psi_t - \psi_t^*\|_{2,\infty}
\end{align}
and assume that for some $\infty > \gamma_n > 0$ and some $\infty > M >0$:
\begin{align}
 \E[a_n(X_0)\, a_n(X_0)'] \succeq~& \gamma_n I &
 \sup_{x_0\in \mcX_0} \|a_n(x_0)\|_2 \leq~& M
\end{align}
Suppose that the quantities $m, \lambda, \{c_{t,j}\}_{1\leq j\leq t\leq m}, M$ (as defined in Lemma~\ref{lem:general}, for $u=2$ and $v=\infty$) are constants independent of $n$. Let $\delta_n = \sqrt{\frac{\max\{d_n, \log\log(n)\}}{n}}$. Suppose that the nuisance estimates of both splits satisfy:
\begin{align}
    \max_{1\leq t \leq j \leq m} \left\{\E\left[\left\|\|\Delta(\hat{p}_{t,t}, \hat{p}_{j,t})\|_{op}\right\|_{1}^2\right], \E\left[\left\|\|\Delta(\hat{p}_{t,t}, \hat{q}_t)\|_{op}\right\|_{1}^2\right]\right\} = O\left(\gamma_n^{-1} r^2\,\delta_{n}^2 + \textsc{bias}_{n,\infty}^2\right)
\end{align}
Then the output of Algorithm~\ref{alg:hetero-sr} satisfies:
\begin{align}
    \max_{t\in [m]} \E[\|\hat{\psi}_t-\psi_t^*\|_{2,\infty}^2] = O\left(\gamma_n^{-2m + 1} \left(r^2\, \frac{\max\{d_n, \log\log(n)\}}{\gamma_n\, n} + \textsc{bias}_{n,\infty}^2\right)\right)
\end{align}
\end{theorem}

\section{Blip Model Selection via High Dimensional Sparse Linear Blip Functions}

The results we have discussed so far assume that the linear feature map that parameterizes the blip functions is low dimensional, i.e. $r\ll n$. Observe for instance that Theorem~\ref{thm:hetero-thm} is vacuous when $r=\Omega(n)$. In this section we examine the case where $r\gg n$ and provide guarantees under sparsity conditions on the true structural parameters. For simplicity, we will not consider non-parametric heterogeneity of the sparse coefficients with respect to some initial state $X_0$, i.e. we consider blip functions of the form: $\gamma_t(\bar{X}_t, \bar{T}_t;\psi_t)=\psi_t'\phi(\bar{X}_t,\bar{T}_t)$, with $\psi_t\in \R^r$ and $r\gg n$. However, we note that here the high-dimensionality of the feature map already offers a lot of modelling flexibility and one could encode heterogeneity of the blip effect through the feature map. 

Apart from the explicit dependence on $r$, when the feature map is high dimensional then the $\ell_{2,2}$ norm of the errors of the nuisance functions in Theorem~\ref{thm:hetero-thm} can accumulate across their $r$-dimensional components. Instead, we would ideally only require a bound on the maximum error across the $r$ dimensions of each nuisance function. Then an exponential tail bound on the MSE of each coordinate would imply a bound on the maximum that scales only logarithmically with $r$. To achieve this we can invoke Lemma~\ref{lem:general} with $u=1, v=\infty$ and $\psi^{*,n}=\psi^*$ and $X_0$ an empty set. Then we get the following recursive bound on the MSE of the heterogeneous structural parameters:
\begin{align}
    \frac{\lambda}{2} \|\nu_t^*\|_{2}^2 - \frac{\sigma}{4} \|\nu_t^*\|_{1}^2\leq~& \epsilon(\psi_t^{*,n},\hat{\underline{\psi}}_{t}, \hat{h}) + \frac{4\rho_{t,1,\infty}(\hat{h})^2}{\sigma} + \frac{2}{\sigma} \left(\sum_{j=t+1}^m c_{t,j}\, \|\nu_j^*\|_{1}\right)^2\\
    \rho_{t,1,\infty}(\hat{h}) :=~& \|\Delta(\hat{p}_{t}, \hat{q}_t)\|_{\infty}
    + 3 M \sum_{j=t}^m  \|\Delta(\hat{p}_{t}, \hat{p}_{j,t})\|_{\infty}
\end{align}
where $c_{t,j} := \|\E[C_{t,j}]\|_{\infty}$ and $M:=\sup_{t\in [m], \psi_t\in \Psi_t} \|\psi_t\|_1$ and we used the short-hand norm notation $\|A\|_{\infty}=\max_{i,j} |A_{i,j}| = \|A\|_{\infty,1}$.

However, we see that we incur a dependency on the $\ell_{1}$ norm of the error $\hat{\psi}_t - \psi_t^*$. Thus we need to be able to relate the $\ell_{1}$ with the $\ell_{2}$ norm of the error of our estimate. This is a restricted cone property on our estimate. We will thus invoke sparsity assumptions on the true parameter $\psi_t^*$ and augment the empirical risk minimization step with an $\ell_1$ penalty on $\psi_t$. This enforces the estimate to be primarily supported on the $s$ relevant dimensions. Within such a restricted cone the $\ell_1$ and the $\ell_2$ norm are equivalent with a constant that depends only on the sparsity level and not the ambient dimension.

Furthermore, the explicit dependence in Theorem~\ref{thm:hetero-thm} also stems from the fact that we invoked a multi-dimensional contraction inequality, across the $r$ output dimensions of $\psi_t$. However, observe that the loss function depends on each $\psi_t$ only through single indices of the form $\psi_t'\tilde{T}_{j,t}$. Thus we could instead control the critical radius of these $m$ index spaces at each iteration of the recursion, which would avoid the explicit dependence on $r$. Together these insights yield the following theorem.

\begin{algorithm}[htpb]
   \caption{Sparse Linear Dynamic DML for Structural Nested Mean Models (SNMMs)}
   \label{alg:sparse-linear}
\begin{algorithmic}
   \STATE {\bf Input:} A data set of $n$ samples of $m$-length paths: $\left\{\left(X_0^i, X_1^i, T_1^i, \ldots, X_m^i, T_m^i, Y^i\right)\right\}_{i=1}^n$
   \STATE {\bf Input:} A target dynamic policy $\pol$
   \STATE {\bf Input:} $\fdim$-dimensional feature vector maps $\{\phi_t\}_{t=1}^m$, which parameterize the blip functions
   \STATE Execute first stage identically to Algorithm~\ref{alg:snmms}
   \STATE Using all the data $S\cup S'$
   \FOR{$t=m$ {\bfseries down to} 1}
        \STATE Estimate heterogeneous structural function $\hat{\psi}_t$ by minimizing the $\ell_1$-penalized square loss:
        \begin{equation}
            \hat{\psi}_{t} := \argmin_{\psi_{t}\in \R^r: \|\psi_t\|_1\leq M} \frac{1}{n} \sum_{i=1}^n \left(\tilde{Y}_{t}- \sum_{j=t+1}^m \hat{\psi}_{j}' \tilde{T}_{j, t} - \psi_{t}'\tilde{T}_{t,t}^{i}\right)^2 + \kappa_t \|\psi_t\|_1
        \end{equation}
   \ENDFOR
   \STATE {\bf Return:} Structural function estimates $\hat{\psi} = (\hat{\psi}_1, \ldots, \hat{\psi}_m)$ of $\psi^*=(\psi_1^*, \ldots, \psi_m^*)$
\end{algorithmic}
\end{algorithm}

\begin{theorem}[$\ell_2$-Error Rate for Sparse Linear Dynamic DML]\label{thm:sparse-linear}
Suppose that $\psi_t^*$ have only $s$ non-zero coefficients and that:
\begin{align}
\E[\Cov(Q_t,Q_t\mid \bar{X}_t,\bar{T}_{t-1})]\succeq \lambda I \tag{average positivity}    
\end{align}
Let $c_0, c_1, c_3, c_4$ be sufficiently large universal constants. Then there exists a sequence of regularization levels (see proof for construction) $\kappa_1,\ldots, \kappa_m$, such that for $n\geq c_2\frac{s^2 \log(m\, r/\zeta)}{\lambda}$, the estimate output by Algorithm~\ref{alg:sparse-linear}, satisfies w.p. $1-c_0\, m\, \zeta$:
\begin{align}
    \forall t\in [m]: \|\hat{\psi}_t-\psi_t^*\|_2 \leq \|\hat{\psi}_t-\psi_t^*\|_1\leq c_1 \left(\frac{s}{\lambda}  \left(\sqrt{\frac{\log(m\, r/\zeta)}{n}} + \epsilon_n\right) \frac{c_n^{m-t + 1} - 1}{c_n - 1}\right)
\end{align}
where:
\begin{align}
    \epsilon_n :=~& \max_{m\geq j\geq t\geq 1, i,i'\in [r]} \max\left\{\|\hat{p}_{t,i} - p_{t,i}^*\|_2\, \|\hat{q}_{t,i'} - q_{t,i'}^*\|_2, M\,m \|\hat{p}_{t,i} - p_{t,i}^*\|_2\, \|\hat{p}_{j,t,i'} - p_{j,t,i'}^*\|_2\right\}\\
     c_n :=~& c_4\,\frac{s}{\lambda} \max_{t\in [m]} \sum_{j=t+1}^m \left(\|\E[\Cov(Q_{t,t}, Q_{t,j}\mid \bar{X}_t, \bar{T}_{t-1})]\|_{\infty} + \sqrt{\frac{\log(m\, r/\zeta)}{n}}\right)
\end{align}
\end{theorem}

\section{Experimental Results}

We consider data drawn from the DGP presented in Equation~\eqref{eqn:maineq}, with a linear observational policy:
\begin{align}
p(T_{t-1}, X_t) := C\, T_{t-1} + D\cdot X_t,
\end{align}
with $X_0, T_0 = 0$ and $\epsilon_t, \zeta_t, \eta_t$ standard normal r.v.'s. (recall that $d$ is the number of treatments and $p$ the number of state variables). We consider the instance where: $A_{ij} = .5$, for all $i\in [p]$, $j\in [d]$, $B = .5\, I_p$, $C = .2\, I_d$, $D[:, 1:2] = .4$, $D[:, 3:p]= 0$, $\mu[1:2] = .8$. We consider settings where the effect is \emph{constant}, i.e. $\theta_0\in \R^d$ or \emph{heterogeneous}, where:
\begin{align}
\theta_0(X)=\theta_0 + \ldot{\beta_0}{X[S]}
\end{align}
for some known low dimensional subset $S$ of the states. 

We compare the dynamic DML to several benchmarks. The results are presented in Figures~\ref{fig:benchmarks}, comparing the estimates of the dynamic DML algorithm to a number of other alternatives on a single instance. 
They fall into two categories.  
In the ``static'' set of approaches, each of the contemporaneous and lag effects is estimated one at a time, either by direct regression or (static) DML.
So for example, to estimate the one period lag effect $\theta_1$, we would regress $Y_t$ on $T_{t-1}$, with controls $X$.
We consider direct regression with no controls (``no-ctrls''), direct and DML with controls from the inital period (i.e $X_{t-1}$, ``init-ctrls'' and ``init-ctrls-dml'') and direct and DML with controls from the same period as the outcome (i.e $X_t$, ``fin-ctrls'' and ``fin-ctrls-dml'').
As an alternative to all of these, we try a ``direct'' dynamic approach, where initially we estimate $\theta_0$ using a lasso regression of $Y_t$ with all the controls and past treatments, and return the coefficient on $T_t$, and then ``peel'' off the estimated effect as in the main text before running another Lasso regression of $Y_t - \theta_0 T_t$ on $T_{t-1}$ to get the first lag effect etc.  
So this approach incorporates the peeling effect but doesn't do any orthogonalization.

The point estimate for $\theta_0$, $\theta_1$ and $\theta_2$ are depicted in the three panels of Figure~\ref{fig:benchmarks} and the error bars correspond to the constructed confidence interval.
For all three, the dynamic DML is relatively close to the truth and the confidence interval contains the truth.
The remaining approaches are not, although for the contemporaneous effect the approaches with final period controls have similar performance - it is really in the lagged effects that the differences become most apparent. Subsequently we run multiple experiments to evaluate the performance of DynamicDML. In each setting, we run $1000$ Monte Carlo experiments, where each experiment draws $N = 500$ samples from the above DGP and then estimated the effects and lag effects based on our Dynamic DML algorithm. Figure~\ref{fig:n500} considers the case of two treatments, and shows that the algorithm performs well in terms of giving reasonable coverage guarantees - for a nominal 95\% coverage, actual coverage varies from 91\% to 94.5\%. The right panel shows that the average estimates are close to the truth.
In Figure~\ref{fig:n2000} we repeat the experiments with $N=2000$, and actual coverage is now tightly in the range 94\% to 95\%, and the average estimates remain relatively unbiased.
We also find qualitatively similar performance for the case of estimating heterogeneous treatment effects and policy values (see Figures\ref{fig:n500hetero} and \ref{fig:n2000hetero}).    

\subsection{Constant Treatment Effects}

\begin{figure}[H]
\centering
\includegraphics[width=320pt]{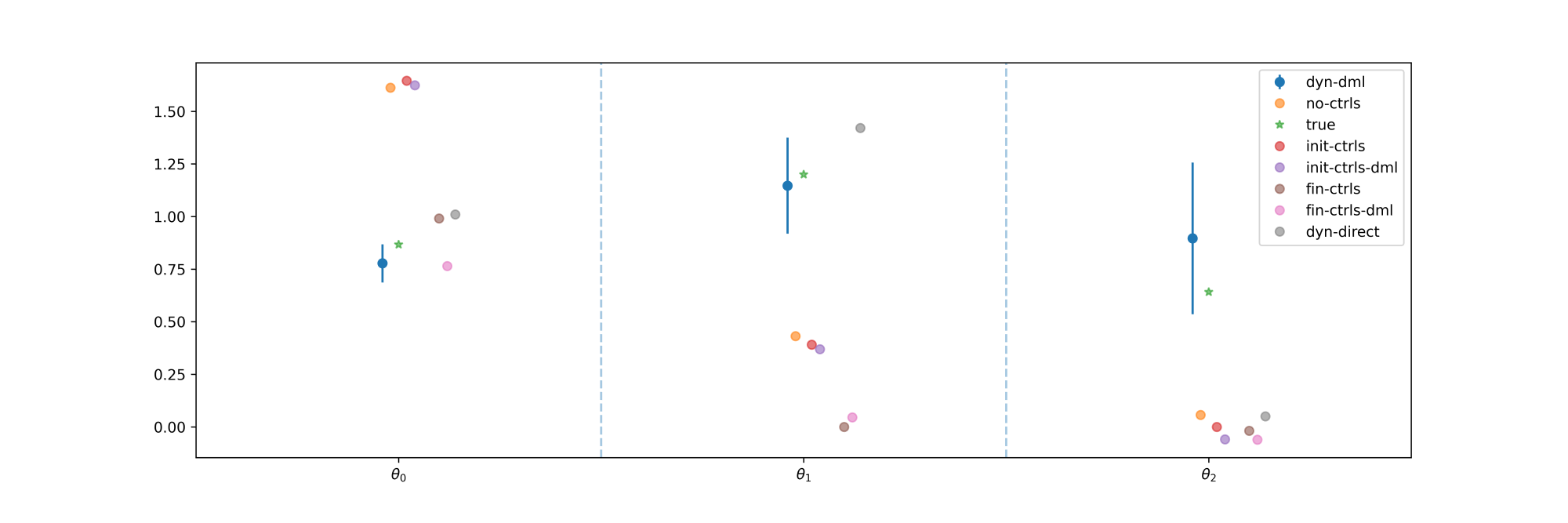}
\caption{Comparison of DynamicDML (with confidence intervals) with benchmarks on a single instance. $n=400$, $n_t=1$, $n_x=100$, $s=10$, $\sigma(\epsilon_t)=.5$, $\sigma(\zeta_t)=.5$, $\sigma(\eta_t)=.5$, $C=0$} 
\label{fig:benchmarks}
\end{figure}

\begin{figure}[H]
\centering
\subfigure[Coverage]{
\includegraphics[width=195pt]{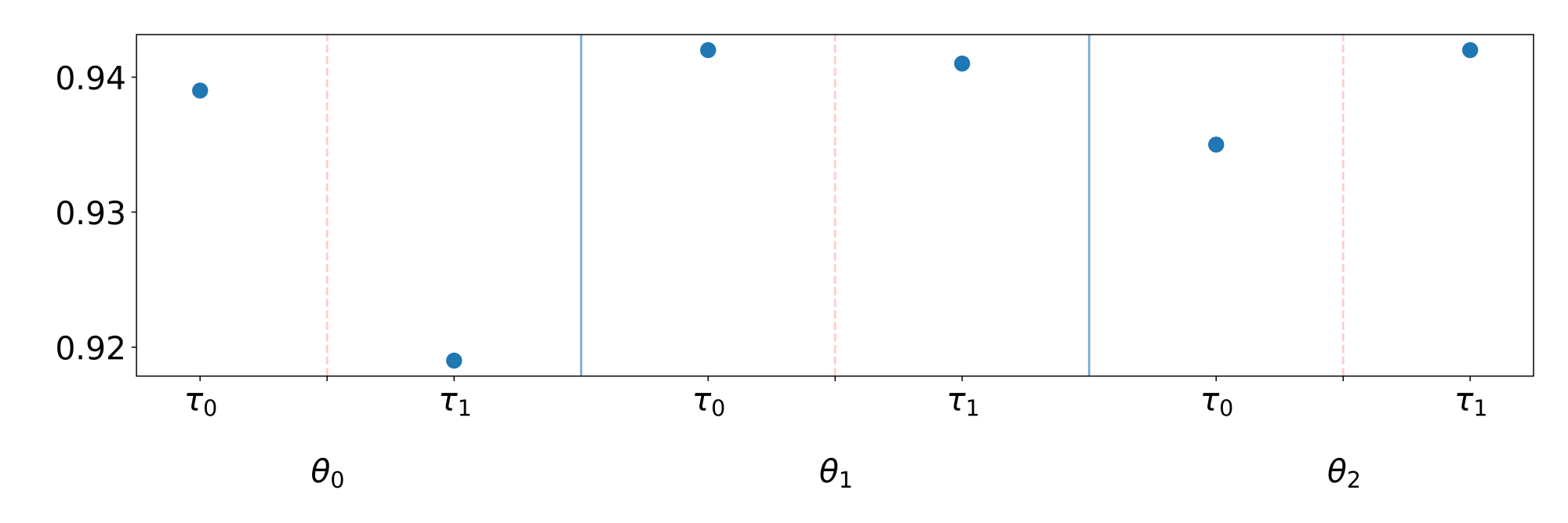}}
\subfigure[Point estimates]{
\includegraphics[width=195pt]{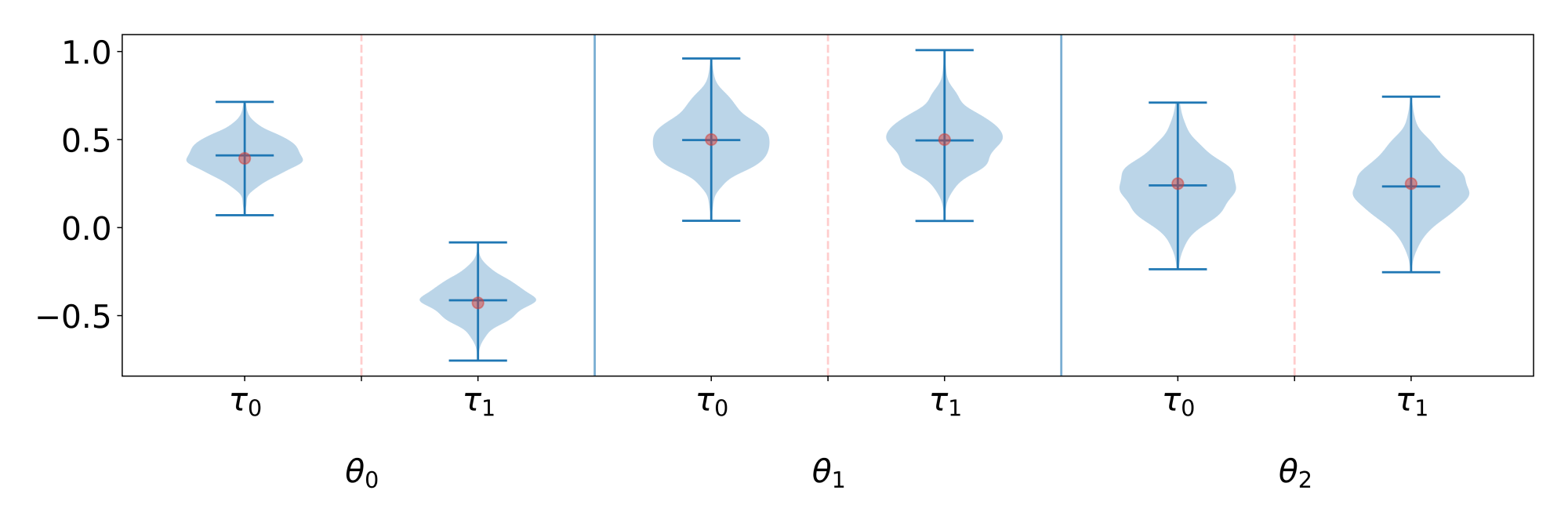}}
\caption{$n=500$, $n_t=2$, $n_x=450$, $s=2$, $\sigma(\epsilon_t)=1$, $\sigma(\zeta_t)=.5$, $\sigma(\eta_t)=1$}
\label{fig:n500}
\end{figure}

\begin{figure}[H]
\centering
\subfigure[Coverage]{
\includegraphics[width=195pt]{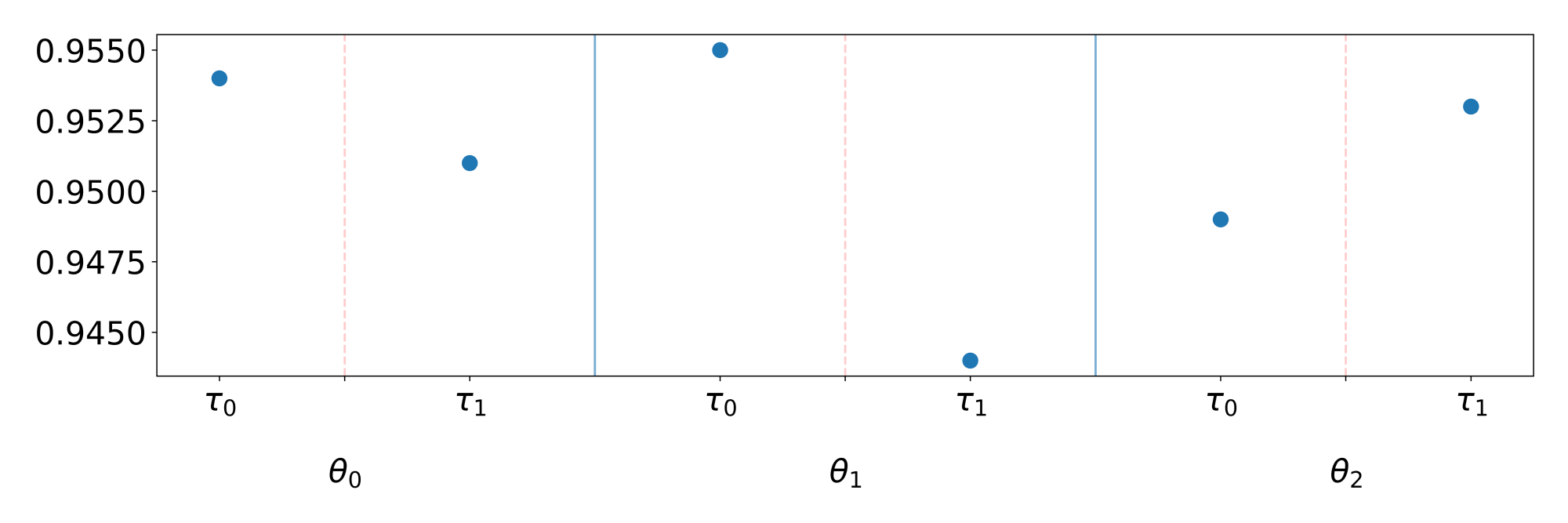}}
\subfigure[Point estimates]{
\includegraphics[width=195pt]{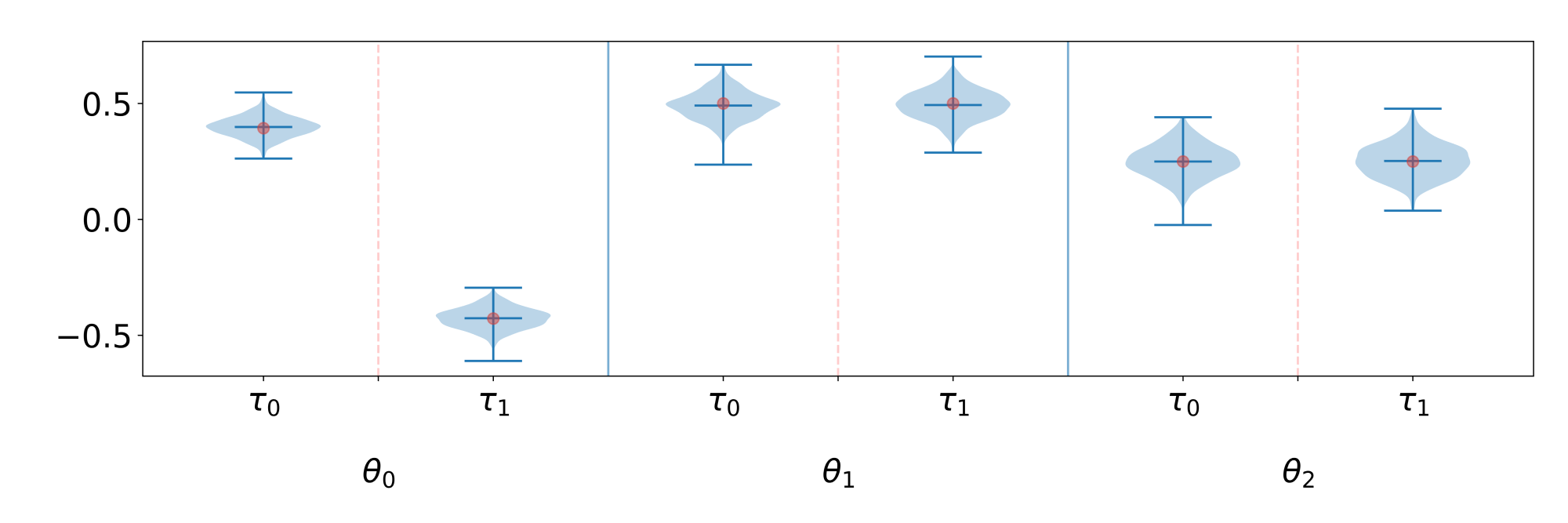}}
\caption{$n=2000$, $n_t=2$, $n_x=450$, $s=2$, $\sigma(\epsilon_t)=1$, $\sigma(\zeta_t)=.5$, $\sigma(\eta_t)=1$}
\label{fig:n2000}
\end{figure}

\subsection{Heterogeneous Treatment Effects}

\begin{figure}[H]
\centering
\subfigure[Coverage]{
\includegraphics[width=195pt]{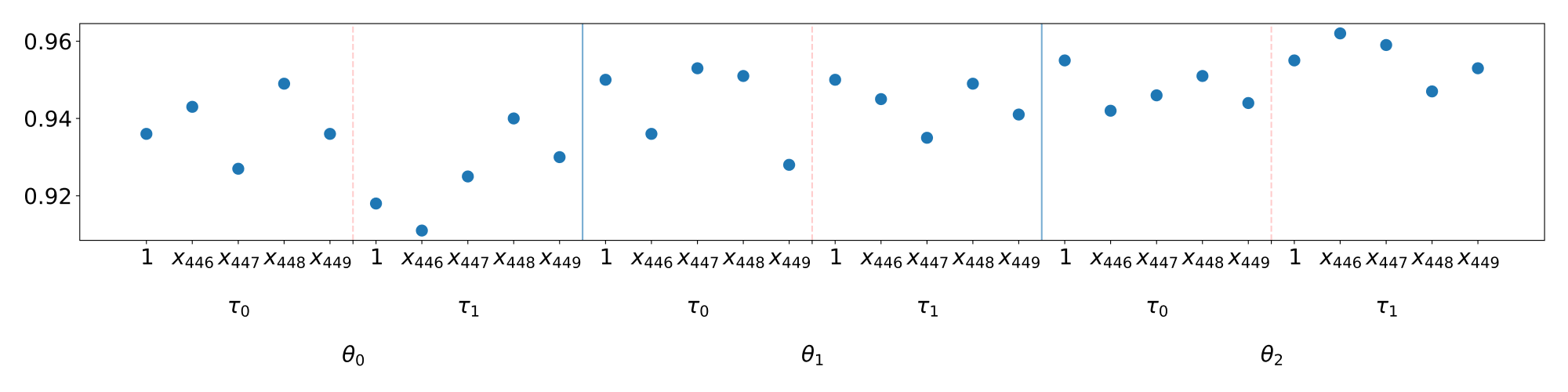}}
\subfigure[Point estimates]{
\includegraphics[width=195pt]{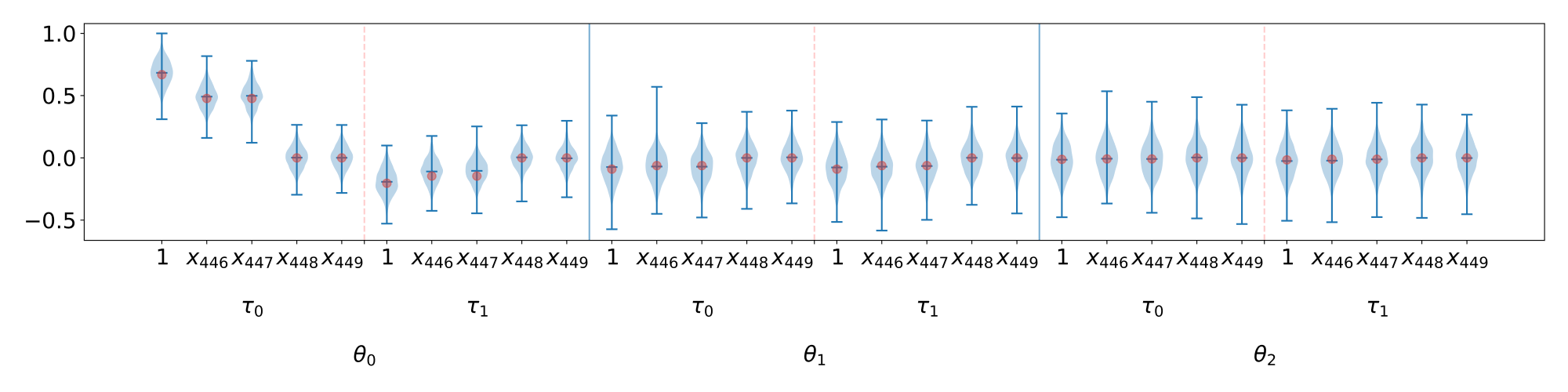}}
\caption{Heterogeneous effects: $n=500$, $n_t=2$, $n_x=450$, $s=2$, $\sigma(\epsilon_t)=1$, $\sigma(\zeta_t)=.5$, $\sigma(\eta_t)=1$}
\label{fig:n500hetero}
\end{figure}

\begin{figure}[H]
\centering
\subfigure[Coverage]{
\includegraphics[width=195pt]{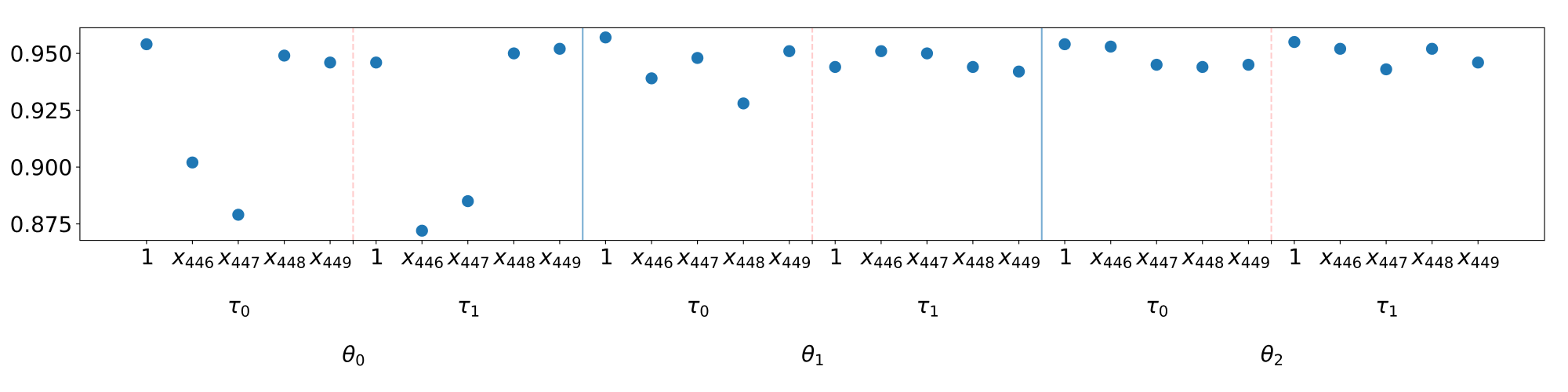}}
\subfigure[Point estimates]{
\includegraphics[width=195pt]{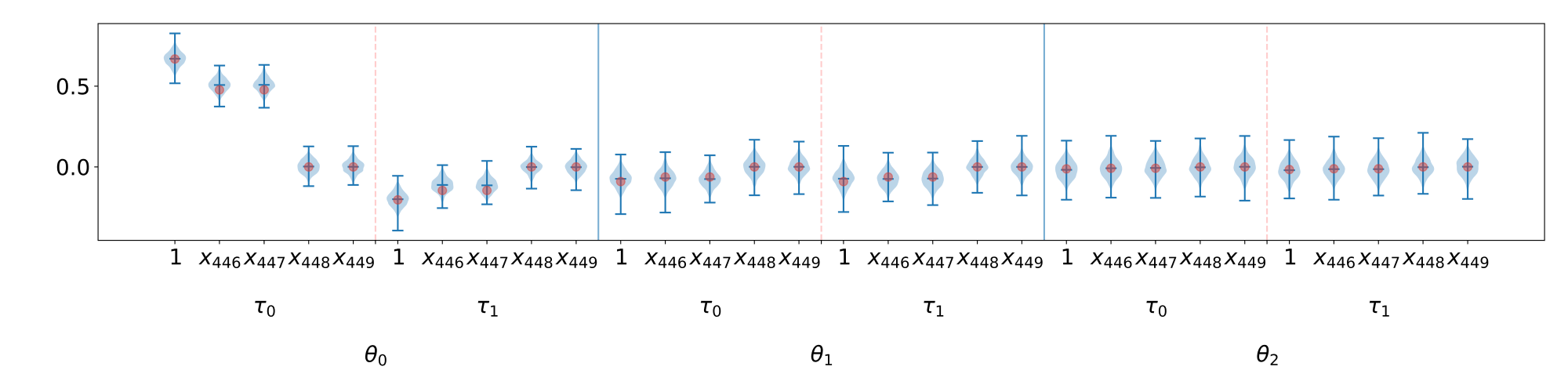}}
\caption{Heterogeneous effects: $n=2000$, $n_t=2$, $n_x=450$, $s=2$, $\sigma(\epsilon_t)=1$, $\sigma(\zeta_t)=.5$, $\sigma(\eta_t)=1$}
\label{fig:n2000hetero}
\end{figure}

\subsection{Counterfactual Policy Values}

\begin{figure}[H]
\centering
\subfigure[Coverage]{
\includegraphics[width=195pt]{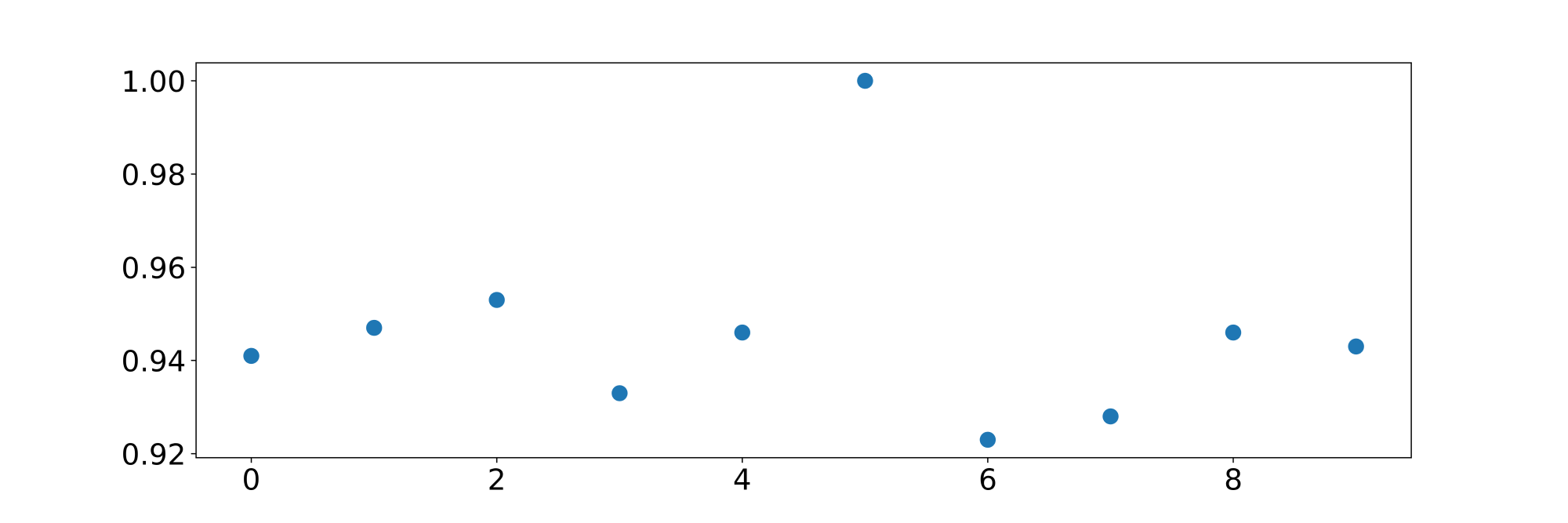}}
\subfigure[Point estimates]{
\includegraphics[width=195pt]{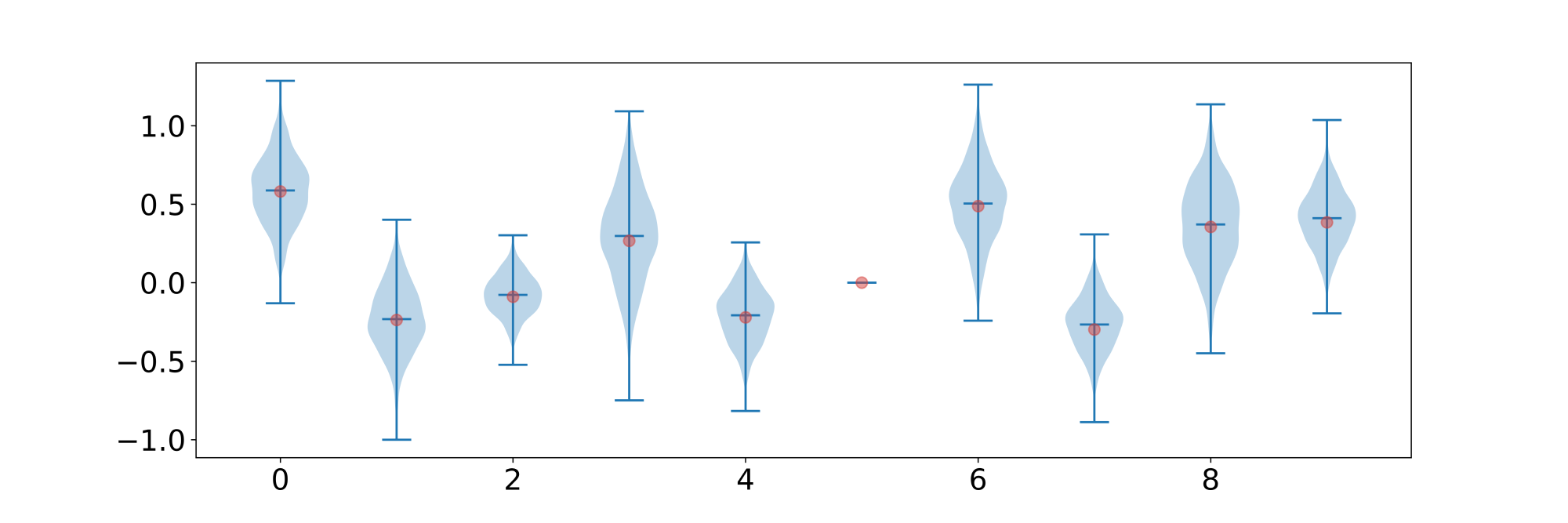}}
\caption{Counterfactual Policy Values for $10$ randomly chosen binary policies based on exogenous features: $n=500$, $n_t=2$, $n_x=450$, $s=2$, $\sigma(\epsilon_t)=1$, $\sigma(\zeta_t)=.5$, $\sigma(\eta_t)=1$}
\end{figure}

\begin{figure}[H]
\centering
\subfigure[Coverage]{
\includegraphics[width=195pt]{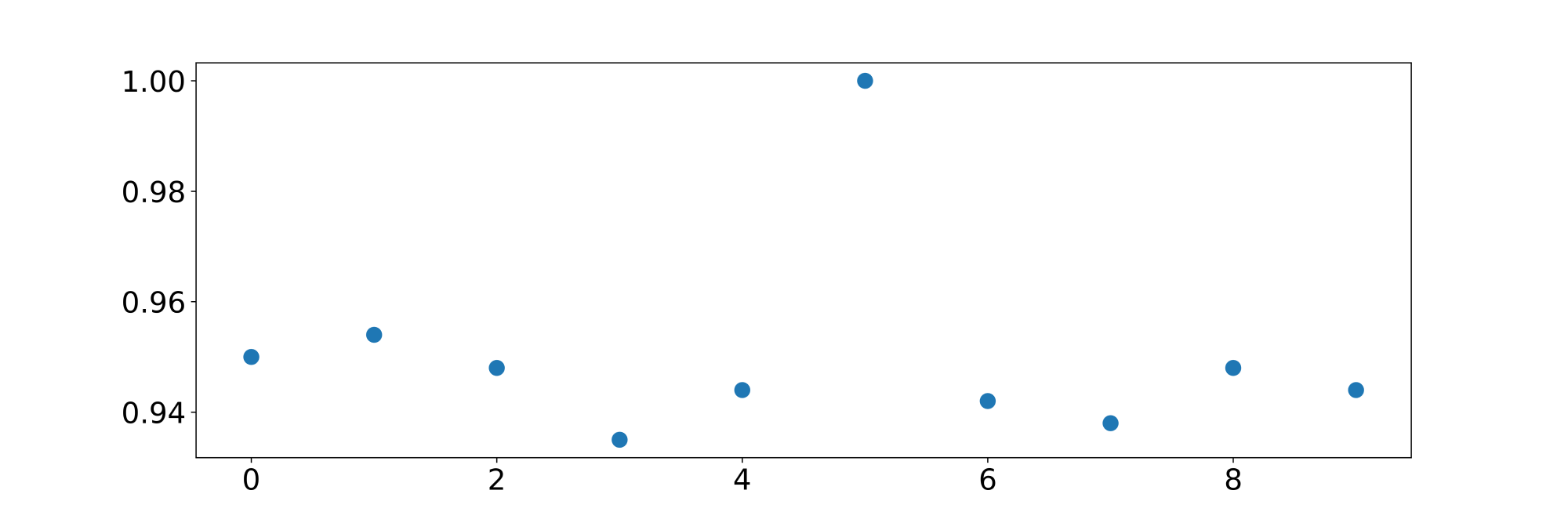}}
\subfigure[Point estimates]{
\includegraphics[width=195pt]{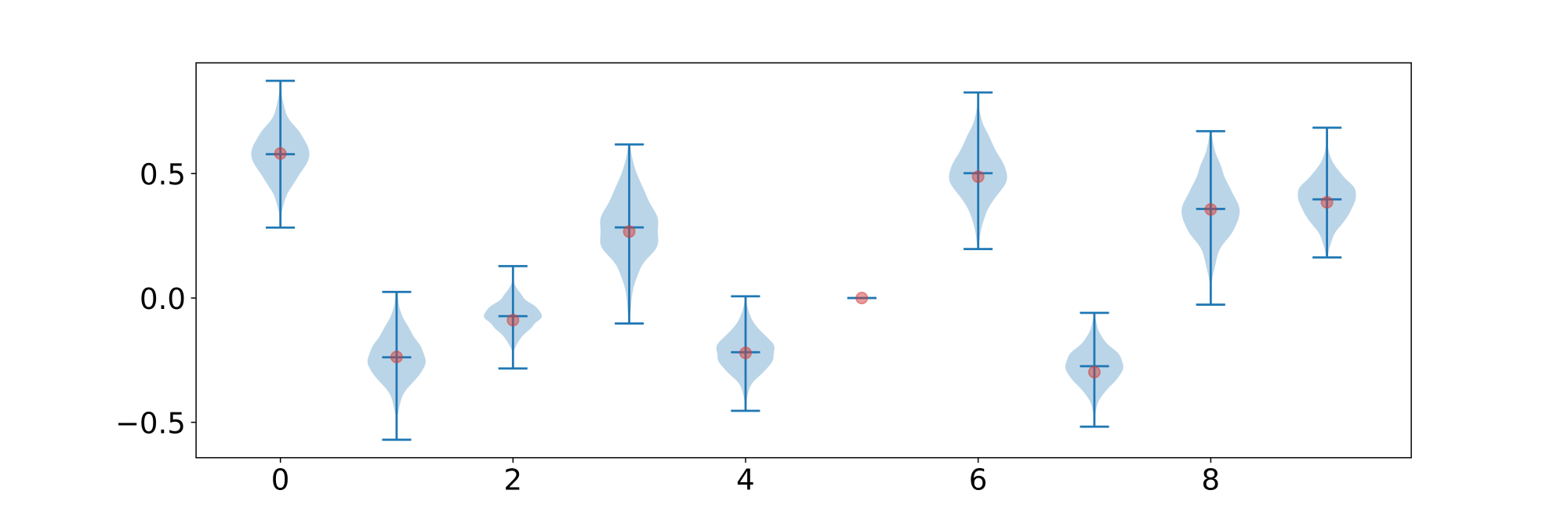}}
\caption{Counterfactual Policy Values for $10$ randomly chosen binary policies based on exogenous features: $n=2000$, $n_t=2$, $n_x=450$, $s=2$, $\sigma(\epsilon_t)=1$, $\sigma(\zeta_t)=.5$, $\sigma(\eta_t)=1$}
\end{figure}

\bibliography{refs}

\newpage

\newpage

\appendix

\etocdepthtag.toc{mtappendix}
\etocsettagdepth{mtchapter}{none}
\etocsettagdepth{mtappendix}{subsection}
\tableofcontents

\section{Omitted Proofs from Section~\ref{sec:expository}}

\subsection{Proof of Lemma~\ref{lem:char}}
\begin{proof}
 Observe that by repeatedly expanding the state space in terms of the previous state and action pairs we can write that under any non-state dependent sequence of interventions $\tau_1, \ldots, \tau_m$, the counterfactual value of the state at time $m$, $X_m^{(\tau)}$ can be written as:
 \begin{align}
     X_m^{(\tau)} =~& \sum_{\kappa=1}^{m-1} B^{\kappa-1} A\cdot \tau_{m-\kappa} + \delta_m 
 \end{align}
where $\delta_t$ is a linear function of the mean-zero exogenous random shocks at all periods. Thus under intervention $\tau$, the final outcome is distributed according to the random variable:
\begin{align}
    Y_m^{(\tau)} =~& \theta_0' \tau_m + \sum_{\kappa=1}^{m-1} \mu'B^{\kappa-1} A\cdot \tau_{m-\kappa} + \mu'\delta_m + \epsilon_t\\
    =~& \psi_m' \tau_m + \sum_{\kappa=1}^{m-1} \psi_{m-\kappa}' \tau_{m-\kappa} + \mu'\delta_m + \epsilon_t
\end{align}
Since $V(\tau)$ is the expected difference of this counterfactual final outcome and the counterfactual final outcome when $\tau=0$, and since $\delta_t$ is mean-zero conditional on any exogenous sequence of interventions that do not depend on the state or prior actions, we get the desired result.
\end{proof}

\subsection{Proof of Theorem~\ref{thm:id}}

\begin{proof}
For any $q\in \{0, \ldots, m-1\}$, by repeatedly unrolling the state $X_{m}$, $q-1$ times we have that:
\begin{align}
    Y_{m} =~& \theta_0' T_m + \sum_{\kappa=0}^{q} \mu'B^{\kappa-1} A T_{m-\kappa} + \mu'B^{q} X_{m-q} + \sum_{\kappa=0}^{q-1} \mu' B^{\kappa} \eta_{m-\kappa} + \epsilon_m\\
    =~& \psi_m' T_m + \sum_{\kappa=0}^{q} \psi_{m-\kappa}' T_{m-\kappa} + \mu'B^{q} X_{m-q}
    + \sum_{\kappa=0}^{q-1} \mu' B^{\kappa} \eta_{m-\kappa} + \epsilon_m\\
    =~& \psi_{m-q}'T_{m-q} + \sum_{j = m-q+1}^m \psi_j' T_j + \mu'B^{q} X_{m-q}
    + \sum_{\kappa=0}^{q-1} \mu' B^{\kappa} \eta_{m-\kappa} + \epsilon_m
\end{align}
Thus by re-arranging the equation and letting $t=m-q$, we have:
\begin{align}\label{eqn:recursive-peeling}
    \bar{Y}_{t} 
    =~& \psi_{t}' T_{t} + \mu'B^{m-t} X_{t} + \sum_{j=t+1}^{m} \mu' B^{m-j} \eta_{j} + \epsilon_m
\end{align}
Since for all $\epsilon_m$ and $\{\eta_j\}_{j=t+1}^m$ are subsequent random shocks to the variables $T_{t}, X_{t}$, we have that they are mean-zero conditional on $T_{t}, X_{t}$. 
Thus we have concluded that for any $t\in \{1, \ldots, m\}$:
\begin{align}\label{eqn:proof-cond-moment}
    \E[\bar{Y}_{t} - \psi_{t}' T_{t} - \mu'B^{m-q}\,X_{t} \mid T_{t}, X_{t}] = 0
\end{align}
which concludes the proof of the first part of the theorem.

For the second part, consider any $q$ and let $t=m-q$. Note that by taking an expectation of the above condition for with respect to $T_{t}$ conditional on $X_{t}$, we get that it implies
\begin{align}
    \E[\bar{Y}_{t}\mid X_t] = \psi_{t}' \E[T_{t}\mid X_t] + \mu'B^{m-q}\,X_{t}
\end{align}
Subtracting this from Equation~\eqref{eqn:proof-cond-moment}, we get:
\begin{align}
    \E[\bar{Y}_{t} - \E[\bar{Y}_t\mid X_t] - \psi_{t}' (T_{t} - \E[T_{t}\mid X_t]) \mid T_{t}, X_{t}] = 0
\end{align}
Since this holds for any $T_{t}, X_{t}$, the following vector of conditions also holds:
\begin{align}
    \E[(\bar{Y}_{t} - \E[\bar{Y}_t\mid X_t] - \psi_{t}' (T_{t} - \E[T_{t}\mid X_t]))\, (T_{t} - \E[T_{t}\mid X_t])] = 0
\end{align}
Noting that $\tilde{T}_t = T_{t} - \E[T_{t}\mid X_t]$, we get that the latter implies that:
\begin{align}
    \E[(\bar{Y}_{t} - \E[\bar{Y}_t\mid X_t] - \psi_{t}'\tilde{T}_t)\, \tilde{T}_t] = 0
\end{align}
If $J=\E[\tilde{T}_t\,\tilde{T}_t']=\E[\Cov(T_t, T_t\mid X_t)]$ is invertible, then this equation has the unique solution of:
\begin{align}
    \psi_{t} = J^{-1} \E[\tilde{T}_t\, (\bar{Y}_{t} - \E[\bar{Y}_t\mid X_t])]
\end{align}
\end{proof}

\subsection{Proof of Neyman Orthogonality}

First we define for any functional $L(f)$ the Frechet derivative as:
\begin{equation}
    D_f L(f)[\nu] = \frac{\partial}{\partial t} L(f + t\, \nu)\mid_{t=0} 
\end{equation}
Similarly, we can define higher order derivatives, denoted as $D_{g,f} L(f, g)[\mu,\nu]$.

\begin{lemma}[Neyman orthogonality]\label{lem:neyman-ortho-stylized}
The moment function:
\begin{equation}
    \mom_{D,t}(\psi; h) := \E\left[ \left(Y - q_t(X_t) - \sum_{j=t+1}^m \psi_j'(T_j - p_{j,t}(X_t)) - \psi_t'(T_t - p_{t,t}(X_t))\right)\, (T_t - p_{t,t}(X_t))\right]
\end{equation}
satisfies Neyman orthogonality with respect to $h$, i.e.:
\begin{equation}
    \forall h: D_h \mom_{D,t}(\psi^*; h^*)[h - h^*] = 0
\end{equation}
\end{lemma}
\begin{proof}
First note that by the definition of the Frechet derivative and the chain rule of differentiation:
\begin{align}
    D_h \mom_{D,t}(\psi^*; h^*)[h - h^*] = D_{q_t} \mom_{D,t}(\psi^*; h^*)[q_t - q_t^*] + \sum_{j=t}^m D_{p_{j,t}} \mom_{D,t}(\phi^*; h^*)[p_{j,t} - p_{j,t}^*]
\end{align}
We consider each component of the nuisance parameters separately. 

For $q_t$, we have:
\begin{align}
    D_{q_t} \mom_{D,t}(\psi^*; h^*)[q_t - q_t^*] = -\E[T_t - p_{t,t}^*(X_t)] = - \E[T_t - \E[T_t\mid X_t]] = 0
\end{align}
For $p_{j,t}$ with $j>t$, we have:
\begin{align}
    D_{p_{j,t}} \mom_{D,t}(\psi^*; h^*)[p_{j,t} - p_{j,t}^*] = \E[T_t - p_{t,t}^*(X_t)] = \E[T_t - \E[T_t\mid X_t]] = 0
\end{align}
For $p_{t,t}$, we have:
\begin{align}
    D_{p_{t,t}} \mom_{D,t}(\psi^*; h^*)[p_{t,t} - p_{t,t}^*] =~& \E[T_t - p_{t,t}^*(X_t)] - \E\left[Y - q_t^*(X_t) - \sum_{j=t}^m (\psi_j^*)'(T_j - p_{j,t}^*(X_t))\right] \\
    =~& - \E[Y - \E[Y|X_t]] + \sum_{j=t}^m (\psi_j^*)'\E[T_t - \E[T_t\mid X_t]] = 0
\end{align}
\end{proof}

\begin{lemma}\label{lem:second-order-frechet-stylized}
The second order Frechet derivative of the moment $\mom_{D,t}$ satisfies:
\begin{align}
   \forall h, \tilde{h}: D_{hh} \mom_{D,t}(\psi^*; \tilde{h})[h - h^*, h-h^*] :=~&  \E\left[\left(\hat{p}_{t,t}(X_t) -p_{t,t}^*(X_t)\right)\, \left(\hat{q}_t(X_t) - q_t^*(X_t)\right)\right] \\
   ~& - \sum_{j=t+1}^m \E\left[\left(\hat{p}_{t,t}(X_t) -p_{t,t}^*(X_t)\right)\, \left(\hat{p}_{j,t}(X_t) - p_{j,t}^*(X_t)\right)'\right]\, \psi_j^*\\
   ~& - 2\, \E\left[\left(\hat{p}_{t,t}(X_t) -p_{t,t}^*(X_t)\right)\, \left(\hat{p}_{t,t}(X_t) - p_{t,t}^*(X_t)\right)'\right]\, \psi_t^*
\end{align}
\end{lemma}
\begin{proof}
Note that by the definition of the Frechet derivative and the chain rule of differentiation:
\begin{align}
    D_{hh} \mom_{D,t}(\psi^*; \tilde{h})[h - h^*, h-h^*] =~& \sum_{j=t}^{m} D_{p_{j,t}, q_t} \mom_{D,t}(\psi^*; \tilde{h})[p_{j,t} - p_{j,t}^*, q_t-q_t^*]\\
    ~&~+ \sum_{j=t}^{m}\sum_{j'=j}^m D_{p_{j,t}, p_{j',t}} \mom_{D,t}(\psi^*; \tilde{h})[p_{j,t} - p_{j,t}^*, p_{j',t}-p_{j',t}^*]
\end{align}
However, note that by the definition of $\mom_{D,t}$, it only contains quadratic nuisance terms of the form $p_{t,t}\,p_{j,t}$ for $j\geq t$ and $p_{t,t}\, q_t$. Thus we have that all other second derivative terms will be zero and we can write:
\begin{align}
    D_{hh} \mom_{D,t}(\psi^*; \tilde{h})[h - h^*, h-h^*] =~& D_{p_{t,t}, q_t} \mom_{D,t}(\psi^*; \tilde{h})[p_{t,t} - p_{t,t}^*, q_t-q_t^*]\\
    ~&~+ \sum_{j=t}^{m} D_{p_{t,t}, p_{j,t}} \mom_{D,t}(\psi^*; \tilde{h})[p_{t,t} - p_{t,t}^*, p_{j,t}-p_{j,t}^*]
\end{align}
By simple calculus each of these terms can be shown to take the form given in the lemma.
\end{proof}

\subsection{Proof of Theorem~\ref{thm:mse}}
\label{app:finite_sample}

\begin{proof}
Let $Z=(X_1, T_1, \ldots, X_m, T_m, Y_m)$ denote the random variables associated with a single sample series and let $h = (p, q)$ denote the union of all nuisance functions. Observe that if we let:
\begin{equation}
\ell_{t}(Z; \psi, h) = \frac{1}{2}\left(Y - q_t(X_t) - \sum_{j=t+1}^m \psi_j'(T_j - p_{j,t}(X_t)) - \psi_t'(T_t - p_{t,t}(X_t)\right)^2
\end{equation} 
and if we let:
\begin{equation}
m_t(Z; \psi, h) := \nabla_{\psi_t} \ell_{t}(Z; \psi, h),
\end{equation}
then we have that:
\begin{align}
m_t(Z;\psi, h) = \left(Y - q_t(X_t) - \sum_{j=t+1}^m \psi_j'(T_j - p_{j,t}(X_t)) - \psi_t'(T_t - p_{t,t}(X_t)\right)\, (T_t - p_{t,t}(X_t))
\end{align}
Finally, let $\E_S[\cdot]$ denote the empirical average over the samples in $S$ and let:
\begin{align}
    \loss_S(\psi, h) :=~& \E_S\left[\ell_t(Z; \psi, h)\right] &  \loss_D(\psi, h) :=~& \E\left[\ell(Z;\psi, h)\right]\\
    \mom_S(\psi, h) :=~& \E_S\left[m_t(Z; \psi, h)\right] & \mom_D(\psi, h) :=~& \E\left[m_t(Z; \psi, h)\right]
\end{align}
Moreover, let $\hat{h}_S$, be the nuisance estimate applied to samples in split $S$, i.e. the estimated trained on $S'$ and vice versa. 

\paragraph{Properties of $\hat{\psi}$ and $\psi^*$.} Note that since $\hat{\psi}_t$ is an optimal point of loss $\loss_n(\hat{\psi},\hat{h})$, i.e.:
\begin{equation}
    \hat{\psi}_t \in \argmin_{\psi_t\in \R^d: \|\psi_t\|_2\leq M} \frac{1}{2}\sum_{O\in \{S,S'\}} \loss_{O}((\psi_t, \hat{\psi}_{-t}), \hat{h}_O)
\end{equation}
Since $\|\psi_t^*\|_2\leq M$, $\hat{\psi}_t$ must also satisfy the first-order optimality condition:
\begin{equation}\label{eqn:cross-fit-opt-hat}
    \frac{1}{2}\sum_{O\in \{S,S'\}} \mom_O(\hat{\psi}, \hat{h}_O)'(\psi_t^* - \hat{\psi}_t) \geq 0
\end{equation}
Note that by Theorem~\ref{thm:id} we have that the true parameters $\psi^*$ and true nuisance parameters $p, q$ satisfy:
\begin{align}
    \mom_D(\psi^*, h^*) := \E[m_t(Z; \psi^*, h^*] = 0
\end{align}
equivalently, $\psi^*$ is the minimum of the following loss:
\begin{align}
    \psi_t^* \in \argmin_{\psi_t\in \R^d: \|\psi_t\|_2\leq M} \loss_D((\psi_t, \psi_{-t}^*), h^*)
\end{align}

\paragraph{Strong convexity of $\loss_D$ for all $\psi, h$.} First we argue that by our assumption on the minimum eigenvalue of $\E[\Cov(T_t,T_t\mid X_t)]\succeq \lambda I$, we have that for any $\psi, h$ the loss function $\loss_D(\psi, h)$ is $\lambda$-strongly convex with respect to $\psi_t$. Note that for any $\nu\in \R^d$:
\begin{align}
    \nu'\nabla_{\psi_t,\psi_t} \loss_D(\psi, h)\nu =~& \nu'\E\left[(T_t - \hat{p}(X_t))\, (T_t - \hat{p}(X_t))'\right]\nu
    = \E\left[\left(\nu'(T_t - \hat{p}(X_t))\right)^2\right]\\
    \geq~& \E\left[\left(\nu'(T_t - \E[T_t\mid X_t])\right)^2\right]\\
    =~& \nu'\E\left[(T_t - \E[T_t\mid X_t])\, (T_t - \E[T_t\mid X_t])'\right]\nu\\
    =~& \nu'\E\left[\Cov(T_t, T_t\mid X_t)\right]\nu\\
    \geq~& \lambda \|\nu\|_2^2
\end{align}

\paragraph{Bounding $\ell_2$-error by moments.} By strong convexity we have for any $\hat{h}$:
\begin{align}
    \loss_D(\hat{\psi}, \hat{h}) - \loss_D((\psi_t^*, \hat{\psi}_{-t}), \hat{h}) \geq \mom_D((\psi_t^*, \hat{\psi}_{-t}), \hat{h})'\,(\hat{\psi}_t - \psi_t^*) + \frac{\lambda}{2} \|\hat{\psi}_t-\psi_t^*\|_2^2
\end{align}
Moreover, by convexity of $\loss_D$ we have:
\begin{align}
    \loss_D((\psi_t^*, \hat{\psi}_{-t}), \hat{h}) - \loss_D(\hat{\psi}, \hat{h})
    \geq \mom_D(\hat{\psi}, \hat{h})'\,(\psi_t^* - \hat{\psi}_t)
\end{align}
Combining the two, yields:
\begin{align}
    \frac{\lambda}{2} \|\hat{\psi}_t-\psi_t^*\|_2^2 \leq~& \left(\mom_D((\psi_t^*, \hat{\psi}_{-t}), \hat{h}) - \mom_D(\hat{\psi}, \hat{h})\right)'\,(\psi_t^* - \hat{\psi}_t)
\end{align}
Using the fact that $\mom_D(\psi^*, h^*)=0$, yields:
\begin{align}
    \frac{\lambda}{2} \|\hat{\psi}_t-\psi_t^*\|_2^2 \leq~& \left(\mom_D((\psi_t^*, \hat{\psi}_{-t}), \hat{h}) - \mom_D(\psi^*, h^*)  - \mom_D(\hat{\psi}, \hat{h})\right)'\,(\psi_t^* - \hat{\psi}_t)
\end{align}
Applying the latter for each $O\in \{S, S'\}$ and averaging, we get:
\begin{align}
    \frac{\lambda}{2} \|\hat{\psi}_t-\psi_t^*\|_2^2 \leq~& \frac{1}{2} \sum_{O\in \{S,S'\}} \left(\mom_D((\psi_t^*, \hat{\psi}_{-t}), \hat{h}_O) - \mom_D(\psi^*, h^*)  - \mom_D(\hat{\psi}, \hat{h}_O)\right)'\,(\psi_t^* - \hat{\psi}_t)
\end{align}
By first-order optimality of $\hat{\psi}_t$, from Equation~\eqref{eqn:cross-fit-opt-hat}, we have that:
\begin{align}
    \frac{\lambda}{2} \|\hat{\psi}_t-\psi_t^*\|_2^2 \leq~& \frac{1}{2} \sum_{O\in \{S,S'\}} \left(\mom_D((\psi_t^*, \hat{\psi}_{-t}), \hat{h}_O) - \mom_D(\psi^*, h^*) + \mom_O(\hat{\psi}, \hat{h}_O) - \mom_D(\hat{\psi}, \hat{h}_O)\right)'\,(\psi_t^* - \hat{\psi}_t) \\
    \leq~& \frac{1}{2} \sum_{O\in \{S,S'\}} \left\|\mom_D((\psi_t^*, \hat{\psi}_{-t}), \hat{h}_O) - \mom_D(\psi^*, h^*) + \mom_O(\hat{\psi}, \hat{h}_O) - \mom_D(\hat{\psi}, \hat{h}_O)\right\|_2\, \|\psi_t^* - \hat{\psi}_t\|_2
\end{align}
Thus we have:
\begin{align}
    \|\hat{\psi}_t-\psi_t^*\|_2 \leq \frac{1}{\lambda} \sum_{O\in \{S,S'\}}\left(\underbrace{\left\|\mom_D((\psi_t^*, \hat{\psi}_{-t}), \hat{h}_O) - \mom_D(\psi^*, h^*)\right\|_2}_{A_O} + \underbrace{\left\|\mom_O(\hat{\psi}, \hat{h}_O) - \mom_D(\hat{\psi}, \hat{h}_O)\right\|_2}_{B_O}\right)
\end{align}

\paragraph{Upper Bounding $A$ via Neyman orthogonality.} 
We will analyze $A_O$ for $O=S$ and let $\hat{h}\equiv \hat{h}_S$ and $A\equiv A_S$. The same analysis holds for $S'$. We further split this term by sequentially changing its inputs:
\begin{align}
    A = \underbrace{\mom_D((\psi_t^*, \hat{\psi}_{-t}), \hat{h}) - \mom_D(\psi^*, \hat{h})}_{A_1} + \underbrace{\mom_D(\psi^*, \hat{h}) - \mom_D(\psi^*, h^*)}_{A_2}
\end{align}
Moreover, by linearity of $\mom_D$ in $\psi_{-t}$, we note that:
\begin{align}
    A_1 =~& -\sum_{j=t+1}^{m} \E[\tilde{T}_{t,t}\, \tilde{T}_{j,t}'] (\hat{\psi}_j - \psi_j^*)\\
    =~& -\sum_{j=t+1}^{m} \left(\E[\Cov(T_t, T_j\mid X_t)] + \E[(\hat{p}_t(X_t) - p_t^*(X_t))\, (\hat{p}_{j,t}(X_t) - p_{j,t}^*(X_t))'] \right)\, (\hat{\psi}_j - \psi_j^*)
\end{align}
where we remind that $\Cov(T_t, T_j\mid X_t) = \E[(T_t - \E[T_t\mid X_t])\, (T_j - \E[T_j\mid X_t])'\mid X_t]$.
Thus if we let $c_t = \sum_{j=t+1}^{m} \|\E[\Cov(T_t, T_j\mid X_t)]\|_{op}$, and using the fact that $\|\hat{\psi}_j - \psi_j^*\|_2\leq 2M$, we have that:
\begin{align}
    \|A_1\|_2 \leq~& c_t \max_{j=t+1}^m \|\hat{\psi}_j - \psi_j^*\|_2 + 2M\, \sum_{j=t+1}^{m} \left\|\E[(\hat{p}_t(X_t) - p_t^*(X_t))\, (\hat{p}_{j,t}(X_t) - p_{j,t}^*(X_t))'] \right\|_{op}
\end{align}

For term $A_2$, we can apply a second order functional Taylor expansion around $h^*$. Using the Neyman orthogonality property of Lemma~\ref{lem:neyman-ortho-stylized} and the form of the second order functional derivatives of Lemma~\ref{lem:second-order-frechet-stylized}, we deduce that:
\begin{align}
    \|A_2\|_2 \leq~& \|D_{h} \mom_D(\psi^*, h^*)[\hat{h}-h^*] + D_{hh} \mom_D(\psi^*, \bar{h})[\hat{h}-h^*, \hat{h}-h^*]\|_2\\
    \leq~& \|\E\left[\left(\hat{p}_{t,t}(X_t) -p_{t,t}^*(X_t)\right)\, \left(\hat{q}_t(X_t) - q_t^*(X_t)\right)\right]\|_2 \\
   ~& + M\, \sum_{j=t+1}^m \|\E\left[\left(\hat{p}_{t,t}(X_t) -p_{t,t}^*(X_t)\right)\, \left(\hat{p}_{j,t}(X_t) - p_{j,t}^*(X_t)\right)'\right]\|_{op}\\
   ~& + 2\, M\, \|\E\left[\left(\hat{p}_{t,t}(X_t) -p_{t,t}^*(X_t)\right)\, \left(\hat{p}_{t,t}(X_t) - p_{t,t}^*(X_t)\right)'\right]\|_{op}
\end{align}
Moreover, note that for any random vectors $X, Y \in \R^d$:
\begin{align}
    \|\E[X\,Y']\|_{op} \leq \E[\|X\, Y'\|_{op}] \leq \E[\|X\|_2 \|Y\|_2] \leq \sqrt{\E[\|X\|_2^2]}\, \sqrt{\E[\|Y\|_2^2]} = \|X\|_{2,2}\, \|Y\|_{2,2}
\end{align}
Thus we can upper bound the $\|A\|_2$ as:
\begin{align}
    \|A\|_2 \leq~& c_t \max_{j=t+1}^m \|\hat{\psi}_j - \psi_j^*\|_2 + \underbrace{\|\hat{p}_{t,t}-p_{t,t}^*\|_{2,2} \left(\|\hat{q}_t - q_t^*\|_{2} + 4\, M\,\sum_{j=t}^m \|\hat{p}_{j, t} - p_{j,t}^*\|_2 \right)}_{\rho(\hat{h})}
\end{align}

Hence, so far we have that:
\begin{align}
    \|\hat{\psi}_t-\psi_t^*\|_2 \leq \frac{1}{\lambda} \left(2\, c_t \max_{j=t+1}^m \|\hat{\psi}_j - \psi_j^*\|_2 + \sum_{O\in \{S,S'\}} (\rho(\hat{h}_O) + B_O)\right)
\end{align}

\paragraph{Upper bounding $B$ via concentration.} We will analyze $B_O$ for $O=S$ and let $\hat{h}\equiv \hat{h}_S$ and $B\equiv B_S$. The same analysis holds for $S'$. Note that each coordinate $i$ of the $d$-dimensional vector $B$ is upper bounded by the supremum of a centered empirical process over $n/2$ samples:
\begin{align}
    |B_i| \leq~& \sup_{\psi\in \Psi} \left|\frac{2}{n} \sum_{i=1}^n \left(m_{t,i}(Z^i; \psi, \hat{h}) - \E[m_{t,i}(Z; \psi, \hat{h})]\right)\right|\\
    \leq~& \sup_{\psi\in \Psi, \delta\in \{-1, 1\}} \frac{2}{n} \sum_{i=1}^n \left(\delta\,m_{t,i}(Z^i; \psi, \hat{h}) - \E[\delta\,m_{t,i}(Z; \psi, \hat{h})]\right)
\end{align}
Moreover, if all random variables and the ranges of all nuisance functions are absolutely bounded by $H$ a.s., then each $m_{t,i}$ is absolutely bounded by $U := 2\, H\, m\, d$. Thus by \cite{bartlett2002rademacher} (see also Theorem~1 of \cite{maurer2016vector}), we can upper bound this w.p. $1-\zeta$ by:
\begin{align}
    |B_i| \leq~& \E_{\epsilon}\left[\sup_{\psi\in \Psi, \delta\in\{-1, 1\}} \frac{4}{n}\sum_{i=1}^n \epsilon_i \delta\, m_{t,i}(Z^i; \psi, \hat{h})\right] + U\,\sqrt{\frac{9\log(2/\zeta)}{n}}\\
    \leq~& \E_{\epsilon}\left[\sup_{\psi\in \Psi} \frac{8}{n}\sum_{i=1}^n \epsilon_i m_{t,i}(Z^i; \psi, \hat{h})\right] + U\,\sqrt{\frac{9\log(2/\zeta)}{n}}
\end{align}
Note that $m_{t,i}$ depends on $\psi$ via a vector of $m-t$ linear indices of the form: $\psi_{j}'\tilde{T}_{j,t}$ for $j\in \{t,\ldots, m\}$. Moreover, the function $m_{t,i}$ is $L:=2H\,\sqrt{m}$-Lipschitz with respect to these $m-t$ indices (and with respect to the $\ell_2$ norm definition of Lipschitzness). Thus by the multivariate contraction inequality of \cite{maurer2016vector}, we have that:
\begin{align}
    \E_{\epsilon}\left[\sup_{\psi\in \Psi} \frac{8}{n}\sum_{i=1}^n \epsilon_i m_{t,i}(Z^i; \psi, \hat{h})\right] \leq \sqrt{2}L\sum_{j=t}^m \E\left[\sup_{\psi_j: \|\psi_j\|_2\leq M} \frac{8}{n} \sum_{i=1}^n \epsilon_i\, \psi_j'\tilde{T}_{j,t}^{i}\right]
\end{align}
since $\|\tilde{T}_{j,t}\|_2\leq 2H\sqrt{d}$, a.s., by Lemma~26.10 of \cite{Shalev2014}, we have that:
\begin{align}
    \E\left[\sup_{\psi_j: \|\psi_j\|_2\leq M} \frac{8}{n} \sum_{i=1}^n \epsilon_i\, \psi_j'\tilde{T}_{j,t}^{i}\right] \leq 16 \frac{M\,H\,\sqrt{d}}{\sqrt{n}}
\end{align}
Thus overall we have:
\begin{align}
    |B_i| \leq 32 L\,m\, \frac{M\,H\,\sqrt{d}}{\sqrt{n}} + U\,\sqrt{\frac{9\log(2/\zeta)}{n}} = 64 H\,m^{3/2}\, \frac{M\,H\,\sqrt{d}}{\sqrt{n}} + 2H\,m\,d\,\sqrt{\frac{9\log(2/\zeta)}{n}}
\end{align}
Applying these for all $i$, we have that w.p. $1-\zeta$:
\begin{align}
    \|B\|_2 \leq~& \sqrt{d}\left(64 H\,m^{3/2}\, \frac{M\,H\,\sqrt{d}}{\sqrt{n}} + 2H\,m\,d\,\sqrt{\frac{9\log(2d/\zeta)}{n}}\right)\\
    \leq~& 2 d H m \left(32 H M \sqrt{\frac{d\, m}{n}} + \sqrt{\frac{9\,d\,\log(2d/\zeta)}{n}}\right) 
\end{align}

\paragraph{Combining bounds and concluding.} Let:
\begin{align}
    \mu :=~& \frac{2}{\lambda} \left(2 d H m \left(32 H M \sqrt{\frac{d\, m}{n}} + \sqrt{\frac{9\,d\,\log(2d/\zeta)}{n}}\right) + \max_{O\in \{S,S'\}} \rho(\hat{h}_O)\right)
\end{align}
Then overall we have:
\begin{equation}
\|\hat{\psi}_t-\psi_t^*\|_2 \leq \mu + \frac{2}{\lambda} c_t \max_{j> t} \|\hat{\psi}_{j} -\psi_{j}^*\|_2
\end{equation}
Further, let:
\begin{align}
    \kappa :=~& \max_{t=1}^m \frac{2\,c_t}{\lambda}
\end{align}
Then observe that:
\begin{equation}
     \|\hat{\psi}_{t} -\psi_{t}^*\|_2 \leq \mu + \kappa \, \max_{j> t} \|\hat{\psi}_{j} -\psi_{j}^*\|_2
\end{equation}
Then we have that: $\|\hat{\psi}_{m} -\psi_{m}^*\|_2 \leq \mu$ and by induction, if we assume that $\|\hat{\psi}_{j} -\psi_{j}^*\|_2\leq \mu  \sum_{\tau=0}^{m-j} \kappa^\tau$ for all $j>t$, then:
\begin{equation}
    \forall t\in [m]: \|\hat{\psi}_{t} -\psi_{t}^*\|_2 \leq \mu + \kappa\, \mu \sum_{\tau=0}^{m-t-1} \kappa^\tau = \mu \sum_{t=0}^{m-t}\kappa^\tau \leq \mu \frac{\kappa^{m - t + 1} - 1}{\kappa - 1}
\end{equation}
which concludes the proof of the theorem.
\end{proof}

\subsection{Proof of Theorem~\ref{thm:normality}}\label{app:normality}

\begin{proof} 
We introduce the following notation:
\begin{align}
m_t(Z;\psi, h) :=~& \left(Y - q_t(X_t) - \sum_{j=t+1}^m \psi_j'(T_j - p_{j,t}(X_t)) - \psi_t'(T_t - p_{t,t}(X_t)\right)\, (T_t - p_{t,t}(X_t)) \\
m(Z;\psi, h) :=~& (m_1(Z;\psi, h); \ldots; m_m(Z;\psi, h))
\end{align}
where by $(a;b)$ we denote the concatenation of two vectors. 
Finally, let $\E_S[\cdot]$ denote the empirical average over the samples in $S$ and let:
\begin{align}
    \mom_S(\psi, h) :=~& \E_S\left[m(Z; \psi, h)\right] & \mom_D(\psi, h) :=~& \E\left[m(Z; \psi, h)\right]
\end{align}
Observe that the estimator from Equation~\eqref{eqn:z-estimator-alg} can be equivalently viewed as the solution to an cross-fitted plug-in empirical version of the following vector of moment conditions: 
\begin{align}
    \frac{1}{2}\sum_{O\in \{S,S'\}} \mom_S(\hat{\psi}, \hat{h}_O) = 0
\end{align}
where each $\hat{h}_O$ is trained on samples outside of set $O$. Moreover, the true parameter $\psi^*$ satisfies the population moment conditions at the true nuisance parameters:
\begin{align}
    \mom_D(\psi^*, h^*) = 0
\end{align}
Furthermore, by Lemma~\ref{lem:neyman-ortho-stylized} the moment vector $\mom_D(\psi, h)$ satisfies the property of Neyman orthogonality with respect to $h$. Moreover, by Lemma~\ref{lem:second-order-frechet-stylized}, the second order term of $\mom_D(\psi^*, \hat{h})$ in second-order Taylor expansion around $h^*$, satisfies:
\begin{align}
    \|D_{hh} \mom_{D,t}(\psi^*; \tilde{h})[\hat{h} - h^*, \hat{h}-h^*]\|_{2} \leq~& \|\hat{p}_{t,t}-p_{t,t}^*\|_{2,2} \left(\|\hat{q}_t - q_t^*\|_{2} + 4\, M\,\sum_{j=t}^m \|\hat{p}_{j, t} - p_{j,t}^*\|_2 \right)\\
     =:~& \rho(\hat{h})
\end{align}

Moreover, the Jacobian $J$ of the moment vector $\mom_D$ at the true values $\psi^*, h^*$ is a block upper triangular matrix whose block values are of the form:
\begin{align}
\forall 1\leq t\leq j \leq m: J_{t, j} = \E[\Cov(T_t, T_j\mid X_t)]
\end{align}
Thus its diagonal block values take the form $J_{t, t}=\E[\Cov(T_t, T_t\mid X_t)]\succeq \lambda I$.  Hence, the minimum eigenvalue of $M$ is at least $\lambda$.

Thus our setting and our estimator satisfy all the assumptions required to apply Theorem~3.1 of \cite{chernozhukov2018double} to get the following result:
if we have that $\|\hat{h}-h^*\|_{2,2}=o_p(1)$ and that:
\begin{equation}
\max_{O\in \{S, S'\}} \rho(\hat{h}_O) = o_p(n^{-1/2})
\end{equation}
Then if we let:
\begin{align}
    \Sigma = \E[m(Z;\psi^*, h^*)\, m(Z;\psi^*, h^*)']
\end{align}
and $V=J^{-1} \Sigma (J^{-1})'$, we have that:
\begin{align}
    \sqrt{n} V^{-1/2} (\hat{\psi} - \psi^*) = \frac{1}{\sqrt{n}}\sum_{i=1}^n \bar{m}(Z_i) + o_p(1) \to_d N(0, I_{d\cdot m})
\end{align}
where:
\begin{align}
    \bar{m}(\cdot) = -V^{-1/2} J^{-1}\, m(\cdot; \psi^*, h^*)
\end{align}
Finally, we expand the block structure of $\Sigma$. If we let $\bar{Y}_t := Y - \sum_{j=t+1}^m \psi_j^* T_j$ then we have:
\begin{align}
    m_t(Z;\psi^*, h^*) = (\bar{Y}_{t-1} - \E\left[\bar{Y}_{t-1}\mid X_t\right]) (T_t - \E[T_t\mid X_t]) 
\end{align}
Hence we can write $\Sigma$ in block format of $d\times d$ blocks, such that the $(t, j)$ block for any $j,t \in [m]$ is of the form:
\begin{align}
    \Sigma_{t,j} =~& \E[(\bar{Y}_{t-1} - \E\left[\bar{Y}_{t-1}\mid X_t\right])\, (T_t - \E[T_t\mid X_t])\, (T_j - \E[T_j\mid X_j])'\,  (\bar{Y}_{j-1} - \E\left[\bar{Y}_{j-1}\mid X_j\right])]
\end{align}

The second part of the theorem on the construction of confidence intervals follows then directly by Corollary 3.1 of \cite{chernozhukov2018double}.
\end{proof}

\subsection{Detailed Exposition of Bounds on Parameter $\kappa$}\label{app:kappa}

If we denote with $\Delta = \alpha\, \gamma' + B$, then we can write by recursively expanding the linear Markovian expressions:
\begin{align}
    X_t =~& \Delta^q X_{t-q} + \eta_t + \sum_{\tau=0}^{q-1} \Delta^{\tau} \eta_{t-\tau} + \sum_{\tau=1}^q \Delta^{\tau} \alpha \zeta_{t-\tau}
\end{align}
Combining with the linear observational policy, we get that:
\begin{align}
    T_t =~& \gamma'\Delta^q X_{t-q} + \gamma'\eta_t + \sum_{r=0}^{q-1} \gamma'\Delta^{r} \eta_{t-r} + \sum_{r=1}^q \gamma'\Delta^{r} \alpha \zeta_{t-r} + \zeta_t =: \gamma'\Delta^q X_{t-q} + \iota_{t,q}
\end{align}
Since all error terms are independent of each other across iterations, we have that:
\begin{align}
    \Cov(T_t, T_j\mid X_t) = \E\left[(T_t - \E[T_t\mid X_t])\, (T_j - \E[T_j\mid X_t])\right] = \E\left[\zeta_t\, \iota_{j, j-t}\right] = \E\left[\zeta_t^2\right] \gamma'\Delta^{j-t}\alpha
\end{align}
Thus we conclude that:
\begin{align}
    \frac{2}{\lambda} \sum_{j=t+1}^{m} \|\Cov(T_t, T_j\mid X_t)\|_{op} =~& 2 \sum_{j=t+1}^m |\gamma'\Delta^{j-t}\alpha| \leq 2 \|\gamma\|_2 \|\alpha\|_2 \sum_{j=t+1}^m \|\Delta\|_{op}^{j-t}\\
    =~& 2 \|\gamma\|_2 \|\alpha\|_2 \sum_{\tau=1}^{m-t} \|\Delta\|_{op}^{\tau} = 2 \|\gamma\|_2 \|\alpha\|_2 \|\Delta\|_{op} \frac{\|\Delta\|_{op}^{m-t} - 1}{\|\Delta\|_{op} -1}
\end{align}
Thus we get that:
\begin{align}
    \kappa =  2 \|\gamma\|_2 \|\alpha\|_2 \|\Delta\|_{op} \frac{\|\Delta\|_{op}^{m-1} - 1}{\|\Delta\|_{op}-1}
\end{align}

\subsection{Proof of Theorem~\ref{thm:normality-dependent}}

Consider blocks of size $m$ and let $Z_{b}$ denote the set of all random variables in block $b$, starting from the first state $X_{b, 0}$ of that block and ending in the final outcome $Y_{b, m}$. Moreover, let $\cF_b$ denote the filtration before each block $b$ (i.e. all past variables before $b$) and for each $t\in [m]$:
\begin{align}
\mom_{b, t}(\psi, h) :=~& \E[m_t(Z_{b}; \psi, h_b)\mid \cF_{b}]\\
\mom_{D, t}(\psi, h) :=~& \frac{2}{B}\sum_{b=B/2}^B \mom_{b,t}(\psi, h_b)\\
\mom_{B, t}(\psi, h) :=~& \frac{2}{B}\sum_{b=B/2}^B m_{t}(Z_b; \psi, h_b)
\end{align}
Similarly, let $\mom_{b}$, $\mom_{D}$ and $\mom_{B}$, be the concatenation of these moment vectors across the $m$ values of $\kappa$.

Note that because the moment function $m_t$ is linear $\psi$, we can write:
\begin{align}
    \mom_{D}(\psi, h) = J(h)\psi + K(h)\\
    \mom_{B}(\psi, h) = \hat{J}(h)\psi + \hat{K}(h)
\end{align}
where $J(h)$ and $\hat{J}(h)$ are $(m d)\times (m d)$ block upper triangular matrices with blocks of size $d\times d$, such that for $1\leq t\leq j\leq m$ the $(t,j)$ blocks take the form:
\begin{align}
    J_{t,j}(h) =~& \frac{2}{B} \sum_{b=B/2}^B \E[(T_{b,t} - p_{b,t,t}(X_{b,t}))\, (T_{b,j} - p_{b,j,t}(X_{b,t}))'\mid \mcF_b]\\
    \hat{J}_{t,j}(h) =~& \frac{2}{B} \sum_{b=B/2}^B (T_{b,t} - p_{b,t,t}(X_{b,t}))\, (T_{b,j} - p_{b,j,t}(X_{b,t}))'
\end{align}

By a second order Taylor expansion of $\mom_{D}(\psi, h)$ around $\psi^*, h^*$, and due to orthogonality of the moment with respect to $h$, the linearity of the moment with respect to $\psi$, and the fact that $\mom_D(\psi^*, h^*)=0$ we have:
\begin{align}
    \mom_D(\psi, h) = \nabla_{\psi} \mom_D(\psi^*, h^*)\, (\psi - \psi^*) + \xi = J(h^*)\, (\psi - \psi^*) + \xi
\end{align}
where $\xi$, satisfies that:
\begin{align}
    \|\xi\| \leq c_0\, m\, d\, \frac{2}{B}\sum_{b=B/2}^B \|\hat{h}_b - h^*\|_{b, 2}^2
\end{align}
for some universal constant $c_0$, where we denote with $\|f\|_{b, 2}=\sqrt{\E[\|f\|_{2}^2\mid \mcF_b]}$ and where $\hat{h}_b$ are the nuisance models used on the samples in block $b$ during the progressive nuisance estimation. For simplicity of exposition, let:
\begin{align}
    \epsilon_B(\hat{h}) :=  \frac{2}{B}\sum_{b=B/2}^B \|\hat{h}_b - h^*\|_{b, 2}^2
\end{align}
Hence,
\begin{align}
    J(h^*)\, (\psi-\psi^*) =~& \mom_D(\psi, h) - \xi
\end{align}

Moreover, observe that by the definition of our estimator, we have:
\begin{equation}
    \mom_{B}(\hat{\psi}, \hat{h}) = 0
\end{equation}
Thus we conclude that:
\begin{align}
    J(h^*) (\hat{\psi}-\psi^*) =~& \left( \mom_D(\hat{\psi}, \hat{h}) - \mom_{B}(\hat{\psi}, \hat{h})\right)  - \xi
\end{align}

Moreover, observe that $\mom_D(\psi^*, h^*)=0$. Thus we can further expand the right hand side as:
\begin{align}
    \mom_D(\hat{\psi}, \hat{h}) - \mom_{B}(\hat{\psi}, \hat{h}) =~& -\mom_{B}(\psi^*, h^*) + \mom_D(\hat{\psi}; \hat{h}) - \mom_B(\hat{\psi}; \hat{h}) - (\mom_D(\psi^*, h^*) - \mom_B(\psi^*, h^*))\\
    =~& \underbrace{-\mom_{B}(\psi^*, h^*)}_{\text{Asymptotic normal term}} + \underbrace{\mom_D(\hat{\psi}; \hat{h}) - \mom_D(\psi^*, h^*) - (\mom_B(\hat{\psi}; \hat{h}) - \mom_B(\psi^*, h^*))}_{A=\text{Stochastic equicontinuous term}}
\end{align}

First we prove that the term $A$ decays to zero faster than a root-$B$ rate. To achieve this we further split $A$ as the sum of two stochastic equicontinuous terms:
\begin{align}
    A_1 :=~& \mom_D(\hat{\psi}; \hat{h}) - \mom_D(\psi^*, \hat{h}) - (\mom_B(\hat{\psi}; \hat{h}) - \mom_B(\psi^*, \hat{h}))\\
    A_2 :=~& \mom_D(\psi^*; \hat{h}) - \mom_D(\psi^*, h^*) - (\mom_B(\psi^*; \hat{h}) - \mom_B(\psi^*, h^*))
\end{align}

We note that due to linearity of the moment we have:
\begin{align}
    \|A_1\|_2 =~& \left\|J(\hat{h}) (\hat{\psi}-\psi^*) - \hat{J}(\hat{h}) (\hat{\psi}-\psi^*)\right\|_2
    = \|(J(\hat{h}) - \hat{J}(\hat{h}))\, (\hat{\psi}-\psi^*)\|_2
\end{align}
Let $A_{1t}$ denote the sub-vector of $A_1$ associated with moment $\mom_{D,t}$ and correspondingly let $J_t(\hat{h}) - \hat{J}_t(\hat{h})$ the sub-matrix containing the corresponding rows. Note that these rows are non-zero only on the parameters $\underline{\psi}_{t} := (\psi_{t}; \ldots; \psi_m)$. Thus we have:
\begin{align}
    \|A_{1,t}\|_2 \leq~&  \sum_{j=t}^m \left\|J_{t,j}(\hat{h}) - \hat{J}_{t,j}(\hat{h})\right\|_{op} \|\hat{\psi}_j - \psi_j^*\|_2
\end{align}
Each entry of $J_{t,j}(\hat{h}) - \hat{J}_{t,j}(\hat{h})$ is a martingale sequence. Thus by a Hoeffding-Azuma bound we have, w.p. $1-\delta$:
\begin{align}
    \forall 1\leq t\leq j \leq m: \|J_{t,j}(\hat{h}) - \hat{J}_{t,j}(\hat{h})\|_{op} \leq c_0\left(d\,\sqrt{\frac{\log(d\,m/\delta)}{B}}\right)
\end{align}
for some universal constant $c_0$.

Moreover, since $\hat{h}_b$ is trained in a progressive manner, we also have that each term in the vector $A_2$ is a martingale.
By a martingale Bernstein inequality (see e.g. \cite{freedman1975} or \cite{peel2013empirical}), and since we assumed that $\|\psi_t^*\|_2\leq M$ (for some $M>0$, which for simplicity we assume $M>1$), w.p. $1-\delta$:
\begin{multline}
    \|A_2\|_{2} \leq\\
      c_1\,d\,\left(\max_{\substack{t\in [m]\\ i\in [d]}}\sqrt{\frac{\frac{1}{B}\sum_{b}\E[(m_{t,i}(Z_b; \psi^*, \hat{h})-m_{t,i}(Z_b; \psi^*, h^*))^2\mid \cF_b]\, \log(m\,d/\delta)}{B}} +\frac{ m\, M\,d\log(m\,d/\delta)}{B}\right)
\end{multline}
for some universal constant $c_1$. Moreover, by the Lipschitzness of $m_{t,i}$ with respect to the nuisances, we have that:
\begin{align}
    \E[(m_{t,i}(Z_b; \psi^*, \hat{h})-m_{t,i}(Z_b; \psi^*, h^*))^2\mid \cF_b]\leq
     c_2\, m\,d\, \|\hat{h}_b-h^*\|_{b, 2}^2
\end{align}
for some universal constant $c_2$.

Thus we conclude that:
\begin{align}
    J(h^*) (\hat{\psi}-\psi^*) = - \mom_B(\phi^*, h^*) + \rho
\end{align}
where if we denote with $\rho_t$ the coordinates of $\rho$ associated with $\mom_{D,t}$, then w.p. $1-\delta$:
\begin{align}
    \forall t\in [m]: \|\rho_t\|_2 \leq~& c_0\, d\, \sqrt{\frac{\log(d\,m/\delta)}{B}} \sum_{j\geq t}\|\hat{\psi}_j-\psi_j^*\|_2 + c_1\,d\,\left(\sqrt{\frac{\epsilon_B(\hat{h})\, \log(m\,d/\delta)}{B}} +\frac{ m\, M\,d\log(m\,d/\delta)}{B}\right)\\
    ~& + c_0\, m\, d\, \epsilon_B(\hat{h})
\end{align}
Applying an AM-GM inequality we have:
\begin{align}
    \|\rho_t\|_2 \leq~& c_0\, d\, \sqrt{\frac{\log(d\,m/\delta)}{B}} \sum_{j\geq t}\|\hat{\psi}_j-\psi_j^*\|_2  + \frac{2c_1 \, M\, m\, d\, \log(d/\delta)}{B} + (c_0 + c_1) \, m\, d\, \epsilon_B(\hat{h})
\end{align}

\paragraph{$\ell_2$-error rate.} First note that:
\begin{align}
    J_{t,j}(h^*) := \E[\Cov(T_t,T_j\mid X_t)]
\end{align}
Since $J_{t,t}(h^*) = \E[\Cov(T_t,T_t\mid X_t)] \succeq \lambda I$, then we have that:
\begin{align}
    \|\hat{\psi}_t-\psi_t^*\|_2 \leq~&  \frac{1}{\lambda}\sum_{j=t+1}^m \|J_{t,j}(h^*)\|_{op} \|\hat{\psi}_j - \psi_j^*\|_2 + \frac{1}{\lambda} \|\mom_{B,t}(\psi^*, h^*)\|_2 + \frac{1}{\lambda}\|\rho\|_2
\end{align}
If we let $\kappa := \max_{t=1}^m \frac{1}{\lambda} \sum_{j=t+1}^m \|J_{t,j}(h^*)\|_{op}$, then
note that for $B\geq \frac{2\,c_0^2 d^2 \log(d\,m/\delta)}{\lambda^2}$, and 
\begin{align}
    \|\hat{\psi}_t-\psi_t^*\|_2 \leq~& \frac{1}{\lambda}\|\mom_{B,t}(\psi^*, h^*)\|_2 + \frac{1}{2} \|\hat{\psi}_t-\psi_t^*\|_2 + \underbrace{\left(\kappa + c_0\, d\, \sqrt{\frac{\log(d\,m/\delta)}{B}}\right)}_{\kappa_B}\max_{j=t+1}^m \|\hat{\psi}_j-\psi_j^*\|_2\\
    ~& + \frac{2c_1 \, M\, m\, d\,\log(d/\delta)}{\lambda\, B} + \frac{(c_0 + c_1) \, m\, d}{\lambda}  \epsilon_B(\hat{h})
\end{align}
and hence:
\begin{align}
    \|\hat{\psi}_t-\psi_t^*\|_2 \leq~& 
    \frac{2}{\lambda}\|\mom_B(\psi^*, h^*)\|_2 + \frac{4c_1 \, M\, m\, d\,\log(d/\delta)}{\lambda\, B} + \frac{2(c_0 + c_1) \, m\, d}{\lambda}  \epsilon_B(\hat{h}) + \kappa_B\max_{j=t+1}^m \|\hat{\psi}_j-\psi_j^*\|_2
\end{align}
Moreover, every coordinate of $\mom_B(\psi^*,h^*)$ is a martingale and thereby by a Hoeffding-Azuma inequality, w.p. $1-\delta$:
\begin{align}
    \|\hat{\psi}_t-\psi_t^*\|_2 \leq~& 
    \frac{2c_3}{\lambda}\sqrt{\frac{m\,d\,\log(m\,d/\delta)}{B}} + \frac{4c_1 \, M\, m\, d\,\log(d/\delta)}{\lambda\, B} + \frac{2(c_0 + c_1) \, m\, d}{\lambda}  \epsilon_B(\hat{h}) + \kappa_B\max_{j=t+1}^m \|\hat{\psi}_j-\psi_j^*\|_2
\end{align}
for some universal constant $c_3$.

Via an inductive argument identical to that at the end of the proof of Theorem~\ref{thm:mse}, we can then conclude that w.p. $1-\delta$, $\forall t\in [m]$:
\begin{align}
    \|\hat{\psi}_t-\psi_t^*\|_2 \leq~& \left(
    \frac{2c_3}{\lambda}\sqrt{\frac{m\,d\,\log(m\,d/\delta)}{B}} + \frac{4c_1 \, M\, m\, d\,\log(d/\delta)}{\lambda\, B} + \frac{2(c_0 + c_1) \, m\, d}{\lambda}  \epsilon_B(\hat{h})\right)  \frac{\kappa_B^{m-t+1}-1}{\kappa_B - 1}
\end{align}

Replacing this bound back in the upper bound on $\rho$, we have that w.p. $1-\delta$:
\begin{align}
    \|\rho\|_2 \leq~& \frac{c_4 \, M\, m^3\, d\, \log(m\,d/\delta)}{\lambda\, B} \frac{\kappa_B^{m-t+1}-1}{\kappa_B - 1} + c_5\, m^2\, d\, \epsilon_B(\hat{h}) + c_0 m^3 d \sqrt{\frac{\log(d\,m/\delta)}{B}} \frac{\kappa_B^{m-t+1}-1}{\kappa_B - 1} \epsilon_B(\hat{h}) 
\end{align}
for some constants $c_4, c_5$. For any constant $\epsilon$ and for $B\geq \frac{2\,c_0^2 d^2 \log(d\,m/\delta)}{\epsilon^2}$, we have that:
\begin{align}
    \frac{\kappa_B^{m-t+1}-1}{\kappa_B - 1} \leq \frac{(\kappa + \epsilon)^{m-t+1}-1}{(\kappa+\epsilon) - 1}
\end{align}

Thus as long as for some constant $\epsilon>0$:
\begin{align}
m^3 \sqrt{\log(d\,m)}\frac{(\kappa + \epsilon)^{m-t+1}-1}{\kappa + \epsilon - 1} \E[\epsilon_B(\hat{h})]=~& o(1),\\
\sqrt{B} m^2 \E[\epsilon_B(\hat{h})] =~& o(1)\\
\frac{m^3 \log(d\,m)}{\sqrt{B}} \frac{(\kappa + \epsilon)^{m-t+1}-1}{\kappa + \epsilon - 1} =~& o(1)
\end{align}
we have that:
\begin{align}
    \sqrt{B}\, \|\rho\|_2 = o_p(1)
\end{align}

\paragraph{Asymptotic normal term.} Thus it suffices to show that:
\begin{equation}
    -\sqrt{B/2}\, \mom_{B}(\psi^*, h^*) \rightarrow_d N(0, \Sigma)
\end{equation}
as we can then apply Slutzky's theorem to get the desired result.

Observe that under the true values of $\psi^*, h^*$, the moment $m_{b, t}(Z_b; \psi^*, h^*)$ can be written as:
\begin{align}
    m_{b, t}(Z_b; \psi^*, h^*) = \left(\sum_{j=t+1}^{m} \mu' B^{m-j} \eta_{b,j} + \epsilon_{b,m}\right) \zeta_{b,t}
\end{align}
Thus the random vector $m_{b}(Z_b; \psi^*, h^*)$ is only a function of the exogenous shocks, $\{\zeta_{b,t}, \eta_{b,t}\}_{t=1}^m, \epsilon_m$. Since these variables are independent and identically distributed across all blocks\footnote{Our theorem would easily extend if these variables form a martingale with a non-zero variance at each step}, we can conclude that the conditional variance:
\begin{equation}
    \Var(m_b(Z_b; \psi^*, h^*) \mid \cF_b) = \Sigma \succ \sigma^2 \lambda I
\end{equation}
where $\Sigma$ is a block diagonal matrix (since $\zeta_t$ is independent of $\zeta_{t'}$ for any $t\neq t'$ and mean zero), with diagonal blocks $\E[\left(\sum_{j=t+1}^{m} \mu' B^{m-j} \eta_{b,j} + \epsilon_{b,m}\right)^2 \zeta_{b,t}\zeta_{b,t}']$ and thus we can take $\sigma^2 = \E[\epsilon_{b,m}^2]$ and $\lambda$ is the minimum eigenvalue of $\E[\zeta_{b,t}\zeta_{b,t}']$. Importantly $\Sigma$ is a constant, independent of $b$ and $\mcF_b$, positive definite symmetric matrix with a lower bound on the minimum eigenvalue that is independent of $m$. We can thus apply a Martingale Central Limit Theorem (see e.g. \cite{Hall1980}) to get the desired statement:
\begin{equation}
    \sqrt{B/2} \Sigma^{-1/2} J(h^*) (\hat{\psi}-\psi^*) \rightarrow_d N(0, I)
\end{equation}

\section{Omitted Proofs from Generalization to SNMMs}

\subsection{Proof of Lemma~\ref{lem:cntf-char-rho}}
\begin{proof}
First we note that by invoking the sequential conditional exogeneity condition, for any dynamic policy $\pol$ and for any treatment sequence $\bar{\tau}_m$, we have:
\begin{align}
    \gamma(\bar{x}_j, (\bar{\tau}_{j-1}, \pol(\bar{x}_j, \bar{\tau}_{j-1}))) =~& \E[Y^{(\bar{\tau}_{j-1}, \underline{\pol}_j)} - Y^{(\bar{\tau}_{j-1}, 0, \underline{\pol}_{j+1})}\mid \bar{X}_j=\bar{x}_j, \bar{T}_{j-1} = \bar{\tau}_{j-1}, T_j=\pol(\bar{X}_j, \bar{T}_{j-1})] \nonumber\\
    =~&  \E[Y^{(\bar{\tau}_{j-1}, \underline{\pol}_j)} - Y^{(\bar{\tau}_{j-1}, 0, \underline{\pol}_{j+1})}\mid \bar{X}_j=\bar{x}_j, \bar{T}_{j-1} = \bar{\tau}_{j-1}, T_j=\tau_j] \nonumber\\
    =~&  \E[Y^{(\bar{\tau}_{j-1}, \underline{\pol}_j)} - Y^{(\bar{\tau}_{j-1}, 0, \underline{\pol}_{j+1})}\mid \bar{X}_j=\bar{x}_j, \bar{T}_{j} = \bar{\tau}_{j}]\nonumber
\end{align}
Using the latter fact and invoking the definition of $\rho_j(\bar{x}_j, \bar{\tau}_j)$, we can write:
\begin{align}
    \rho_j(\bar{x}_j, \bar{\tau}_j) =~& \gamma(\bar{x}_j, (\bar{\tau}_{j-1}, \pol(\bar{x}_j, \bar{\tau}_{j-1}))) - \gamma(\bar{x}_j, \bar{\tau}_j) \nonumber\\
    =~& \E\left[Y^{(\bar{\tau}_{j-1}, \underline{\pol}_j)} - Y^{(\bar{\tau}_{j-1}, 0, \underline{\pol}_{j+1})} - \left(Y^{(\bar{\tau}_{j}, \underline{\pol}_{j+1})} - Y^{(\bar{\tau}_{j-1}, 0, \underline{\pol}_{j+1})}\right)\mid \bar{X}_j=\bar{x}_j, \bar{T}_j = \bar{\tau}_j\right]\nonumber\\
    =~& - \E\left[Y^{(\bar{\tau}_j, \underline{\pol}_{j+1})} - Y^{(\bar{\tau}_{j-1}, \underline{\pol}_j)}\mid \bar{X}_j=\bar{x}_j, \bar{T}_j = \bar{\tau}_j\right]  \label{eqn:rho-identity}
\end{align}

Second, note that by the definition of the observed $Y$, we have that:
\begin{align}
    Y \equiv Y^{(\bar{T}_m)}
\end{align}
and that for any treatment sequence $\bar{\tau}_m$, via a telescoping sum argument, we can write:
\begin{align}
Y^{(\tau)} - Y^{(\bar{\tau}_{t-1}, \underline{\pol}_t)} = \sum_{j=t}^m Y^{(\bar{\tau}_j, \underline{\pol}_{j+1})} - Y^{(\bar{\tau}_{j-1}, \underline{\pol}_j)}
\end{align}
Thus, applying linearity of expectation, the tower law of expectations and Equation~\eqref{eqn:rho-identity}, we have:
\begin{align}
    \E\left[Y - Y^{(\bar{\tau}_{t-1}, \underline{\pol}_t)} \mid \bar{X}_t, \bar{T}_{t}=\bar{\tau}_t\right] =~& 
    \E\left[Y^{(T)} - Y^{(\bar{T}_{t-1}, \underline{\pol}_t)} \mid \bar{X}_t, \bar{T}_{t}=\bar{\tau}_t\right]\\
    =~& 
    \E\left[\sum_{j=t}^m Y^{(\bar{T}_j, \underline{\pol}_{j+1})} - Y^{(\bar{T}_{j-1}, \underline{\pol}_j)} \mid \bar{X}_t, \bar{T}_{t}=\bar{\tau}_t\right]\\
    =~& 
    \sum_{j=t}^m \E\left[Y^{(\bar{T}_j, \underline{\pol}_{j+1})} - Y^{(\bar{T}_{j-1}, \underline{\pol}_j)} \mid \bar{X}_t, \bar{T}_{t}=\bar{\tau}_t\right]\\
    =~& 
    \sum_{j=t}^m \E\left[\E\left[Y^{(\bar{T}_j, \underline{\pol}_{j+1})} - Y^{(\bar{T}_{j-1}, \underline{\pol}_j)}\mid \bar{X}_j, \bar{T}_j\right] \mid \bar{X}_t, \bar{T}_{t}=\bar{\tau}_t\right]\\
    =~& 
    \sum_{j=t}^m - \E\left[\rho(\bar{X}_j, \bar{T}_j) \mid \bar{X}_t, \bar{T}_{t}=\bar{\tau}_t\right]\\
    =~& - \E\left[
    \sum_{j=t}^m \rho(\bar{X}_j, \bar{T}_j) \mid \bar{X}_t, \bar{T}_{t}=\bar{\tau}_t\right]
\end{align}
By re-arranging we conclude the desired property.
\end{proof}

\subsection{Proof of Lemma~\ref{lem:moment-restrictions}}
\begin{proof}
For simplicity of notation, for any $f\in \mcF$, let:
\begin{align}
    \bar{f}(\bar{X}_t, \bar{T}_t) = f(\bar{X}_t, \bar{T}_t) - \E[f(\bar{X}_t, \bar{T}_t)\mid \bar{X}_t, \bar{T}_{t-1}]
\end{align}
and observe that by the definition of $\bar{f}$, we crucially have that for any $f\in \mcF$:
\begin{align}
\E\left[\bar{f}(\bar{X}_t, \bar{T}_t)\mid \bar{X}_t, \bar{T}_{t-1}\right] = 0.
\end{align}

By Lemma~\ref{lem:cntf-char-rho}, we have that at the true $\psi^*$:
\begin{align}
    \E\left[ H_t(\psi^*)\, \bar{f}(\bar{X}_t, \bar{T}_t)\right]=~& \E\left[ \E\left[H_t(\psi^*) \mid \bar{X}_t, \bar{T}_t\right]\, \bar{f}(\bar{X}_t, \bar{T}_t)\right] \tag{tower law}\\
    =~& \E\left[ \E\left[Y^{(\bar{T}_{t-1}, \bar{\pol}_t)} \mid \bar{X}_t, \bar{T}_t\right]\, \bar{f}(\bar{X}_t, \bar{T}_t)\right] \tag{Equation~\eqref{eqn:cntf-rho}}\\
    =~& \E\left[ \E\left[Y^{(\bar{T}_{t-1}, \bar{\pol}_t)} \mid \bar{X}_t, \bar{T}_{t-1}\right]\, \bar{f}(\bar{X}_t, \bar{T}_t)\right] \tag{sequential exogeneity}\\
    =~& \E\left[ \E\left[Y^{(\bar{T}_{t-1}, \bar{\pol}_t)} \mid \bar{X}_t, \bar{T}_{t-1}\right]\, \E\left[\bar{f}(\bar{X}_t, \bar{T}_t)\mid \bar{X}_t, \bar{T}_{t-1}\right]\right] \tag{tower law}\\
    =~& 0 \nonumber
\end{align}
\end{proof}

\subsection{Proof of Universal Neyman Orthogonality Property of Loss}

First we define for any functional $L(f)$ the Frechet derivative as:
\begin{equation}
    D_f L(f)[\nu] = \frac{\partial}{\partial t} L(f + t\, \nu)\mid_{t=0} 
\end{equation}
Similarly, we can define higher order derivatives, denoted as $D_{g,f} L(f, g)[\mu,\nu]$, $D_{h, g,f} L(h, f, g)[\kappa, \mu,\nu]$.

\begin{lemma}[Universal Neyman orthogonal loss]\label{lem:neyman-ortho-snmm}
The loss function:
\begin{equation}
    \loss_{D,t}(\psi; h) := \frac{1}{2}\E\left[ \left(Y - q_t(X_t) - \sum_{j=t+1}^m \psi_j'(T_j - p_{j,t}(X_t)) - \psi_t'(T_t - p_{t,t}(X_t))\right)^2\right]
\end{equation}
satisfies universal Neyman orthogonality with respect to $h$, i.e.:
\begin{equation}
    \forall h, \psi_t, \tilde{\psi}: D_{h,\psi_t} \loss_{D,t}(\tilde{\psi}; h^*)[h - h^*, \psi_t - \tilde{\psi}_t] = 0
\end{equation}
\end{lemma}
\begin{proof}
Let $\nu_t := \psi_t - \tilde{\psi}_t$. First note that by the definition of the Frechet derivative and the chain rule of differentiation:
\begin{align}
    D_{h,\psi_t} \loss_{D,t}(\tilde{\psi}; h^*)[h - h^*, \nu_t] =~& D_{q_t, \psi_t} \loss_{D,t}(\tilde{\psi}; h^*)[q_t - q_t^*, \nu_t]
    + \sum_{j=t}^m D_{p_{j,t}, \psi_t} \loss_{D,t}(\tilde{\psi}; h^*)[p_{j,t} - p_{j,t}^*, \nu_t]
\end{align}
We consider each component of the nuisance parameters separately. We also introduce the short-hand notation: $I_t=(\bar{X}_t, \bar{T}_{t-1})$ and not that $\bar{X}_t$ also contains $X_0$.

For $q_t$, we have:
\begin{align}
    D_{q_t, \psi_t} \loss_{D,t}(\tilde{\psi}; h^*)[q_t - q_t^*, \nu_t] =~& -\E[\nu_t(X_0)'(T_t - p_{t,t}^*(I_t))]\\
    =~& - \E[\nu_t(X_0)'(T_t - \E[T_t\mid I_t])] = 0
\end{align}
For $p_{j,t}$ with $j>t$, we have:
\begin{align}
    D_{p_{j,t} \psi_t} \loss_{D,t}(\tilde{\psi}; h^*)[p_{j,t} - p_{j,t}^*, \nu_t] =~& \E[\nu_t(X_0)'(T_t - p_{t,t}^*(I_t))]\\
    =~& \E[\nu_t(X_0)'(T_t - \E[T_t\mid I_t])] = 0
\end{align}
For $p_{t,t}$, we have:
\begin{align}
    D_{p_{t,t} \psi_t} \loss_{D,t}(\tilde{\psi}; h^*)[p_{t,t} - p_{t,t}^*, \nu_t] =~& \E[\nu_t(X_0)'(T_t - p_{t,t}^*(I_t))]\\
    ~& - \sum_{i=1}^r \E\left[\nu_{t,i}(X_0)\, (Y - q_t^*(I_t))\right] \\
    ~& + \sum_{i=1}^r\sum_{j=t}^m \E\left[\nu_{t,i}(X_0)\, \tilde{\psi}_j(X_0))'(T_j - p_{j,t}^*(I_t))\right] \\
    =~&  0
\end{align}
\end{proof}

\subsection{Proof of Lemma~\ref{lem:general}}

\begin{lemma}\label{lem:cov-property}
For any random vectors $I, X, Y$ and any two vector-valued functions $f, g$, we have that:
\begin{align}
    \E[(X - f(I))\, (Y - g(I))'\mid I] = \Cov(X, Y\mid I) + \E\left[\left(\E[X\mid I] - f(I)\right)\, \left(\E[Y\mid I] - g(I)\right)' \mid I\right]
\end{align}
\end{lemma}
\begin{proof}
Let:
\begin{align}
A(f, g) :=~& \E[(X - f(I))\, (Y - g(I))']\\
B(f, g) :=~&\E\left[\left(\E[X\mid I] - f(I)\right)\, \left(\E[Y\mid I] - g(I)\right)'\right]
\end{align}
By simple numeric manipulations we have:
\begin{align}
    A(f, g) =~& \E[(X - \E[X\mid I])\, (Y - g(I))'\mid I] + \E[(\E[X\mid I] - f(I))\, (Y - g(I))'\mid I]\\
    =~& \E[(X - \E[X\mid I])\, (Y - g(I))'\mid I] + \E[(\E[X\mid I] - f(I))\, (Y - \E[Y\mid I])'\mid I] + B(f,g)\\
    =~& \E[(X - \E[X\mid I])\, (Y - g(I))'\mid I] + B(f,g)\\
    =~& \E[(X - \E[X\mid I])\, (Y - \E[Y\mid I])'\mid I] + \E[(X - \E[X\mid I])\, (\E[Y\mid I] - g(I))'\mid I] + B(f,g)\\
    =~& \E[(X - \E[X\mid I])\, (Y - \E[Y\mid I])'\mid I] + B(f,g)\\
    =~& \Cov(X, Y\mid I)  + B(f,g)
\end{align}
\end{proof}

\begin{proof}[Proof of Lemma~\ref{lem:general}]
First we define for any functional $L(f)$ the Frechet derivative as:
\begin{equation}
    D_f L(f)[\nu] = \frac{\partial}{\partial t} L(f + t\, \nu)\mid_{t=0} 
\end{equation}
Similarly, we can define higher order derivatives, denoted as $D_{g,f} L(f, g)[\mu,\nu]$ and $D_{h,g,f}L(f,g,h)[\kappa, \mu, \nu]$.

We first remind the definition of the population loss that we will use throughout the proof:
\begin{equation}
    \loss_{D,t}(\psi_t; \underline{\psi}_t, h) := \frac{1}{2}\E\left[ \left(Y - q_t(X_t) - \sum_{j=t+1}^m \psi_j'(T_j - p_{j,t}(X_t)) - \psi_t'(T_t - p_{t,t}(X_t))\right)^2\right]
\end{equation}
For simplicity, in the remainder of the proof we will use the short-hand notation $\loss_{D,t}\equiv \loss_D$, since the time $t$ will be evident by the inputs to the loss. Moreover, we'll use the shorthand notation $I_t := (\bar{X}_t, \bar{T}_{t-1})$. Moreover, we will be using the shorthand notation $\tilde{T}_{t} \equiv \tilde{T}_{t,t}$ and $p_{t,t}^* \equiv p_t$, $\hat{p}_{t,t} \equiv \hat{p}_t$. Finally, let:
\begin{align}
\nu_t :=~& \hat{\psi}_t - \psi_t^{*,n} & \nu_t^* :=~& \hat{\psi}_t - \psi_t^{*}
\end{align}

\paragraph{Strong convexity.} First, we argue that the plug-in loss function is functionally $\lambda$-strongly convex. 
For any $\psi_t\in \Psi_t^n$:
\begin{align}
    D_{\psi_t, \psi_t}  \loss_{D}(\psi_t; \hat{\underline{\psi}}_{t+1}, \hat{h}) [\nu_t, \nu_t]
    =~& \E\left[\nu_t(X_0)' \tilde{T}_{t} \tilde{T}_{t}' \nu_t(X_0)\right]\\
    =~& \E\left[\nu_t(X_0)' \Cov(Q_{t,t}, Q_{t,t}\mid I_t)\, \nu_t(X_0)\right]
     + \E\left[\left(\nu_t(X_0)' (p_{t}^*(I_t) - \hat{p}_{t}(I_t))\right)^2\right]\\
    \geq~& \E\left[\nu_t(X_0)' \E\left[\Cov(Q_{t,t}, Q_{t,t}\mid I_t)\mid X_0\right]\nu_t(X_0)\right]\\
    \geq~& \lambda \|\nu_t\|_{2,2}^2
\end{align}
where in the second identity we invoked Lemma~\ref{lem:cov-property} and in the last inequality we used our positivity assumption. 
By the $\lambda$-strong convexity of $\loss_D$, we have that
\begin{align}
    \loss_D(\hat{\psi}; \hat{\underline{\psi}}_{t+1}, \hat{h}) &\ge \loss_{D}(\psi_t^{*,n}; \hat{\underline{\psi}}_{t+1}, \hat{h}) + D_{\psi_t} \loss_{D}(\psi_t^{*,n}; \hat{\underline{\psi}}_{t+1}, \hat{h})[\nu_t] +  \frac{\lambda}{2}\,\|\nu_t\|_{2,2}^2.
\end{align}

Furthermore, our excess risk assumption and the optimality of $\psi^*$ give us 
\begin{align}
    \frac{\lambda}{2}\|\nu_t\|_{2,2}^2 \leq~& \underbrace{\loss_{D}(\hat{\psi}_t; \hat{\underline{\psi}}_{t+1}, \hat{h}) - \loss_{D}(\psi_t^{*,n}; \hat{\underline{\psi}}_{t+1}, \hat{h})}_{\text{plug-in excess risk of $\hat{\psi}_t$}} - D_{\psi_t}\loss_{D}(\psi_t^{*,n}; \hat{\underline{\psi}}_{t+1}, \hat{h})[\nu_t]\\
    \overset{(a)}{\leq}~& \epsilon(\hat{\underline{\psi}}_{t}, \hat{h}) - \underbrace{D_{\psi_t}\loss_{D}(\psi_t^{*,n}; \underline{\psi}_{t+1}^*, h^*)[\nu_t]}_{=:F}\\
    &~~ + \underbrace{D_{\psi_t} (\loss_{D}(\psi_t^{*,n}; \underline{\psi}_{t+1}^*, h^*)- \loss_{D}(\psi_t^{*,n}; \underline{\psi}_{t+1}^*, \hat{h}))[\nu_t]}_{=:A}\\
    &~~ + \underbrace{D_{\psi_t} (\loss_{D}(\psi_t^{*,n}; \underline{\psi}_{t+1}^*, \hat{h})- \loss_{D}(\psi_t^{*,n}; \hat{\underline{\psi}}_{t+1}, \hat{h}))[\nu_t]}_{=:B}.
\end{align}

\paragraph{Upper bounding $|F|$.} Note that we can write $D_{\psi_t}\loss_{D}(\psi_t^{*,n}; \underline{\psi}_{t+1}^*, h^*)[\nu_t]$ as:
\begin{align}
     -\E[(H_{t+1}(\psi^*) - \E[H_{t+1}(\psi^*)\mid I_t] - \psi_t^{*,n}(X_0)'(Q_t - \E[Q_t\mid I_t]))\, (Q_t - \E[Q_t\mid I_t])'\nu_t(X_0)]
\end{align}
Moreover, by applying the conditional moment restriction Lemma~\ref{lem:moment-restrictions} with $f(I_t)=Q_t'\,\nu_t(X_0)$ we have that:
\begin{align}
    \E[(H_{t+1}(\psi^*) - \E[H_{t+1}(\psi^*)\mid I_t] - \psi_t^{*}(X_0)'(Q_t - \E[Q_t\mid I_t]))\, (Q_t - \E[Q_t\mid I_t])'\nu_t(X_0)] = 0
\end{align}
Thus, we can equivalently write $D_{\psi_t}\loss_{D}(\psi_t^{*,n}; \underline{\psi}_{t+1}^*, h^*)[\nu_t]$ as:
\begin{multline}
    -\E[(\psi_t^*(X_0) - \psi_t^{*,n}(X_0))'(Q_t - \E[Q_t\mid I_t])\, (Q_t - \E[Q_t\mid I_t])'\nu_t(X_0)] \\
    = -\E[(\psi_t^*(X_0) - \psi_t^{*,n}(X_0))'\E[C_{t,t}\mid X_0]\nu_t(X_0)]
\end{multline}
By Holder's inequality, we have:
\begin{align}
    |F| \leq~& \E\left[\|\psi_t^*(X_0) - \psi_t^{*,n}(X_0)\|_{u} \|\E[C_{t,t}\mid X_0]\|_{\bar{u}, u} \|\nu_t(X_0)\|_{u}\right]\\
    \leq~& \|\nu_t\|_{u,v} \E\left[ \|\E[C_{t,t}\mid X_0]\|_{\bar{u}, u}^{\bar{v}} \|\psi_t^*(X_0) - \psi_t^{*,n}(X_0)\|_{u}^{\bar{v}}\right]^{1/\bar{v}}\\
    \leq~& \|\nu_t\|_{u,v} \|\psi_t^* - \psi_t^{*,n}\|_{u, \bar{v}} \sup_{x_0} \|\E[C_{t,t}\mid X_0]\|_{\bar{u}, u}\\
    =~& \|\nu_t\|_{u,v} \|\psi_t^* - \psi_t^{*,n}\|_{u, \bar{v}} c_{t,t}
\end{align}

\paragraph{Upper bounding $|B|$.} Next, note that:
\begin{align}
    D_{\psi_t} \loss_{D}(\psi_t; \underline{\psi}_{t+1}, \hat{h})[\nu_t] = - \E\left[\nu_t(X_0)'\tilde{T}_{t}\, \left(\tilde{Y}_t - \sum_{j=t+1}^m \psi_{j}(X_0)'\tilde{T}_{j,t} - \psi_t(X_0)'\tilde{T}_{t}\right)\right]
\end{align}
Thus by invoking Lemma~\ref{lem:cov-property} and applying the tower law of expectations, we can upper bound $B$ as:
\begin{align}
    |B| \leq~& \sum_{j=t+1}^m \left|\E\left[\nu_t(X_0)'\tilde{T}_{t}\, \tilde{T}_{j,t}' \nu_j^*(X_0)\right]\right|\\
    \leq~& \sum_{j=t+1}^m \left|\E\left[\nu_t(X_0)'\Cov(Q_{t,t}, Q_{j,t}\mid I_t)\,  \nu_j^*(X_0)\right]\right|\\
    ~& + \left|\E\left[\nu_t(X_0)'\left(\hat{p}_{t}(I_t) - p_{t}^*(I_t)\right)\, (\hat{p}_{j,t}(I_t) - p_{j,t}^*(I_t))'\, \nu_j^*(X_0)\right]\right|\\
    =~& \sum_{j=t+1}^m \left|\E\left[\nu_t(X_0)'\E[\Cov(Q_{t,t}, Q_{j,t}\mid I_t)\mid X_0]  \nu_j^*(X_0)\right]\right|\\
    ~& + \left|\E\left[\nu_t(X_0)'\E\left[\left(\hat{p}_{t}(I_t) - p_{t}^*(I_t)\right)\, (\hat{p}_{j,t}(I_t) - p_{j,t}^*(I_t))'\mid X_0\right]\, \nu_j^*(X_0)\right]\right|
\end{align}
Invoking the definition of $C_{t,j}$ and $\Delta(f,g)$ from the theorem statement and by repeatedly applying Holder's inequality for the dual norms associated with $u,\bar{u}$ and $v,\bar{v}$, we have:
\begin{align}
    |B| \leq~& \sum_{j=t+1}^m \|\nu_t\|_{u,v} \E\left[\|\E[C_{t,j}\mid X_0]\|_{\bar{u},u}^{\bar{v}} \left\|\nu_j^*(X_0)\right\|_{u}^{\bar{v}}\right]^{1/\bar{v}}
    + 2\,M\, \|\nu_t\|_{u,v} \left\|\|\Delta(\hat{p}_{t}, \hat{p}_{j,t})\|_{\bar{u},u}\right\|_{\bar{v}}\\
    \leq~& \sum_{j=t+1}^m \sup_{x_0} \|\E[C_{t,j}\mid X_0=x_0]\|_{\bar{u},u}\, \|\nu_t\|_{u,v} \|\nu_j^*\|_{u,\bar{v}}
    + 2\,M\, \|\nu_t\|_{u,v}\left\|\|\Delta(\hat{p}_{t}, \hat{p}_{j,t})\|_{\bar{u},u}\right\|_{\bar{v}}\\
    =~& \sum_{j=t+1}^m c_{t,j}\, \|\nu_t\|_{u,v} \|\nu_j^*\|_{u,\bar{v}} + 2\,M\, \|\nu_t\|_{u,v} \left\|\|\Delta(\hat{p}_{t}, \hat{p}_{j,t})\|_{\bar{u},u}\right\|_{\bar{v}}
\end{align}

\paragraph{Upper bound $|A|$.} Finally, we analyze the quantity $A$. By the universal Neyman orthogonality property of the loss from Lemma~\ref{lem:neyman-ortho-snmm}, we have for any $h$:
\begin{align}
    D_{h,\psi_t} \loss_{D}(\psi_t^{*,n}; \underline{\psi}_{t+1}^*, h^*))[h - h^*, \nu_t] = 0
\end{align}
By a second order Taylor's theorem with integral remainder,
\begin{align}
    A =~& \frac{1}{2} \int_{0}^1  \underbrace{D_{h,h,\psi_t}\loss_{D}(\psi_t^{*,n}; \underline{\psi}_{t+1}^*, h_\tau))[\hat{h}-h^*, \hat{h}-h^*, \nu_t]}_{G(h_\tau, \hat{h}, \nu_t)} d\tau
\end{align}
where $h_\tau = h^* + \tau (\hat{h} - h^*)$. Moreover:
\begin{align}
   G(h_\tau, \hat{h}, \nu_t) =~&  - \E\left[\nu_t(X_0)'\left(\hat{p}_{t}(I_t) -p_{t}^*(I_t)\right)\, \left(\hat{q}_t(I_t) - q_t^*(I_t)\right)\right] \\
   ~&~ + \sum_{j=t+1}^m\E\left[\nu_t(X_0)'\left(\hat{p}_{t}(I_t) -p_{t}^*(I_t)\right)\, \left(\hat{p}_{j,t}(I_t) - p_{j,t}^*(I_t)\right)'\psi_j^*(X_0)\right]\\
   ~&~ + 2\, \E\left[\nu_t(X_0)'\left(\hat{p}_{t}(I_t) -p_{t}^*(I_t)\right)\, \left(\hat{p}_{t,t}(I_t) - p_{t,t}^*(I_t)\right)'\psi_t^{*,n}(X_0)\right]
\end{align}
Thus by applying the tower law of expectations and repeatedly applying Holder's inequality, we have:
\begin{align}
    |A| \leq~& \|\nu_t\|_{u,v} \left\|\|\Delta(\hat{p}_{t}, \hat{q}_t)\|_{\bar{u},u}\right\|_{\bar{v}}
    + M \|\nu_t\|_{u,v} \sum_{j=t}^m  \left\|\|\Delta(\hat{p}_{t}, \hat{p}_{j,t})\|_{\bar{u},u}\right\|_{\bar{v}}
\end{align}

\paragraph{Concluding.} Overall we have:
\begin{align}
    \frac{\lambda}{2} \|\nu_t\|_{2,2}^2 \leq~& \epsilon(\psi_t^{*,n}, \hat{\underline{\psi}}_{t}, \hat{h}) + |F| + |A| + |B|\\
    \leq~& \epsilon(\psi_t^{*,n}, \hat{\underline{\psi}}_{t}, \hat{h}) + \|\nu_t\|_{u,v}\, \left(c_{t,t} \|\psi_t^{*,n} - \psi_t^*\|_{u,\bar{v}} + \sum_{j=t+1}^m c_{t,j} \|\nu_j^*\|_{u,\bar{v}} \right) \\
    ~&~ + \|\nu_t\|_{u,v} \left(\left\|\|\Delta(\hat{p}_{t}, \hat{q}_t)\|_{\bar{u},u}\right\|_{\bar{v}}
    + 3 M \sum_{j=t}^m  \left\|\|\Delta(\hat{p}_{t}, \hat{p}_{j,t})\|_{\bar{u},u}\right\|_{\bar{v}}\right)
\end{align}
Let:
\begin{align}
    \rho_{t,u,v}(\hat{h}) = \left\|\|\Delta(\hat{p}_{t}, \hat{q}_t)\|_{\bar{u},u}\right\|_{\bar{v}}
    + 3 M \sum_{j=t}^m  \left\|\|\Delta(\hat{p}_{t}, \hat{p}_{j,t})\|_{\bar{u},u}\right\|_{\bar{v}}
\end{align}

By an AM-GM inequality, for all $a,b\geq 0$ and $\sigma >0$: $a\cdot b \leq \half(\frac{2}{\sigma} a^2 + \frac{\sigma}{2} b^2)$. Applying the AM-GM inequality to the term that contains $\|\nu_t\|_{u,v}$ and re-arranging, yields that:
\begin{align}
    \frac{\lambda}{2}\|\nu_t\|_{2,2}^2 - \frac{\sigma}{4} \|\nu_t\|_{u,v}^2 \leq~& \epsilon(\psi_t^{*,n}, \hat{\underline{\psi}}_{t}, \hat{h}) + \frac{1}{\sigma}\left(c_{t,t} \|\psi_t^{*,n} - \psi_t^*\|_{u,\bar{v}} + \rho_{t,u,v}(\hat{h}) + \sum_{j=t+1}^m c_{t,j} \|\nu_j^*\|_{u,\bar{v}}\right)^2
\end{align}
Repeatedly invoking the fact that for any $a,b \geq 0$, $(a+b)^2 \leq a^2 + b^2$:
\begin{align}
    \frac{\lambda}{2}\|\nu_t\|_{2,2}^2 - \frac{\sigma}{4} \|\nu_t\|_{u,v}^2 \leq~&  \epsilon(\psi_t^{*,n}, \hat{\underline{\psi}}_{t}, \hat{h}) + \frac{4 c_{t,t} }{\sigma} \|\psi_t^* - \psi_{t}^{*,n}\|_{u,\bar{v}} + \frac{4}{\sigma} \rho_{t,u,v}(\hat{h})^2 + \frac{2}{\sigma} \left(\sum_{j=t+1}^m c_{t,j}\, \|\nu_j^*\|_{u,\bar{v}}\right)^2
\end{align}
\end{proof}

\subsection{Proof of Theorem~\ref{thm:hetero-thm}}

First we prove a preliminary auxiliary lemma.

\begin{lemma}\label{lem:inductive-bound}
Consider any sequence of non-negative numbers $a_1, \ldots, a_m$ satisfying the inequality:
\begin{align}
    a_t \leq \mu_t + c_t \max_{j=t+1}^m a_j
\end{align}
with $\mu_t, c_t\geq 0$. Let $c := \max_{t\in [m]} c_t$ and $\mu:=\max_{t\in [m]}\mu_t$. Then it must also hold that:
\begin{align}
    a_t \leq \mu \frac{c^{m-t+1}-1}{c-1}
\end{align}
\end{lemma}
\begin{proof}
By the assumed inequality we have that:
\begin{align}
    a_t \leq \mu_t + c_t \max_{j>t} a_j \leq \mu + c \max_{j>t} a_j
\end{align}
We prove by induction that $a_t\leq \mu \sum_{\tau=0}^{m-t} c^\tau$. By the above inequality, the property holds for $t=m$. Assuming it holds for $j>t$, we have that:
\begin{align}
    a_t \leq~& \mu + c \mu \sum_{\tau=0}^{m-t-1} c^\tau = \mu + \mu \sum_{\tau=0}^{m-t-1} c^{\tau+1} = \mu \sum_{\tau=0}^{m-t} c^{\tau}
\end{align}
which completes the inductive proof. Finally, observe that $\sum_{\tau=0}^{m-t} c^{\tau} = \frac{c^{m-t+1}-1}{c-1}$.
\end{proof}

\begin{proof}[Proof of Theorem~\ref{thm:hetero-thm}]
Let $\loss_{S,t}, \loss_{S',t}$ denote the empirical loss over the samples in set $S$ and $S'$ correspondingly and let $\hat{h}_{S}, \hat{h}_{S'}$ the nuisance functions trained on sample $S$ and $S'$ correspondingly.

Denote the random variable $Z=\{X_0,X_1,T_1,\ldots,X_m,T_m,Y\}$ and consider the loss function:
\begin{align}
    u(Z, \underline{\psi}_{t}(X_0); \hat{h}) := \left(Y - \hat{q}(I_t) - \sum_{j=t}^m \psi_j(X_0)'(Q_{j,t} - \hat{p}_{j,t}(I_t)) - \psi_t(X_0)'(\phi(\bar{X}_t, \bar{T}_t) - \hat{p}_{t,t}(I_t))\right)^2
\end{align}
and the centered loss function:
\begin{align}
    \ell(Z, \underline{\psi}_{t}(X_0); \hat{h}) := u(Z,  \underline{\psi}_{t}(X_0); \hat{h}) - u(Z, \psi_t^{*,n}, \underline{\psi}_{t+1}(X_0); \hat{h})
\end{align}
where $\psi_t^{*,n} := \arginf_{\psi_t \in \Psi_t^n} \|\psi_t - \psi_t^*\|_{2,2}$.
When all random variables are bounded, then $\ell$ is $L$-Lipschitz, for some constant $L$, with respect to $\underline{\psi}_t(X_0)$, i.e.:
\begin{align}
    |\ell(Z; \underline{\psi}_{t}(X_0); \hat{h}) - \ell(Z; \underline{\psi}_{t}'(X_0); \hat{h})| \leq L \|\underline{\psi}_{t}(X_0) - \underline{\psi}_{t}'(X_0)\|_2
\end{align}
Moreover, note that:
\begin{align}
    \loss_{S,t}(\psi_t; \underline{\psi}_{t+1}, \hat{h}) - \loss_{S,t}(\psi_t^{*,n}; \underline{\psi}_{t+1}, \hat{h}) = \frac{1}{|S|} \sum_{i\in S} \ell(Z^i, \underline{\psi}_{t}(X_0^i); \hat{h})
\end{align}
and that $\ell(Z, \underline{\psi}_{t}^*(X_0); \hat{h}) = 0$. For succinctness, for any given $\hat{h}$, let:
\begin{align}
    \loss_{S,t}^*(\underline{\psi}_{t}, \hat{h}) :=~& \loss_{S,t}(\psi_t; \underline{\psi}_{t+1}, \hat{h}) - \loss_{S,t}(\psi_t^{*,n}; \underline{\psi}_{t+1}, \hat{h}) = \frac{1}{|S|} \sum_{i\in S} \ell(Z^i, \underline{\psi}_{t}(X_0^i); \hat{h})\\
    \loss_{D,t}^*(\underline{\psi}_{t}, \hat{h}) :=~& \loss_{D,t}(\psi_t; \underline{\psi}_{t+1}, \hat{h}) - \loss_{D,t}(\psi_t^{*,n}; \underline{\psi}_{t+1}, \hat{h}) = \E\left[\ell(Z; \underline{\psi}_{t}(X_0); \hat{h})\right]\\
\end{align}
Suppose that $\delta_n$ upper bounds the critical radius of the function classes $\{\Psi_{t,i}^n - \psi_{t,i}^{*,n}\}_{t\in [m], i\in [\fdim]}$ and $\delta_n = \Omega\left(\frac{r\log(\log(n))}{n}\right)$ and let $\delta_{n,\zeta}=\delta_n + c_0\sqrt{\frac{\log(c_1\,m/\zeta)}{n}}$, for some appropriately defined universal constants $c_0,c_1$. Moreover, note that $\loss_{S,t}^*({\underline{\psi}}_{t}^{*,n}, \hat{h}_{S'})=\loss_{D,t}^*(\underline{\psi}_{t}^{*,n}, \hat{h}_{S'})=0$. Then by Lemma~11 of \citep{foster2019orthogonal}, w.p. $1-\zeta$, for all $t\in [m]$:
\begin{align}
    \left|\loss_{S,t}^*(\hat{\underline{\psi}}_{t}, \hat{h}_{S'}) - \loss_{D,t}^*(\hat{\underline{\psi}}_{t}, \hat{h}_{S'})\right| =~& \left|\loss_{S,t}^*(\hat{\underline{\psi}}_{t}, \hat{h}_{S'}) - \loss_{S,t}^*({\underline{\psi}}_{t}^{*,n}, \hat{h}_{S'}) - \left(\loss_{D,t}^*(\hat{\underline{\psi}}_{t}, \hat{h}_{S'}) - \loss_{D,t}^*(\underline{\psi}_{t}^{*,n}, \hat{h}_{S'})\right)\right|\\ 
    \leq~& c_0\left(r\,m\,\delta_{n/2,\zeta} \|\hat{\underline{\psi}}_{t} - \underline{\psi}_{t}^{*,n}\|_{2,2} + r\,m\,\delta_{n/2,\zeta}^2\right)
\end{align}
for some universal constant $c_0\geq 1$.
Analogously we can exchange the roles of $S$ and $S'$ and define corresponding losses and properties.

Moreover, we have that by the definition of the final stage of Dynamic DML:
\begin{equation}
    \frac{1}{2} \left(\loss_{S,t}^*(\underline{\psi}_{t}, \hat{h}_{S'}) + \loss_{S',t}^*(\underline{\psi}_{t}, \hat{h}_{S})\right)  \leq 0
\end{equation}
Thus we have that w.p. $1-2\zeta$:
\begin{equation}\label{eqn:bound-on-average}
\frac{1}{2} \sum_{O\in \{S, S'\}} \loss_{D,t}^*(\hat{\underline{\psi}}_{t})    \leq 
c_0\left(r\,m\,\delta_{n/2,\zeta} \|\hat{\underline{\psi}}_{t} - \underline{\psi}_{t}^{*,n}\|_{2,2} + r\,m\,\delta_{n/2,\zeta}^2\right)
\end{equation}
Moreover, note that instantiating Lemma~\ref{lem:general} with nuisance functions $\hat{h}_S$ and $\hat{h}_{S'}$ correspondingly, we have that in the notation of that lemma:
\begin{align}
    \loss_{D,t}^*(\hat{\underline{\psi}}_{t}, \hat{h}_S) := \epsilon(\psi_t^{*,n}, \hat{\underline{\psi}}_t, \hat{h}_S)
\end{align}
Applying Lemma~\ref{lem:general} for each $\hat{h}_S, \hat{h}_{S'}$, with $u=v=2$ and $\sigma=\lambda$ and $\psi_t^{*,n}$ as defined by the theorem statement and averaging the final inequality we get:
\begin{align}\label{eqn:inductive-bound-hetero}
    \frac{\lambda}{4} \|\nu_t\|_{2,2}^2\leq~& \epsilon(\psi_t^{*,n},\hat{\underline{\psi}}_{t}, \hat{h}) + \frac{4c_{t,t}}{\lambda}\textsc{bias}_n^2 + \frac{2}{\lambda}\sum_{O\in \{S,S'\}}\rho_{t,2,2}(\hat{h}_O)^2 + \frac{2}{\lambda} \left(\sum_{j=t+1}^m c_{t,j}\, \|\nu_j^*\|_{2,2}\right)^2
\end{align}
where $\nu_t=\hat{\psi}_t - \psi_{t}^{*,n}$ and $\nu_t^* = \hat{\psi}_t - \psi_t^*$ and $\rho_{t,2,2}$ as defined in Lemma~\ref{lem:general}.
Moreover, we have:
\begin{align}
    \frac{\lambda}{8} \|\hat{\psi}-\psi_t^{*}\|_{2,2}^2 \leq \frac{\lambda}{4} \|\hat{\psi}-\psi_t^{*,n}\|_{2,2}^2 + \frac{\lambda}{4} \textsc{bias}_n^2
\end{align}
and therefore:
\begin{align}
    \frac{\lambda}{8} \|\nu_t^*\|_{2,2}^2\leq~& \epsilon(\psi_t^{*,n},\hat{\underline{\psi}}_{t}, \hat{h}) + \frac{16c_{t,t} + \lambda^2}{4\lambda}\textsc{bias}_n^2 + \frac{2}{\lambda}\sum_{O\in \{S,S'\}}\rho_{t,2,2}(\hat{h}_O)^2 + \frac{2}{\lambda} \left(\sum_{j=t+1}^m c_{t,j}\, \|\nu_j^*\|_{2,2}\right)^2
\end{align}
If we let $c_t=\sum_{j=t+1}^m c_{t,j}$, then:
\begin{align}
    \frac{\lambda}{8} \|\nu_t^*\|_{2,2}^2\leq~& \epsilon(\psi_t^{*,n},\hat{\underline{\psi}}_{t}, \hat{h}) + \frac{16c_{t,t} + \lambda^2}{4\lambda}\textsc{bias}_n^2 + \frac{2}{\lambda}\sum_{O\in \{S,S'\}}\rho_{t,2,2}(\hat{h}_O)^2 + \frac{2c_t^2}{\lambda} \max_{j > t} \|\nu_j^*\|_{2,2}^2
\end{align}
Plugging in the bound from Equation~\eqref{eqn:bound-on-average}, we have:
\begin{align}
\frac{\lambda}{8} \|\nu_t^*\|_{2,2}^2\leq~& c_0\,r\,m\,\delta_{n/2,\zeta} \|\hat{\underline{\psi}}_{t} - \underline{\psi}_{t}^{*,n}\|_{2,2} + c_0\,r\,m\,\delta_{n/2,\zeta}^2 + \frac{16c_{t,t} + \lambda^2}{4\lambda}\textsc{bias}_n^2 \\
~&~  + \frac{2}{\lambda}\sum_{O\in \{S,S'\}}\rho_{t,2,2}(\hat{h}_O)^2 + \frac{2c_t^2}{\lambda} \max_{j > t} \|\nu_j^*\|_{2,2}^2
\end{align}
By an AM-GM inequality, for all $a,b\geq 0$: $a\cdot b \leq \half(\frac{2}{\sigma} a^2 + \frac{\sigma}{2} b^2)$. Thus for any $\epsilon>0$, invoking the latter AM-GM inequality:
\begin{align}
    c_0\,r\,m\,\delta_{n/2,\zeta} \|\hat{\underline{\psi}}_{t} - \underline{\psi}_{t}^{*,n}\|_{2,2} + c_0\,r\,m\,\delta_{n/2,\zeta}^2 \leq~& 2\,c_0^2\frac{4+\epsilon}{\epsilon}r^2 m^2\delta_{n/2,\zeta}^2 +  \frac{\epsilon}{16} \|\hat{\underline{\psi}}_{t} - \underline{\psi}_{t}^{*,n}\|_{2,2}^2\\
    =~& 2\,c_0^2\frac{4+\epsilon}{\epsilon}r^2 m^2\delta_{n/2,\zeta}^2 +  \frac{\epsilon}{16} \sum_{j=t}^m \|\nu_j\|_{2,2}^2\\
    \leq~& 2\,c_0^2\frac{4+\epsilon}{\epsilon}r^2 m^2\delta_{n/2,\zeta}^2 + \frac{\epsilon\,m}{8} \textsc{bias}_n^2 +  \frac{\epsilon}{8} \sum_{j=t}^m \|\nu_j^*\|_{2,2}^2
\end{align}
Thus, we get for some universal constant $c_1$:
\begin{align}
\frac{\lambda-\epsilon}{8} \|\nu_t^*\|_{2,2}^2\leq~& c_1\frac{4+\epsilon}{\epsilon}r^2 m^2\delta_{n/2,\zeta}^2 + \left(\frac{16c_{t,t} + \lambda^2}{4\lambda} + \frac{\epsilon\,m}{8}\right)\textsc{bias}_n^2 \\
~&~  + \frac{2}{\lambda}\sum_{O\in \{S,S'\}}\rho_{t,2,2}(\hat{h}_O)^2 + \left(\frac{2c_t^2}{\lambda} + \frac{\epsilon\,m}{8}\right) \max_{j > t} \|\nu_j^*\|_{2,2}^2
\end{align}

Invoking Lemma~\ref{lem:inductive-bound}, we thus get that:
\begin{align}
    \|\hat{\psi}_t-\psi_t^*\|_{2,2}^2 \leq \mu_{\epsilon} \frac{c_{\epsilon}^{m-t+1} - 1}{c_{\epsilon}-1}
\end{align}
where:
\begin{align}
    c_{\epsilon} :=~& \max_{t\in [m]} \frac{8}{\lambda - \epsilon} \left(\frac{2c_t^2}{\lambda} + \frac{\epsilon\,m}{8}\right)\\
    \mu_{\epsilon} :=~& \frac{8}{\lambda-\epsilon}\left(c_1\frac{4+\epsilon}{\epsilon}r^2 m^2\delta_{n/2,\zeta}^2 + \left(\frac{16c_{t,t} + \lambda^2}{4\lambda} + \frac{\epsilon\,m}{8}\right)\textsc{bias}_n^2 + \max_{t=1}^m \frac{2}{\lambda}\sum_{O\in \{S,S'\}}\rho_{t,2,2}(\hat{h}_O)^2\right)
\end{align}
Let $c :=  \max_{t\in [m]} \frac{16 c_t^2}{\lambda^2}$.
Choosing $\epsilon \leq \frac{c \lambda}{m^2}$ and $\epsilon \leq \frac{\lambda}{2m}$, we have:
\begin{align}
    c_{\epsilon} \leq~& \frac{16 c_t^2}{\lambda^2 (1-1/2m)} + \frac{m}{\lambda (1-1/2m)} \epsilon \leq \frac{1+1/m}{1-1/2m} c \\
    \mu_{\epsilon} \leq~& \frac{16c_1}{\lambda} \left(1 + \frac{4m}{\lambda}\max\left\{\frac{m}{c}, 2\right\}\right) r^2\, m^2\,\delta_{n/2,\zeta}^2 + \frac{32}{\lambda} \max_{\substack{t\in [m]\\ O\in \{S,S'\}}} \rho_{t,2,2}(\hat{h}_O)^2 + \frac{16c_{t,t} + 2\lambda^2}{4\lambda}\textsc{bias}_n^2
\end{align}
Moreover, if we let $\kappa = \frac{(1+1/m)}{(1-1/2m)}$:
\begin{align}
    c_{\epsilon}^m \leq \kappa^m c^m \leq 2e c^m
\end{align}
If $\kappa\,c > 1$, then $\frac{c_{\epsilon}^m-1}{c_{\epsilon}-1}\leq \kappa^m \frac{c^m-1}{c-1}\leq 2e \frac{c^m-1}{c-1}$ and if $\kappa\, c<1$ then $\frac{c_{\epsilon}^m-1}{c_{\epsilon}-1}\leq \frac{1}{1-\kappa c}$.
Thus, letting $c_* = \max\left\{\frac{1}{1-\kappa c}, 2e \frac{c^m-1}{c-1}\right\}$ we have that:
\begin{align}
    \|\hat{\psi}_t-\psi_t^*\|_{2,2}^2 \leq~& \frac{c_*}{\lambda} \left(16c_1 \left(1 + \frac{4m}{\lambda}\max\left\{\frac{m}{c}, 2\right\}\right) r^2\, m^2\,\delta_{n/2,\zeta}^2 + \frac{16c_{t,t} + 2\lambda^2}{4}\textsc{bias}_n^2\right)\\
    ~&~ + \frac{c_*}{\lambda} 32 \max_{\substack{t\in [m]\\ O\in \{S,S'\}}} \rho_{t,2,2}(\hat{h}_O)^2
\end{align}

Taking expectation of this inequality and integrating the tail bound, we have that:
\begin{align}
    \E_{\hat{\psi}}\left[\|\hat{\psi}_t-\psi_t^*\|_{2,2}^2 \right] \leq~& \frac{c_*}{\lambda} \left(16c_1 \left(1 + \frac{4m}{\lambda}\max\left\{\frac{m}{c}, 2\right\}\right) r^2\, m^2\,\delta_{n/2}^2 + \frac{16c_{t,t} + 2\lambda^2}{4}\textsc{bias}_n^2\right)\\
    ~&~ + \frac{c_*}{\lambda} 64\,m \max_{\substack{t\in [m]\\ O\in \{S,S'\}}} \E\left[\rho_{t,2,2}(\hat{h}_O)^2\right]
\end{align}
If $m, \{c_{t,j}\}_{1\leq t\leq j\leq m}, \lambda, M$ are constants and
\begin{align}
    \max_{1\leq t \leq j \leq m} \E_{\hat{p}_{t,t}, \hat{p}_{j,t}}\left[\left\|\|\Delta(\hat{p}_{t,t}, \hat{p}_{j,t})\|_{op}\right\|_{2}^2\right] =~& O(r^2 \delta_{n/2}^2 + \textsc{bias}_n^2)\\
    \max_{1\leq t\leq m} \E_{\hat{p}_{t,t}, \hat{q}_{t}}\left[\left\|\|\Delta(\hat{p}_{t,t}, \hat{q}_t)\|_{op}\right\|_{2}^2\right] =~& O(r^2 \delta_{n/2}^2 + \textsc{bias}_n^2)
\end{align}
then we have that:
\begin{align}
    \max_{t\in [m]}\E\left[\|\hat{\psi}_t-\psi_t^*\|_{2,2}^2\right] \leq O\left(r^2\, \delta_{n/2}^2 + \textsc{bias}_n^2\right)
\end{align}

\end{proof}

\subsection{Proof of Theorem~\ref{thm:uniform}}

\begin{proof}
Following identical reasoning as in the proof of Theorem~\ref{thm:hetero-thm}, and invoking Lemma~\ref{lem:general} with $u=2$ and $v=\infty$, and $\psi_t^{*,n}=\arg\inf_{\psi_t\in \Psi_t^n}\|\psi_t - \psi_t^*\|_{2,\infty}$, we can arrive at a analogous to that of Equation~\eqref{eqn:inductive-bound-hetero}:
\begin{align}
    \frac{\lambda}{4} \|\nu_t\|_{2,2}^2 - \frac{\sigma}{4} \|\nu_t\|_{2,\infty}^2\leq~& \epsilon(\psi_t^{*,n},\hat{\underline{\psi}}_{t}, \hat{h}) + \frac{4c_{t,t}}{\sigma}\textsc{bias}_{n,\infty}^2 + \frac{2}{\sigma}\sum_{O\in \{S,S'\}}\rho_{t,2,\infty}(\hat{h}_O)^2 + \frac{2}{\sigma} \left(\sum_{j=t+1}^m c_{t,j}\, \|\nu_j^*\|_{2,\infty}\right)^2\\
    \leq~& \epsilon(\psi_t^{*,n},\hat{\underline{\psi}}_{t}, \hat{h}) + \frac{4c_{t,t}}{\sigma}\textsc{bias}_{n,\infty}^2 + \frac{2}{\sigma}\sum_{O\in \{S,S'\}}\rho_{t,2,\infty}(\hat{h}_O)^2 + \frac{2c_t^2}{\sigma}\max_{j>t} \|\nu_j^*\|_{2,\infty}^2
\end{align}
where for some universal constant $c_0$:
\begin{align}
    \epsilon(\psi_t^{*,n},\hat{\underline{\psi}}_{t}, \hat{h})  :=~& c_0\,r\,m\,\delta_{n/2,\zeta} \|\hat{\underline{\psi}}_{t} - \underline{\psi}_{t}^{*,n}\|_{2,2} + c_0\,r\,m\,\delta_{n/2,\zeta}^2
\end{align}
with $\delta_{n,\zeta} = O\left(M\sqrt{\frac{d_n}{n}} + \sqrt{\frac{\log\log(n)}{n}} + \sqrt{\frac{\log(1/\zeta)}{n}}\right)$. The latter follows since the critical radius of linear predictors in $d_n$-dimensions with a uniform bound of $M$, is at most $M\sqrt{\frac{d_n}{n}}$ (see example 13.8 of \cite{wainwright_2019}).

Note that $\hat{\psi}(\cdot) = \hat{\theta}'a_n(\cdot)$ and $\psi_t^{*,n}(\cdot) = (\theta_t^{*,n})' a_n(\cdot)$. Thus:
\begin{align}
    \|\hat{\psi}-\psi_t^{*,n}\|_{2,2}^2 =~& \E\left[(\hat{\theta} - \theta_t^{*,n})'a_n(X_0)\, a_n(X_0)' (\hat{\theta} - \theta_t^{*,n})\right]\\
    \geq~& \gamma_n \|\hat{\theta} - \theta_t^{*,n}\|_{2}^2\\
    \geq~& \frac{\gamma_n}{M^2} \sup_{x_0\in \mcX_0} \left((\hat{\theta} - \theta_t^{*,n})'a_n(x_0)\right)^2 = \frac{\gamma_n}{M^2} \|\hat{\psi}-\psi_t^{*,n}\|_{2,\infty}^2
\end{align}

Thus setting $\sigma=\lambda_*:=\frac{\lambda \gamma_n}{M^2}$, we thus have:
\begin{align}
    \frac{\lambda_*}{4} \|\nu_t\|_{2,\infty}^2
    \leq~& \epsilon(\psi_t^{*,n},\hat{\underline{\psi}}_{t}, \hat{h}) + \frac{4c_{t,t}}{\lambda_*}\textsc{bias}_{n,\infty}^2 + \frac{2}{\lambda_*}\sum_{O\in \{S,S'\}}\rho_{t,2,\infty}(\hat{h}_O)^2 + \frac{2c_t^2}{\lambda_*}\max_{j>t} \|\nu_j^*\|_{2,\infty}^2
\end{align}
Moreover:
\begin{align}
    c_0\,r\,m\,\delta_{n/2,\zeta} \|\hat{\underline{\psi}}_{t} - \underline{\psi}_{t}^{*,n}\|_{2,2} + c_0\,r\,m\,\delta_{n/2,\zeta}^2 \leq~& 2\,c_0^2\frac{4+\epsilon}{\epsilon}r^2 m^2\delta_{n/2,\zeta}^2 +  \frac{\epsilon}{16} \|\hat{\underline{\psi}}_{t} - \underline{\psi}_{t}^{*,n}\|_{2,2}^2\\
    =~& 2\,c_0^2\frac{4+\epsilon}{\epsilon}r^2 m^2\delta_{n/2,\zeta}^2 +  \frac{\epsilon}{16} \sum_{j=t}^m \|\nu_j\|_{2,2}^2\\
    \leq~& 2\,c_0^2\frac{4+\epsilon}{\epsilon}r^2 m^2\delta_{n/2,\zeta}^2 + \frac{\epsilon\,m}{8} \textsc{bias}_n^2 +  \frac{\epsilon}{8} \sum_{j=t}^m \|\nu_j^*\|_{2,2}^2\\
    \leq~& 2\,c_0^2\frac{4+\epsilon}{\epsilon}r^2 m^2\delta_{n/2,\zeta}^2 + \frac{\epsilon\,m}{8} \textsc{bias}_{n,\infty}^2 +  \frac{\epsilon}{8} \sum_{j=t}^m \|\nu_j^*\|_{2,\infty}^2
\end{align}
Moreover, we have:
\begin{align}
    \frac{\lambda_*}{8} \|\nu_t^*\|_{2,\infty}^2 \leq \frac{\lambda_*}{4} \|\nu_t\|_{2,\infty}^2 + \frac{\lambda_*}{4} \textsc{bias}_{n,\infty}^2
\end{align}
Thus, we get for some universal constant $c_1$:
\begin{align}
\frac{\lambda_*-\epsilon}{8} \|\nu_t^*\|_{2,\infty}^2\leq~& c_1\frac{4+\epsilon}{\epsilon}r^2 m^2\delta_{n/2,\zeta}^2 + \left(\frac{16c_{t,t} + \lambda_*^2}{4\lambda_*} + \frac{\epsilon\,m}{8}\right)\textsc{bias}_{n,\infty}^2 \\
~&~  + \frac{2}{\lambda_*}\sum_{O\in \{S,S'\}}\rho_{t,2,\infty}(\hat{h}_O)^2 + \left(\frac{2c_t^2}{\lambda_*} + \frac{\epsilon\,m}{8}\right) \max_{j > t} \|\nu_j^*\|_{2,\infty}^2
\end{align}
Let $c :=  \max_{t\in [m]} \frac{16 c_t^2}{\lambda_*^2}$.
Analogous to the proof of Theorem~\ref{thm:hetero-thm}, choosing $\epsilon \leq \frac{c \lambda_*}{m^2}$ and $\epsilon \leq \frac{\lambda_*}{2m}$, and letting $\kappa = \frac{(1+1/m)}{(1-1/2m)}$ and $c_* = \max\left\{\frac{1}{1-\kappa c}, 2e \frac{c^m-1}{c-1}\right\}$ we have that:
\begin{align}
    \|\hat{\psi}_t-\psi_t^*\|_{2,\infty}^2 \leq~& \frac{c_*}{\lambda_*} \left(16c_1 \left(1 + \frac{4m}{\lambda_*}\max\left\{\frac{m}{c}, 2\right\}\right) r^2\, m^2\,\delta_{n/2,\zeta}^2 + \frac{16c_{t,t} + 2\lambda_*^2}{4}\textsc{bias}_{n,\infty}^2\right)\\
    ~&~ + \frac{c_*}{\lambda_*} 32 \max_{\substack{t\in [m]\\ O\in \{S,S'\}}} \rho_{t,2,\infty}(\hat{h}_O)^2
\end{align}

Taking expectation of this inequality and integrating the tail bound, we have that:
\begin{align}
    \E_{\hat{\psi}}\left[\|\hat{\psi}_t-\psi_t^*\|_{2,\infty}^2 \right] \leq~& \frac{c_*}{\lambda_*} \left(16c_1 \left(1 + \frac{4m}{\lambda_*}\max\left\{\frac{m}{c}, 2\right\}\right) r^2\, m^2\,\delta_{n/2}^2 + \frac{16c_{t,t} + 2\lambda_*^2}{4}\textsc{bias}_{n,\infty}^2\right)\\
    ~&~ + \frac{c_*}{\lambda_*} 64\,m \max_{\substack{t\in [m]\\ O\in \{S,S'\}}} \E\left[\rho_{t,2,\infty}(\hat{h}_O)^2\right]
\end{align}
Moreover, note that:
\begin{align}
    \rho_{t,2,\infty}(\hat{h}) =~& \left\|\|\Delta(\hat{p}_{t}, \hat{q}_t)\|_{op}\right\|_{1}
    + 3 M \sum_{j=t}^m  \left\|\|\Delta(\hat{p}_{t}, \hat{p}_{j,t})\|_{op}\right\|_{1}
\end{align}

Assume that $m, \{c_{t,j}\}_{1\leq t\leq j\leq m}, \lambda, M$ are constants. If $\gamma_n\to 0$, then we have that $c = O(\gamma_n^{-2})$, $c_*=O(\gamma_n^{-2(m-1)})$ and $\lambda_* = O(\gamma_n)$ and $c \lambda_*=O(\gamma_n^{-1})$ and hence:
\begin{align}
    \E_{\hat{\psi}}\left[\|\hat{\psi}_t-\psi_t^*\|_{2,\infty}^2 \right] \leq~& O\left(\gamma_n^{-2m + 1} \left(\gamma_n^{-1} r^2\, \delta_{n/2}^2 + \textsc{bias}_{n,\infty}^2 + \max_{\substack{t\in [m]\\ O\in \{S,S'\}}} \E\left[\rho_{t,2,\infty}(\hat{h}_O)^2\right] \right)\right)
\end{align}
Thus if we have that:
\begin{align}
    \max_{1\leq t \leq j \leq m} \left\{\E_{\hat{p}_{t,t}, \hat{p}_{j,t}}\left[\left\|\|\Delta(\hat{p}_{t,t}, \hat{p}_{j,t})\|_{op}\right\|_{1}^2\right], \E_{\hat{p}_{t,t}, \hat{q}_{t}}\left[\left\|\|\Delta(\hat{p}_{t,t}, \hat{q}_t)\|_{op}\right\|_{1}^2\right]\right\} =~& O\left(\gamma_n^{-1} r^2\,\delta_{n/2}^2 + \textsc{bias}_{n,\infty}^2\right)
\end{align}
Then:
\begin{align}
    \E_{\hat{\psi}}\left[\|\hat{\psi}_t-\psi_t^*\|_{2,\infty}^2 \right] \leq~& O\left(\gamma_n^{-2m + 1} \left(\gamma_n^{-1} r^2\, \delta_{n/2}^2 + \textsc{bias}_{n,\infty}^2\right)\right)
\end{align}
If $\gamma_n$ remains bounded away from zero, then the above bound conditions to hold since the remainder of the quantities $m, \{c_{t,j}\}_{1\leq t\leq j\leq m}, \lambda, M$ are constants.
\end{proof}

\subsection{Proof of Theorem~\ref{thm:sparse-linear}}

\begin{lemma}[Restricted Cone Property]\label{lem:cone}
Assume that all random variables and functions are bounded. Suppose we set $\kappa_t$ such that it satisfies:
\begin{align}
    \frac{\kappa_t}{2} \geq \delta_{n,\zeta} + \sum_{j=t+1}^m (c_{t,j} + \delta_{n,\zeta}) \|\nu_j\|_1 + \epsilon_n
\end{align}
where:
\begin{align}
    \delta_{n,\zeta} :=~& c_0 \sqrt{\frac{\log(c_1\, m\, r/\zeta)}{n}}\\
    \epsilon_n :=~& \max_{m\geq j\geq t\geq 1, i,i'\in [r]} \max\left\{\|\hat{p}_{t,i} - p_{t,i}^*\|_2\, \|\hat{q}_{t,i'} - q_{t,i'}^*\|_2, 4M\,m \|\hat{p}_{t,i} - p_{t,i}^*\|_2\, \|\hat{p}_{j,t,i'} - p_{j,t,i'}^*\|_2\right\}
\end{align}
where $c_0, c_1$ are universal constants and $c_{t,j} = \|\E[\Cov(Q_{t,t}, Q_{j,t}\mid \bar{X}_t, \bar{T}_{t-1})]\|_{\infty}$. Let $\nu_t:=\hat{\psi}_t-\psi_t^*$, let $T$ be the support of $\psi_t^*$ and $T^c$ its complement. Let $\nu_t^{(T)}$ denote the sub-vector support on $T$ and $\mcC(T;3):=\{\nu_t\in \R^r: \|\nu_t^{(T^c)}\|_1\leq 3\|\nu_t^{(T)}\|_1\}$. Then w.p. $1-\zeta$: $\nu_t \in \mcC(T;3)$.
\end{lemma}

\begin{proof}[Proof of Lemma~\ref{lem:cone}]
Let $\loss_{n,t}(\psi_t; \underline{\psi}_{t+1}, \hat{h})$ denote the un-penalized empirical loss using the estimated first stage nuisance functions, where by $\hat{h}=(\hat{h}_S,\hat{h}_{S'})$ we denote the union of both nuisance functions estimated on each side of the split. Moreover, let $\bar{\loss}_{D,t}(\psi_t; \underline{\psi}_{t+1}, \hat{h})=\frac{1}{2}\sum_{O\in \{S,S\}}\loss_{D,t}(\psi_t; \underline{\psi}_{t+1}, \hat{h}_O)$. By optimality of $\hat{\psi}_t$ for the penalized empirical square loss:
\begin{align}
    \loss_{n,t}(\hat{\psi}_t; \hat{\underline{\psi}}_{t+1}, \hat{h}) - \loss_{n,t}(\psi_t^*; \hat{\underline{\psi}}_{t+1}, \hat{h}) \leq \kappa_t \left(\|\psi_t^*\|_1 - \|\hat{\psi}_t\|_1\right)
\end{align}
Let $\nu_t = \hat{\psi}_t - \psi_t^*$ and $T$ denote the support of $\psi_t^*$ (i.e. the set of non-zero coordinates). Denote with $\nu^{(T)}$ the sub-vector on the set $T$ and $T^{c}$ the complement of $T$. By additivity of the $\ell_1$ norm and the triangle inequality:
\begin{align}
    \|\hat{\psi}_t\|_1 =~& \|\psi_t^* + \nu_t^{(T)}\|_1 + \|\nu_t^{(T^c)}\|_1 \geq \|\psi_t^*\| - \|\nu_t^{(T)}\|_1 +\|\nu_t^{(T^c)}\|_1\\
    \|\psi_t^*\|_1 - \|\hat{\psi}_t\|_1 \leq~& \|\nu_t^{(T)}\|_1 - \|\nu_t^{(T^c)}\|_1
\end{align}
Moreover, by convexity of the empirical square loss, we have:
\begin{align}
    \loss_{n,t}(\hat{\psi}_t; \hat{\underline{\psi}}_{t+1}, \hat{h}) - \loss_{n,t}(\psi_t^*; \hat{\underline{\psi}}_{t+1}, \hat{h}) \geq~& D_{\psi_t}\loss_{n,t}(\psi_t^*; \hat{\underline{\psi}}_{t+1}, \hat{h})[\nu_t]\\
    =~& D_{\psi_t}\left(\loss_{n,t}(\psi_t^*; \hat{\underline{\psi}}_{t+1}, \hat{h}) - \bar{\loss}_{D,t}(\psi_t^*; \hat{\underline{\psi}}_{t+1}, \hat{h})\right)[\nu_t]\\
    ~&~ + D_{\psi_t}\loss_{D,t}(\psi_t^*; \hat{\underline{\psi}}_{t+1}, \hat{h})[\nu_t]\\
    \geq~& \underbrace{D_{\psi_t}\left(\loss_{n,t}(\psi_t^*; \hat{\underline{\psi}}_{t+1}, \hat{h}) - \bar{\loss}_{D,t}(\psi_t^*; \hat{\underline{\psi}}_{t+1}, \hat{h})\right)[\nu_t]}_{F}\\
    ~&~ + \underbrace{D_{\psi_t}\left(\bar{\loss}_{D,t}(\psi_t^*; \hat{\underline{\psi}}_{t+1}, \hat{h}) - \loss_{D,t}(\psi_t^*; \underline{\psi}_{t+1}^*, h^*)\right)[\nu_t]}_{G}
\end{align}
where in the last step we used the optimality of $\psi_t^*$ for the oracle population loss. 

Observe that the second quantity $G$ is equal to the terms $A+B$ from the proof of Lemma~\ref{lem:general}. Thus applying the same arguments as in that proof, with $u=1$ and $v=\infty$ and with $X_0$ the empty set, we have that:
\begin{align}
    |G| \leq~& \|\nu_t\|_1\, \left(\sum_{j=t+1}^m c_{t,j} \|\nu_j\|_{1}
   +  \| \Delta(\hat{p}_{t}, \hat{q}_t) \|_{\infty} + 4M\sum_{j=t}^m \left\| \Delta(\hat{p}_t, \hat{p}_{j,t})\right\|_{\infty}\right)
\end{align}
where we used the short-hand norm notation $\|A\|_{\infty}=\max_{i,j} |A_{i,j}| = \|A\|_{\infty,1}$. Note that for any two $r$-dimensional vector valued nuisance functions and for $X_0$ the empty set:
\begin{align}
    \|\Delta(\hat{f},\hat{g})\|_{\infty} =~& \left\|\E[(\hat{f}(Z_f) - f^*(Z_f))\, (\hat{g}(Z_g) - g^*(Z_g))']\right\|_{\infty}\\
    \leq~& \max_{i, i' \in [r]} \left|\E[(\hat{f}_i(Z_f) - f_i^*(Z_f))\, (\hat{g}_{i'}(Z_g) - g_{i'}^*(Z_g))']\right|\leq \max_{i, i' \in [r]} \|\hat{f}_i - f_i^*\|_2\, \|\hat{g}_{i'}- g_{i'}^*\|_2
\end{align}
Thus we can bound $G$ as:
\begin{align}
   |G| \leq~& \|\nu_t\|_1\, \left(\sum_{j=t+1}^m c_{t,j} \|\nu_j\|_{1}
   +  \underbrace{\max_{j\geq t, i,i'\in [r]} \left\{\|\hat{p}_{t,i} - p_{t,i}^*\|_2\, \|\hat{q}_{t,i'} - q_{t,i'}^*\|_2, 4M\,m \|\hat{p}_{t,i} - p_{t,i}^*\|_2\, \|\hat{p}_{j,t,i'} - p_{j,t,i'}^*\|_2\right\}}_{\epsilon_{n,t}}\right)
\end{align}
where $c_{t,j} := \|\E[C_{t,j}]\|_{\infty}$ and $M := \max_{t\in [m], \psi_t\in \Psi_t}\|\psi_t\|_1$. Moreover, note that we can write:
\begin{align}
    |F| \leq~& |D_{\psi_t}\left(\loss_{n,t}(\psi_t^*; \underline{\psi}_{t+1}^*, \hat{h}) - \bar{\loss}_{D,t}(\psi_t^*; \underline{\psi}_{t+1}^*, \hat{h})\right)[\nu_t]|\\
    ~&~ + |D_{\psi_t}\left(\loss_{n,t}(\psi_t^*; \hat{\underline{\psi}}_{t+1}, \hat{h}) - \loss_{n,t}(\psi_t^*; \underline{\psi}_{t+1}^*, \hat{h}) - (\bar{\loss}_{D,t}(\psi_t^*; \hat{\underline{\psi}}_{t+1}, \hat{h}) - \bar{\loss}_{D,t}(\psi_t^*; \underline{\psi}_{t+1}^*, \hat{h}))\right)[\nu_t]|\\
    \leq~& \|\grad_{\psi_t}\left(\loss_{n,t}(\psi_t^*; \underline{\psi}_{t+1}^*, \hat{h}) - \bar{\loss}_{D,t}(\psi_t^*; \underline{\psi}_{t+1}^*, \hat{h})\right)\|_{\infty} \|\nu_t\|_1\\
    ~&~ + \|\grad_{\psi_t}\left(\loss_{n,t}(\psi_t^*; \hat{\underline{\psi}}_{t+1}, \hat{h}) - \loss_{n,t}(\psi_t^*; \underline{\psi}_{t+1}^*, \hat{h}) - (\bar{\loss}_{D,t}(\psi_t^*; \hat{\underline{\psi}}_{t+1}, \hat{h}) - \bar{\loss}_{D,t}(\psi_t^*; \underline{\psi}_{t+1}^*, \hat{h}))\right)\|_{\infty} \|\nu_t\|_1
\end{align}

The first term can be written as the maximum among $2\, r$ centered empirical processes of $n$ bounded i.i.d. random variables. Thus we have that for some universal constants $c_0,c_1$, w.p. $1-\zeta$:
\begin{align}
    \|\grad_{\psi_t}\left(\loss_{n,t}(\psi_t^*; \underline{\psi}_{t+1}^*, \hat{h}) - \bar{\loss}_{D,t}(\psi_t^*; \underline{\psi}_{t+1}^*, \hat{h})\right)\|_{\infty} \leq c_0 \sqrt{\frac{\log(c_1\,r/\zeta)}{n}}
\end{align}
For the second term note that for each half-split, we have that:
\begin{align}
    \grad_{\psi_t}\left(\loss_{n,t}(\psi_t^*; \hat{\underline{\psi}}_{t+1}, \hat{h}) - \loss_{n,t}(\psi_t^*; \underline{\psi}_{t+1}^*, \hat{h})\right) := \sum_{j=t+1}^{m} \E_n\left[\tilde{T}_{t,t}\, \tilde{T}_{j,t}'\right] \nu_j\\
    \grad_{\psi_t}\left(\bar{\loss}_{D,t}(\psi_t^*; \hat{\underline{\psi}}_{t+1}, \hat{h}) - \bar{\loss}_{D,t}(\psi_t^*; \underline{\psi}_{t+1}^*, \hat{h})\right) := \sum_{j=t+1}^{m} \E\left[\tilde{T}_{t,t}\, \tilde{T}_{j,t}'\right] \nu_j
\end{align}
Thus we can upper bound that term by:
\begin{align}
    \sum_{j=t+1}^m \left\| \E_n\left[\tilde{T}_{i,t}\, \tilde{T}_{j,t}'\right] -  \E\left[\tilde{T}_{i,t}\, \tilde{T}_{j,t}'\right] \right\|_{\infty} \|\nu_j\|_{1}
\end{align}
Thus we need to control the maximum of $2\, m\, r^2$ centered empirical processes. Hence, we have that this term is upper bounded by $c_2\sqrt{\frac{\log(c_3\, m\, r/\zeta)}{n}}$, for some universal constants $c_2,c_3$. Thus for some universal constants $c_0, c_1$, we have that, w.p. $1-\zeta$:
\begin{align}
    |F| \leq \|\nu_t\|_1 \underbrace{c_0 \sqrt{\frac{\log(c_1\, m\, r/\zeta)}{n}}}_{\delta_{n,\zeta}} \left(1 + \sum_{j=t+1}^m \|\nu_j\|_1\right)
\end{align}
Thus we have that:
\begin{align}
    F + G \geq - \|\nu_t\|_1 \left( \delta_{n,\zeta} + \sum_{j=t+1}^m (c_{t,j}+\delta_{n,\zeta})\|\nu_j\|_1 + \epsilon_{n,t}\right)
\end{align}
Assuming that $\frac{\kappa_t}{2} \geq  \delta_{n,\zeta} + \sum_{j=t+1}^m (c_{t,j}+\delta_{n,\zeta})\|\nu_j\|_1 + \epsilon_{n,t}$, we thus have that:
\begin{align}
    \kappa_t \left(\|\nu_t^{(T)}\|_1 - \|\nu_t^{(T^c)}\|_1\right) \geq F + G \geq -\frac{\kappa_t}{2} \|\nu_t\|_1 \implies \|\nu_t^{(T^c)}\|_1\leq 3\|\nu_t^{(T)}\|_1
\end{align}
Hence $\nu_t\in \mcC(T;3)$.
\end{proof}

\begin{proof}[Proof of Theorem~\ref{thm:sparse-linear}]
First note that for any $\underline{\psi}_t$ and for any $\nu_t$:
\begin{align}
    \nu_t'\nabla_{\psi_t\psi_t} \loss_{n,t}(\psi_t; \underline{\psi}_{t+1}, \hat{h})\nu_t =~& \nu_t'\E_n[\tilde{T}_{t,t}\, \tilde{T}_{t,t}']\nu_t\\
    =~& \nu_t'\E[\tilde{T}_{t,t}\, \tilde{T}_{t,t}']\nu_t + \nu_t'\left(\E[\tilde{T}_{t,t}\, \tilde{T}_{t,t}'] - \E_n[\tilde{T}_{t,t}\, \tilde{T}_{t,t}']\right)\nu_t\\
    \geq~& \nu_t'\E[\Cov(Q_{t,t},Q_{t,t}\mid \bar{X}_t, \bar{T}_{t-1})]\nu_t - \left\|\E[\tilde{T}_{t,t}\, \tilde{T}_{t,t}'] - \E_n[\tilde{T}_{t,t}\, \tilde{T}_{t,t}']\right\|_{\infty} \|\nu_t\|_1^2 \tag{by Lemma~\ref{lem:cov-property}}\\
    \geq~& \lambda \|\nu_t\|_2^2 - \left\|\E[\tilde{T}_{t,t}\, \tilde{T}_{t,t}'] - \E_n[\tilde{T}_{t,t}\, \tilde{T}_{t,t}']\right\|_{\infty} \|\nu_t\|_1^2
\end{align}
With probability $1-\zeta$, for all $t$, the latter is at least:
\begin{align}
    \lambda \|\nu_t\|_2^2 - \delta_{n,\zeta}\|\nu_t\|_1^2
\end{align}
Since $\kappa_t$ is chosen such that $\nu_t\in \mcC(T;3)$, we have that:
$$\|\nu_t\|_1 \leq 4\|\nu_t^{(T)}\|_1 \leq 4\sqrt{s}\|\nu_t^{(T)}\|_2 \leq 4\sqrt{s}\|\nu_t\|_2.$$
Thus we can lower bound the latter by:
\begin{align}
    \underbrace{(\lambda - 16 s \delta_{n,\zeta})}_{\lambda_n} \|\nu_t\|_2^2
\end{align}
Thus the loss function $\loss_{n,t}(\psi_t; \underline{\psi}_{t+1}, \hat{h})$ is $\lambda_n$-strongly convex. Combining with the optimality of $\hat{\psi}_t$ and the lower bound on $\kappa_t$, we have:
\begin{align}
    \kappa_t \left(\|\nu_t^{(T)}\|_1 - \|\nu_t^{(T^c)}\|_1\right)\geq~& \loss_{n,t}(\hat{\psi}_t; \hat{\underline{\psi}}_{t+1}, \hat{h}) - \loss_{n,t}(\psi_t^*; \hat{\underline{\psi}}_{t+1}, \hat{h})\\
    \geq~& D_{\psi_t}\loss_{n,t}(\psi_t^*; \hat{\underline{\psi}}_{t+1}, \hat{h})[\nu_t] + \frac{1}{2} D_{\psi_t,\psi_t}\loss_{n,t}(\psi_t^*; \hat{\underline{\psi}}_{t+1}, \hat{h})[\nu_t,\nu_t]\\
    \geq~& - \frac{\kappa_t}{2} \|\nu\|_1 + \frac{\lambda_n}{2} \|\nu_t\|_2^2
\end{align}
Re-arranging yields:
\begin{align}
    \frac{\lambda_n}{2} \|\nu_t\|_2^2 \leq \frac{3\kappa_t}{2} \|\nu_t^{(T)}\|_1 \leq \frac{3\kappa_t}{2}  \sqrt{s} \|\nu_t\|_2 
\end{align}
Thus we conclude that w.p. $1-\zeta$: $\|\nu_t\|_2\leq \kappa_t\frac{3\sqrt{s}}{\lambda_n}$ and $\|\nu_t\|_1 \leq \kappa_t \frac{12 s}{\lambda_n}$.

We also note that we need to chose $\kappa_t$ such that:
\begin{align}
    \kappa_t \geq  2(\delta_{n,\zeta} + \epsilon_{n,t}) + \sum_{j=t+1}^m 2 (c_{t,j}+\delta_{n,\zeta})\|\nu_j\|_1
\end{align}
If we let $\mu_n = \max_{t\in [m]} 2(\delta_{n,\zeta} + \epsilon_{n,t})$ and $c_n = \max_{t\in [m]} \sum_{j=t+1}^m 2(c_{t,j} + \delta_{n,\zeta})$, then it suffices to choose:
\begin{align}
    \kappa_t \geq \mu_n + c_n \max_{j > t} \|\nu_j\|_1
\end{align}
Finally, let $\beta_n = \frac{12\,s}{\lambda_n}$. Then under the above condition we have that $\|\nu_t\|_1 \leq \beta_n \kappa_t$.

We will show that for some sequence of $\kappa_1,\ldots,\kappa_m$, we have: $\|\nu_j\|_1 \leq \beta_n \mu_n \sum_{j'>j} \beta_n^{m-j} c_n^{m-j}$. This holds for $j=m$ if we choose $\kappa_m := \mu_n$. If we assume this property holds for any $j>t$, then we have that for time $t$, if we choose:
\begin{align}
    \kappa_t := \mu_n + c_n \beta_n \mu_n \sum_{j>t+1} \beta_n^{m-j} c_n^{m-j} = \mu_n + \mu_n \sum_{j>t+1} \beta_n^{m-j+1} c_n^{m-j+1} = \mu_n \sum_{j>t} \beta_n^{m-j} c_n^{m-j}
\end{align}
then it satisfies the desired property that $\kappa_t \geq \mu_n + c_n \max_{j > t} \|\nu_j\|_1$ and therefore also:
\begin{align}
    \|\nu_t\|_1 \leq \beta_n \mu_n \sum_{j > t} \beta_n^{m-j} c_n^{m-j}
\end{align}
completing the inductive step. Thus we conclude that for all $t\in [m]$:
\begin{align}
    \|\nu_t\|_2 \leq \|\nu_t\|_1 \leq \beta_n \mu_n \sum_{j > t} \beta_n^{m-j} c_n^{m-j} = \beta_n \mu_n \frac{(\beta_n c_n)^{m-t + 1} - 1}{\beta_n c_n - 1}
\end{align}
\end{proof}

\end{document}